%% file: Main.tex
\documentclass[12pt,english,a4paper]{article}
\usepackage[T1]{fontenc}
\usepackage[utf8]{luainputenc}
\usepackage{times}
\usepackage{setspace}
\usepackage{amsmath,amsthm,mathtools}
\usepackage{graphicx}
\usepackage{algorithmic}
\usepackage[margin=1in]{geometry}       
\usepackage{amssymb}
\usepackage{enumitem}
\usepackage{subcaption}
\usepackage{url}  
\usepackage{float}
\usepackage{placeins}
\usepackage{multicol}

\theoremstyle{plain}

\newtheorem{remark}{Remark}

\newtheorem{theorem}{Theorem}
\newtheorem{proposition}{Proposition}
\newtheorem{corollary}{Corollary}
\newtheorem{lemma}{Lemma}
\newtheorem{assumption}{Assumption}
\newtheorem{algorithm}{Algorithm}
\newtheorem{claim}{Claim}

\usepackage[sort]{natbib}     
\usepackage{hyperref}
\usepackage{etoolbox}
\hypersetup{
	colorlinks = true, linkcolor = blue,
	citecolor = blue, runcolor = blue, urlcolor = blue }
\bibpunct{(}{)}{;}{a}{,}{,}             
\usepackage[nameinlink,noabbrev]{cleveref}
\usepackage{dsfont}
\usepackage{color}
\usepackage{sgame, tikz}
\usepackage{istgame}

\def\A{\mathcal{A}}
\def\S{\mathcal{S}}

\def\R{\mathbb{R}}

\def\P{\mathbb{P}}
\def\E{\mathbb{E}}

\def\V{\boldsymbol{V}}
\def\Vone{{\V}_{\!\!1}}
\def\Vonet{\tilde{V}_{1}^i}
\def\Vzerot{\tilde{V}_{0}^i}
\def\ssigma{\boldsymbol{\sigma}}
\def\SSigma{\boldsymbol{\Sigma}}
\def\ttau{\boldsymbol{\tau}}
\def\iiota{\boldsymbol{\iota}}

\def\kkappa{\boldsymbol{\kappa}}
\def\Q{\boldsymbol{Q}}

\def\1{\boldsymbol{1}}

\def\u{\boldsymbol{u}}
\def\v{\boldsymbol{v}}

\def\p{\boldsymbol{p}}

\def\s{\boldsymbol{s}}

\def\bA{\boldsymbol{A}}
\def\w{\boldsymbol{w}}

\newcommand{\pibar}{\bar{\pi}}
\newcommand{\dbar}{\bar{\delta}}
\newcommand{\wstar}{\boldsymbol{w}^{*}}
\newcommand{\norminf}[1]{\lVert #1 \rVert_{\infty}}
\newcommand{\Prob}{\mathbb{P}}

\def\argmax{\textnormal{argmax}}
\newcommand*{\Scale}[2][4]{\scalebox{#1}{$#2$}}

\title{\LARGE{\textbf{Learning to Charge More: How Q-Learning Firms Lock In Supracompetitive Prices\thanks{This work was partially supported by 
NSF award DMS 2427955. An earlier version of this paper circulated under the title ``Learning to Charge More: A Theoretical Study of Collusion by Q-Learning Agents''.}}}}\author{Cristian Chica, Yinglong Guo, and Gilad Lerman\thanks{School of Mathematics, University of Minnesota. Email: chica013@umn.edu, guo00413@umn.edu, lerman@umn.edu.}}

\begin{document}

\maketitle
\begin{abstract}

There is growing experimental evidence that $Q$-learning agents can learn to charge supracompetitive prices. We ask when such behavior persists, whether it emerges, and which strategies support it. In a repeated pricing game in which firms condition on the previous period's price profile, we give conditions on the learned action values under which supracompetitive prices become an absorbing state once experimentation ends. We then show that under softmax exploration these conditions are reached with positive and explicitly bounded probability, both from a random and from an all-zero initial $Q$-table. In the latter case firms lock in supracompetitive prices because they never learn that undercutting pays; the same mechanism can lock in the competitive price, so the long-run outcome is selected by the exploration path. Finally, we give sufficient conditions on the learned table under which the induced strategy is naive collusion, a grim trigger strategy or an increasing-price strategy; grim trigger strategies are subgame perfect under a standard discounting condition, whereas naive collusion is generally not an equilibrium.



\end{abstract}

\setcounter{footnote}{0} 
\section{Introduction}
Collusion by algorithmically driven firms has become a central topic in recent discussions of competition policy. Since the influential study by \cite{calvano2020artificial}, a growing body of work has examined whether reinforcement learning algorithms can lead firms to learn collusive outcomes. Although these studies span diverse economic settings and algorithmic designs, most rely on numerical simulations. As a result, several key theoretical questions remain unanswered:

\begin{enumerate}
    \item Under what conditions do $Q$-learning firms keep charging supracompetitive prices once their experimentation ends?
    \item Does the learning process itself reach these conditions, and with what probability?
    \item Are the resulting supracompetitive prices supported by policies that incorporate punishment and reward?
    \item Does the learned behavior constitute a Nash equilibrium?
\end{enumerate}

In this paper, we provide formal answers to these questions. We introduce a framework based on stochastic games with bounded memory and analyze their subgame perfect equilibria (SPEs). We then formulate a version of $Q$-learning with bounded experimentation and study the emergence and persistence of supracompetitive pricing behavior in an infinite repeated games setting.

Our model features $n$ firms competing over an infinite time horizon. In each period, firms choose prices based on a simple form of one-memory policies (i.e., strategies): these policies depend only on the current state of the environment and the prices chosen in the previous period. Firms may use one policy in the initial period ($t = 0$), and a distinct, time-invariant policy from period $t \geq 1$ onward. The environment is described by a finite set of states, which evolves over time according to a probability distribution that depends on the current state and the firms' chosen prices. Each firm earns a profit in every period as a function of the current state and the full price vector. To evaluate behavior over time, we define value functions that capture expected discounted profit. These value functions form the basis for our analysis of long-run behavior and equilibrium. The use of one-memory policies connects to prior work on bounded-recall and finite automaton strategies in repeated games (e.g., \cite{rubinstein1986finite}, \cite{lehrer1988repeated},  \cite{aumann1989cooperation} and  \cite{barlo2009repeated}).

We analyze how firms learn in our stochastic game setting by studying a variant of the $Q$-learning algorithm, one of the most widely used approaches in reinforcement learning. $Q$-learning enables agents to estimate the long-run value of actions through repeated interaction with the environment, without requiring knowledge of transition probabilities or future profits. This makes it a natural candidate for modeling firms that adaptively update their pricing policies based solely on observed outcomes.

We first consider a version of $Q$-learning \emph{without experimentation}, in which firms always choose prices that maximize their current estimated value function, known as the $Q$-function. We show that, at a fixed point of this algorithm, each firm's $Q$-value at its greedy action coincides with its conditional value function in the stochastic game under a specific class of one-memory policies, which we refer to as induced policies.

We then introduce a more realistic version of the algorithm, known as \emph{$Q$-learning with bounded experimentation}. In this setting, firms initially explore pricing actions using a softmax response---occasionally choosing suboptimal prices---but eventually switch to greedy behavior based on their learned $Q$-functions. We identify conditions under which firms using this $Q$-learning process learn to charge supracompetitive prices.  
Our results apply to widely studied economic environments, including dynamic Bertrand competition and recent models of platform markets (e.g., \cite{tirole1988theory}, \cite{dewenter2011semi} and  \cite{chica2025competition}).

The sufficient conditions for learning to charge supracompetitive prices involve comparisons between the profits from the collusive-enabling price and the $Q$-values of alternative actions at the time experimentation ends. Intuitively, they ensure that the collusive-enabling price is reinforced through learning as the most profitable option, both in the short run and over time.  These conditions are inequalities on the action values at the end of experimentation, which the learning process itself produces. We therefore ask whether they are actually reached, and show that they are, with positive and explicitly bounded probability, through two channels. First, when the initial $Q$-table is drawn at random, the conditions can be secured well before experimentation ends. Second, starting from the uninformative table in which every action value is zero, the softmax draws of the experimentation phase alone lead to absorption at the collusive-enabling price with positive probability. We bound this probability from below using a single path, on which every firm charges the collusive-enabling price in every period of experimentation. On this path firms never try to undercut, so they never learn that undercutting pays: supracompetitive prices emerge without communication, without any punishment being executed, and without any instruction to collude. Because the argument uses no property of the collusive-enabling price other than its positive profits, it applies equally to the competitive price, and the long-run outcome is selected by the realized exploration path. 

We show that such collusive behavior can be generated by three types of policy profiles: naive collusion, grim trigger policies, and increasing policies.  The latter two involve punishment and reward, aligning with pricing behavior observed in simulation studies. In fact, we show that naive collusion is not a Nash equilibrium unless the collusive-enabling price is a one-stage Nash equilibrium, whereas grim trigger policies are subgame perfect under a standard discounting condition.  

 To our knowledge, this is the first analysis of standard $Q$-learning with bounded experimentation that establishes, for $n$-firm repeated pricing games with an arbitrary finite price grid, both conditions under which supracompetitive prices persist and an explicit lower bound on the probability that they emerge.

\subsection*{Related Literature} 

This paper contributes to the growing literature on algorithmic pricing and collusion, particularly under reinforcement learning. A number of recent studies (e.g., \cite{waltman2008q}, \cite{calvano2020artificial}, \cite{klein2021autonomous} and  \cite{chica2024artificial}) have shown via simulations that $Q$-learning agents can learn to charge supracompetitive prices in repeated pricing environments. These findings have raised concerns among policymakers and competition authorities (e.g., \cite{OECD2017, assad2024algorithmic}) about the potential for algorithmic collusion, even without explicit coordination.

Recent theoretical work has shown that simple algorithmic pricing rules can lead to higher prices in competitive markets, even in the absence of explicit coordination \citep{brown2023competition}. However, these results do not address reinforcement learning. A widely used approach in this domain is $Q$-learning, introduced by \cite{watkins1992q}, which allows agents to estimate long-run profit-maximizing policies without knowing the environment's transition structure. While convergence of $Q$-learning is well understood in the single-agent case \citep{jaakkola1993convergence}, much less is known in multi-agent settings. Existing work on multi-agent learning, such as \cite{hu2003nash}, assumes agents compute equilibrium strategies at each stage, which is far from what is observed in decentralized learning environments.

\cite{possnig2023reinforcement} studies actor-critic reinforcement learning in a repeated Cournot game using stochastic approximation methods: long-run policies correspond to stable points of an associated differential equation, and when algorithms condition on information richer than the past period's price, policies converge with positive probability to a collusive equilibrium. \cite{possnig2024bounds} relates the set of outcomes that a broad class of learning dynamics, including $Q$-learning, can reach in repeated games to the folk theorem. Our approach is complementary and deliberately elementary: we work directly with the discrete, tabular $Q$-learning process at a finite time, obtain pathwise absorption, and bound the probability of emergence in closed form. We do not characterize the limiting distribution over outcomes.

\cite{banchio2023adaptive} show that independent reinforcement-learning agents can endogenously synchronize their actions, leading to tacit collusion without explicit coordination. This ``spontaneous coupling'' arises from persistent, correlated estimation errors under infrequent exploration, which cause agents to repeatedly cycle through high-price actions.  In our analysis, by contrast, supracompetitive prices are absorbing: once experimentation ends, play locks in at a fixed price profile. 

Closest to our emergence results is \cite{bertrand2025selfplay}, who study $Q$-learners with one-period memory in the iterated prisoner's dilemma. They characterize conditions under which self-play $Q$-learning, in which both players share a single $Q$-table, provably learns the cooperative win-stay, lose-shift (Pavlov) policy rather than permanent defection; their conditions involve optimistic initial $Q$-values, and they extend the result to $\varepsilon$-greedy exploration. Our random-initialization result (Proposition~\ref{prop:random}) operates through a related initialization channel. We differ in three respects. We allow $n$ firms with independent $Q$-tables and an arbitrary finite price grid; our absorption result holds at any price (Appendix~\ref{app:anyprice}), so that collusive and competitive lock-in are two instances of one theorem; and Theorem~\ref{thm:cleanslate} obtains lock-in from the uninformative initialization $Q_1\equiv 0$, without optimism.

The mechanism behind Theorem~\ref{thm:cleanslate} has two precedents. The first is incomplete learning. In \cite{rothschild1974two}, a seller who does not know its demand learns about it by charging prices, and because the profitability of a price is revised only when that price is charged, even a seller who experiments optimally may settle, with positive probability, on a price that is not the most profitable one. Our firms stop at $p^C$ for the same reason; what they never learn is the value of undercutting their rivals. The second is the self-confirming equilibrium of \cite{fudenberg1993self}, in which beliefs about the opponents are correct on the path of play and may be wrong off it. 

\cite{abada2023artificial} attribute seemingly collusive outcomes of independent $Q$-learning algorithms to imperfect exploration rather than to algorithmic sophistication. \cite{calvano2023genuine} call collusion genuine when high prices are sustained by a reward-punishment scheme and spurious when the algorithms merely fail to find the undercutting that would pay. In this terminology the lock-in of Theorem~\ref{thm:cleanslate} is spurious. Our contribution is to give this mechanism an explicit probability, to show in Proposition~\ref{prop:random} that the initial table alone can decide the outcome, and to describe in Propositions~\ref{prop:grimtriggerQ} and~\ref{prop:increasing_Q} the tables under which the learned behavior instead punishes deviations. Finally, \cite{dolgopolov2024reinforcement} characterizes the stochastically stable outcomes of stateless $Q$-learning in a prisoner's dilemma as experimentation vanishes and finds that cooperation without memory is possible under logit exploration only if learning is fast enough; we allow memory of the previous price profile, $n$ firms and an arbitrary finite price grid.

Our framework also contributes to the literature on general-sum stochastic games and on strategies with bounded memory. Classical work (e.g., \cite{fink1964equilibrium}) established the existence of stationary equilibria in stochastic games. Similarly, \cite{chakrabarti1999markov} showed that Markov perfect equilibrium exists for a specific class of stochastic games.\footnote{More specifically, stochastic games where the transition probabilities are Markovian, product-measurable in the previous period's realizations of the state and actions, and norm-continuous with respect to past actions.} We analyze a broader class of one-memory policies that accommodate punishment and reward behavior, such as grim trigger strategies.\footnote{Our one-memory equilibria generalize Markov equilibria by allowing strategies to depend explicitly on the last period's joint actions, which enables the modeling of punishment and reward mechanisms crucial for collusive outcomes.} This notion of memory-bounded behavior has also been studied in repeated games, where \cite{rubinstein1986finite} introduced finite automata strategies,   \cite{lehrer1988repeated} characterized Nash equilibria under bounded recall and \cite{aumann1989cooperation} analyzed cooperation under bounded recall. Our results complement those of \cite{barlo2009repeated}, who showed that one-memory strategies can support any individually rational payoff as a subgame perfect equilibrium when players are sufficiently patient. Existence of one-memory SPEs in our setting follows from the existence of Markov perfect equilibria; what we add is a verification procedure, adapted from \cite{fink1964equilibrium}, that checks whether a given one-memory profile is an SPE.

\subsection*{Structure of the Rest of the Paper}
Section~\ref{sec:Model} introduces stochastic games with one-memory strategies and the solution concepts we use. Section~\ref{sec:RLalgo} relates $Q$-learning to the value functions of the game and gives conditions under which supracompetitive prices become absorbing once experimentation ends. Section~\ref{sec:emergence} bounds the probability that these conditions are reached, from a random and from an all-zero initial table. Section~\ref{sec:supporting} studies the strategies that the learned table induces and whether naive collusion and grim trigger strategies are equilibria. Section~\ref{sec:policy_extensions} discusses policy implications and extensions, and Section~\ref{sec:conclusion} concludes. Appendices~\ref{AppendixA}--\ref{Sec:FinksProof} contain the proofs and extensions of the main results.
\section{A Model for Stochastic Games with Bounded Memory}
\label{sec:Model}
We adapt the stochastic game framework of \cite{fink1964equilibrium}, originally developed for stationary strategies, to accommodate one-memory strategies. This model generalizes repeated games with perfect monitoring; however, certain parts of our analysis---specifically Proposition~\ref{prop:grimTriger}, Section~\ref{sect:convergenceQlearning} and Sections~\ref{sec:emergence}--\ref{sec:policy_extensions}---focus on the repeated game case. A special case of one-memory strategies are the  well-known Markov strategies (see, e.g., \cite{mailath2018modeling}).

We begin by describing the basic components of the stochastic game. 
To make the setting more concrete, we assume that $n$ firms (or agents) compete by setting prices over an infinite time horizon, where each firm is indexed by $i \in [n] := \{1, \dots, n\}$. More generally, we consider a finite, ordered set of actions, which in our context correspond to prices. 
After describing the model, Section~\ref{sec:v_function} defines two types of conditional value functions for firm $i$ and establishes their basic properties. 
Section~\ref{sec:nash_definition} formalizes the notions of best response, Nash equilibrium from time $t = 1$, and a subgame perfect equilibrium (SPE). Finally, Section~\ref{sec:construction} explains how one-memory SPEs are constructed and verified.

\textbf{Actions:} We assume a set of actions $\A:= \{a^0,\dots, a^m\}$. We recall that in our context taking actions means charging prices. The set of actions for $n$ agents is $\A^n$ and we commonly denote by $\p =(p^1,\dots, p^n)$, a vector of prices in $\A^n$.

\textbf{States and their dynamics:} We assume a state space of $r$ states: $\S := \{s^1,\dots, s^r\}$. 
Every state may represent a market demand or cost level, which will directly affect the profit functions defined below. States change with time and consequently affect the profits agents receive. 
At time $t+1$, given state $s_t = s \in \S$ and vector of prices $\p =(p^1,\dots, p^n) \in \A^n$, the state at $t+1$, $s_{t+1} \in \S$, follows the probabilistic law 
\begin{equation}
\label{eq:def_P}
s_{t+1}\sim \P(\cdot | \p, s).     
\end{equation} 
Therefore, the state at $t+1$ only depends on the state and price vector at time $t$. 

\textbf{Profit functions:} The profit function for each firm $i$ is a function, 
\begin{equation}
    \label{eqn:profit_function}
    \pi^i:\A^n\times\S \rightarrow \R.
\end{equation}
We note that it is a function of the current vector of prices,  $\p =(p^1,\dots, p^n) \in \A^n$, and state, $s \in \S$,  but independent of the time $t$. Moreover, we assume that $\pi^i\geq 0$. 
In the reinforcement learning literature, 
$\pi^i$   is commonly referred to as the reward function.

\textbf{Policies:} A policy, or strategy, for firm $i$ is a sequence of probability distributions $\ssigma^i = (\sigma_t^i)_{t=0}^\infty$ over the action space $\A$.\footnote{In machine learning, the term ``policy'' is commonly used, whereas in economics, the term ``strategy'' is more standard.} Considering all $n$ firms, the overall policy is $\ssigma = (\ssigma^i)_{i \in [n]}$. 
At time $t=0$ and given a state $s_0\in\S$, firm $i$ chooses $p\in\A$ with probability $\sigma^i_0(p|s_0)$, where $\sum_{p\in \A}\sigma^i_0(p|s_0) = 1$. 
Let $p^i_{t-1}$ denote the price chosen by firm $i$ in period $t-1$
and let $\p_{t-1}=(p^1_{t-1},\dots,p^n_{t-1})\in \A^n$ denote the vector of all these prices. 
We assume that at time $t\geq 1$, $\p_{t-1}$ is publicly available.  
At time $t\geq1$ and given $s_t\in \S$ and $\p_{t-1}\in\A^{n}$, firm $i$ chooses $p\in \A$ with probability $\sigma^i_t(p|\p_{t-1},s_t)$, where $\sum_{p\in \A}\sigma^i_t(p|\p_{t-1},s_t) = 1$. 
We assume that $\sigma_t^1(p_{t}^1|\p_{t-1},s_t),\dots,\sigma_t^n(p_t^n|\p_{t-1},s_t)$ are independent random variables. Consequently, we define
$$\sigma_t(\p_t|\p_{t-1},s_t) = \prod_{i=1}^n\sigma_t^i(p_t^i|\p_{t-1},s_t) \ \text{ and }
\ \sigma^{-i}_t(\p_t|\p_{t-1},s_t) = \prod_{j\neq i}\sigma_t^j(p_t^j|\p_{t-1},s_t).$$
We similarly define $\sigma_0(\p_0|s_0)$ and $\sigma_0^{-i}(\p_0|s_0)$. 

We impose a key modeling assumption, commonly used in repeated games with bounded memory\footnote{For simplicity, we focus on one-memory strategies. Nevertheless, some of our results may extend to strategies with finite-length memory, though doing so would require significantly more cumbersome notation and technical development.}  
(see, e.g., \cite{barlo2009repeated} and \cite{barlo2016bounded}):
\begin{assumption}[One-memory policies]\label{A_onememory}
   Firms choose policies that depend only on the current state and the previous period's actions, and remain fixed for all $t \geq 1$. That is, for each $t \geq 1$, $\sigma^i_t(p | \p_{t-1}, s_t)$ is independent of $t$ and depends only on $p \in \A$, $\p_{t-1} \in \A^n$, and $s_t \in \S$, while at $t = 0$, $\sigma^i_0(p | s_0)$ depends only on $p \in \A$ and $s_0 \in \S$.
\end{assumption}

\begin{remark}[Reading guide]\label{rem:readingguide}
From Section~\ref{sect:convergenceQlearning} onward, every result is stated under Assumption~\ref{assumption_deltai}, which sets $|\S|=1$. A reader interested only in the repeated pricing game may therefore take the state to be the previous price profile, $\s_t=\p_{t-1}\in\A^n$, and read $\pi^i(\p_t,s_t)$ as $\pi^i(\p_t)$ throughout those sections. The general stochastic-game formulation is used for Propositions~\ref{prop:QfixedisV1} and~\ref{prop:Sufficient_Q_inducesNash} and for the existence of one-memory SPEs (Appendix~\ref{sec:NashT1}).
\end{remark}

We remark that while we use in different places the general term $\sigma^i_t(p | \p_{t-1}, s_t)$, the above assumption implies that it equals $\sigma^i_1(p | \p_{t-1}, s_t)$ for all $t \geq 1$, and $\sigma^i_0(p | s_0)$ for $t = 0$. Similarly, we note that the overall policy $\ssigma$  can be identified with $(\ssigma_0, \ssigma_1)$, the pair of overall policies used at time $t = 0$ and for all $t \geq 1$, respectively.

\textbf{Solution Concept:} We study the existence of a one-memory subgame perfect equilibrium (SPE) of the stochastic game---a refinement of Nash equilibrium in which firms' strategies must be optimal at every possible decision point. The formal definition is provided in Section~\ref{sec:nash_definition}.

\textbf{Additional Notation:} We introduce notation used throughout the paper to compactly describe policy spaces, expectations, and value functions.\\ (i) We denote $M = |\A^n|$ and write the set $\S\times\A^n$ as follows 
\begin{equation}\label{not_StimesAn}
    \S\times\A^n=\left\{(s^1,\p^1),\cdots, (s^1,\p^M), \cdots, (s^r,\p^1),\cdots, (s^r,\p^M)\right\}.
\end{equation}
(ii) The set of policies available at time $t\geq 0$ for firm $i$ is denoted by $\SSigma_t^i$. Using the enumeration in \eqref{not_StimesAn} and the notation $\hat{M}=(m+1)rM$, the set $\SSigma_t^i$ can be represented as\footnote{It follows from \eqref{def:Sigma_ti} that $\SSigma_t^i$ is a product of $rM$ copies of the $m$-dimensional simplex, and consequently it is a compact and convex subset of $\mathbb{R}^{(m+1)rM}$. } 
    \begin{align}
        \SSigma_t^i = \{& (\sigma_t^i(a^0|\p^1,s^1),\dots,\sigma_t^i(a^m|\p^1,s^1),\dots,\sigma_t^i(a^0|\p^M,s^r),\dots,\sigma_t^i(a^m|\p^M,s^r))\in [0,1]^{\hat{M}} \notag \\
        &\textnormal{s.t. } \sum_{k=0}^m\sigma_t^i(a^k|\p_0,s_1) = 1 \quad \forall (s_1,\p_0)\in \S\times\A^n \}.\label{def:Sigma_ti}
    \end{align}

The set of policies at time $t \geq 0$ for all firms is $\SSigma_t:= \times_{i=1}^n \SSigma_t^i$. 
The set of all policies is $\SSigma:= \times_{t \geq 0} \SSigma_t$.

(iii) A policy profile for time $t\geq0$ contains the policies for all firms at that time and is described by $\ssigma_t = (\sigma_t^i)_{i=1}^n$. 
We denote by $\ssigma_t^{-i} = (\sigma_t^j)_{j \neq i}$ the profile excluding firm $i$'s policy at time $t$.
Similarly,  $\SSigma_t^{-i}:=\times_{j\neq i}\SSigma_t^j$.  For each $i\in[n]$, we interchange between $(\sigma_t^i,\ssigma_t^{-i})$ and $(\ssigma_t^j)_{j=1}^n$.\\ 

\subsection{\texorpdfstring{The $\boldsymbol{V}^i$-Functions}{}}
\label{sec:v_function}

The initial state $s_0\in \S$ along with a profile of policies for all firms $\ssigma\in\SSigma$  determine the evolution of the stochastic game via 
conditional value functions, which we clarify in this section. Let $\ssigma = (\ssigma_t)_{t=0}^\infty \in \SSigma$ be a one-memory policy. We recall that by Assumption \ref{A_onememory}, for each firm $i\in[n]$, $\ssigma$ is characterized by two policies: (i) $\sigma_0^i(\cdot|s_0)$ at $t=0$; and (ii) $\sigma_1^i(\cdot|\p_{t-1},s_t)$ at $t \geq 1$. We will thus obtain conditional value functions for $t=0$ and $t=1$. 

We define the conditional value function using the well-known definition of $\E_{\ssigma,\P}$ in \eqref{not_Esigma_P_g} and \eqref{not_Esigma_P_g2} of Appendix \ref{sect:formalizeEPsigma}.
We recall that $\P$ is the distribution defined in \eqref{eq:def_P},  and $\pi^i: \A^n \times \S \to \R$,  $i \in [n]$, are the profit functions. Let  $\delta_i\in(0,1)$ denote the discount factor for firm $i \in [n]$, which represents the present value of future profits.   
For $\ssigma = (\ssigma^i,\ssigma^{-i}) \in\SSigma$ and $s_0\in \S$, the conditional value function at time $t=0$ of firm $i$ is given by
\begin{equation}\label{conditional_value_function}
\begin{split}
    \Vzerot(s_0, \ssigma^i|\ssigma^{-i}) &:=\E_{\ssigma, \P}\left[ \sum_{t=0}^\infty \delta_i^t \pi^i(\p_t,s_t)\Big| s_0 \right].
\end{split}
\end{equation}
Given state $s_0$ at time $t=0$, (\ref{conditional_value_function}) measures the expected payoff that firm $i$ receives after playing the infinite stochastic game using $\ssigma^i$, while firms other than $i$ follow $\ssigma^{-i}$. Since  $\pi^i(\p_t,s_t)$ is bounded by $\sup_{(s,\p)\in\S\times \A^n} |\pi^i(\p,s)|$
for all $i\in[n]$ and $t\geq0$, (\ref{conditional_value_function}) is bounded by $(1-\delta_i)^{-1}\sup_{(s,\p)\in\S\times \A^n} |\pi^i(\p,s)|$ and thus well-defined. 

Next, we characterize the conditional value function of firm $i$ at time $t=1$. 
For $s_1\in \S$, $\p_0\in \A^n$ and  $\ssigma_1 = (\sigma^i_1, \ssigma^{-i}_1)\in \SSigma_1$, the conditional value function of firm $i$ at time $t=1$ is given by 
\begin{equation}\label{conditional_value_function_t1}
\begin{split}
    \Vonet
(s_1, \p_0, \sigma^i_1|\ssigma^{-i}_1) &:=\E_{\ssigma_1, \P}\left[ \sum_{t=1}^\infty \delta_i^{t-1} \pi^i(\p_t,s_t)\Big|\p_0, s_1 \right].
\end{split}
\end{equation}
For the pair $(s_1,\p_0)$ at time $t=1$, (\ref{conditional_value_function_t1}) measures the expected payoff that firm $i$ receives after playing the infinite stochastic game using $\sigma^i_1$, while firms other than $i$ follow $\ssigma^{-i}_1$. 
If firm $i$ uses a policy $\sigma_1^i\in\SSigma_1^i$ such that $\sigma_1^i(\tilde{a}|\p_0,s_1) = 1$ for $\tilde{a}\in\A$ and for all $(\p_0,s_1)\in\A^n\times\S$, we write $\Vonet(s_1,\p_0,\tilde{a}|\ssigma_1^{-i})$ instead of $\Vonet(s_1, \p_0, \sigma^i_1|\ssigma^{-i}_1)$. 

For technical reasons that will be explained in the next section, it is useful to define a $\Vone$ vector function. Its definition below uses a vector $\v$ whose coordinates are indexed by $i \in [n]$ and $(s_1,\p_0) \in \S\times\A^n$. In view of the enumeration of $\S\times\A^n$ in \eqref{not_StimesAn}, $\v \in \R^{n r M}$. The $\Vone$ vector function is given by $$\Vone:\SSigma_1\times \SSigma_1 \times\mathbb{R}^{nrM}\longrightarrow \mathbb{R}^{nrM} \textnormal{ s.t. }(\ssigma_1,\ttau_1,\v)\mapsto\Vone(\ssigma_1,\ttau_1,\v)$$ where the $(i,s_1,\p_0)$-coordinate of $\Vone(\ssigma_1,\ttau_1,\v)$ is given by 
\begin{equation}\label{V1_function}
\begin{split}
&\Vone(\ssigma_1,\ttau_1,\v)_{i,s_1,\p_0}\\
&:= \sum_{\p_1\in\A^n}\tau_1^i(p_1^i|\p_0,s_1)\sigma_1^{-i}(\p_1^{-i}|\p_0,s_1)\left[\pi^i(\p_1,s_1)+\delta_i\sum_{s_2\in\S}\P(s_2|\p_1,s_1)v_{i,s_2,\p_1}\right].
\end{split}
\end{equation}
For the pair $(s_1, \p_0)$, equation~\eqref{V1_function} represents firm $i$'s expected payoff from time $t = 1$ to time $t = 2$, assuming that firm $i$ follows $\tau_1^i$ at time $t = 1$, firms other than $i$ follow $\ssigma_1^{-i}$, and the payoffs for all firms at time $t = 2$ are given by the vector $\v$. Note that the $(i, s_1, \p_0)$-coordinate of $\Vone(\ssigma_1, \ttau_1, \v)$ depends only on $\tau_1^i$. For this reason, when no confusion can arise, we often write $\Vone(\ssigma_1, \tau_1^i, \v)_{i, s_1, \p_0}$ instead of $\Vone(\ssigma_1, \ttau_1, \v)_{i, s_1, \p_0}$.


Equation \eqref{Bellman_equationV1} in Appendix \ref{sec:relation_vi} also establishes the following direct relationship between the conditional value function at time $t=1$ 
and the vector-valued function 
$\Vone$, facilitating our analysis of equilibrium conditions: 
\begin{equation}
\label{eqn:equality_V1_V1conditional}
    \Vonet(s_1, \p_0, \sigma^i_1|\ssigma^{-i}_1)
    =\Vone(\ssigma_1,\ssigma_1,\tilde{\v})_{i,s_1,\p_0},
\end{equation}
for each $(i,s_1,\p_0)$-coordinate, where $\tilde{\v}_{i,s_1,\p_0}:= \Vonet(s_1, \p_0, \sigma^i_1|\ssigma^{-i}_1)$. To see this, observe that equation~\eqref{Bellman_equationV1} is identical to \eqref{V1_function} when we set $\tau_1 = \sigma_1$ and $\v = \tilde{\v}$.

\subsection{Nash equilibrium}
\label{sec:nash_definition}

Using the definitions of the two conditional value functions at times $t=0$ and $t=1$, we define the concepts of a Nash equilibrium from time $t=1$ and an SPE. 

A policy  $\sigma^{i*}_1 \equiv (\sigma^{*})_1^{i} \in \SSigma_1^i$ is called a best-response policy to $\ssigma^{-i}_1\in \SSigma_1^{-i}$ if for all $(s_1,\p_0)\in\S\times\A^n$, 
\begin{equation}\label{sigmai_bestR_V1}
    \sigma_1^{i*}\in \argmax_{\sigma^i_1\in \SSigma^i_1} \Vonet(s_1, \p_0, \sigma^{i}_1|\ssigma^{-i}_1),
\end{equation}
where $\Vonet(s_1, \p_0, \sigma^{i}_1|\ssigma^{-i}_1)$ is given by \eqref{conditional_value_function_t1}. 
We say that  $\ssigma^*_1\in\SSigma_1$ is a Nash equilibrium from time $t=1$, if for all $i\in[n]$,
$\sigma^{i*}_1$  is a best-response policy to $\ssigma^{-i*}_1 \equiv (\ssigma^{*})_1^{-i}$. In other words, $\ssigma^*_1 \in\SSigma_1$ is a Nash equilibrium from time $t=1$, if for all $i\in[n]$, and $(s_1,\p_0)\in\S\times\A^n$,  
\begin{equation}\label{eqn:def_nash_eq}
    \sigma_1^{i*}\in \argmax_{\sigma^i_1\in \SSigma^i_1} \Vonet(s_1, \p_0, \sigma^{i}_1|\ssigma^{-i*}_1). 
\end{equation}

We define a subgame perfect  equilibrium (SPE) as a profile $(\ssigma_0^*,\ssigma_1^*)$ such that $\ssigma_1^*$ is a Nash equilibrium from time $t=1$, and for each $i\in[n]$, $\ssigma_0^*\in\SSigma_0$ satisfies 
\begin{equation}\label{eqn:Nasht_0}
    \sigma_0^{i*}\in \argmax_{\sigma_0^i\in\SSigma_0^i} \Vzerot(s_0, (\sigma_0^{i},\sigma_1^{i*})|(\ssigma_0^{-i*},\ssigma_1^{-i*})).
\end{equation}
That is, no firm can profitably deviate from its initial strategy $\sigma_0^i$, given that all players follow the strategy profile  $\ssigma_1^*$ from time $t=1$ onward. A one-memory SPE is a subgame perfect equilibrium of the stochastic game in the usual sense.\footnote{Restricting the deviations in \eqref{sigmai_bestR_V1} and \eqref{eqn:Nasht_0} to one-memory policies is without loss of generality. When the other firms use one-memory policies, firm $i$ faces a discounted Markov decision problem with the finite state space $\S\times\A^n$, and such a problem admits a one-memory policy that is optimal from every state among all history-dependent policies \citep{blackwell1965discounted}.}

\subsection{Construction of One-Memory SPEs}\label{sec:construction} 
Since Markov Perfect Equilibrium (MPE) is a particular case of one-memory SPEs (see, e.g., \cite{mailath2018modeling}), existence of one-memory SPEs follows from the existence of MPEs, which is well-known (see, e.g., \cite{chakrabarti1999markov}). Nevertheless, we prefer developing another direct proof, which extends  \cite{fink1964equilibrium} to the one-memory setting. The ideas of this proof are crucial for developing our next theories, but since the existence itself is not novel, we present it in the appendix. Specifically, Appendix \ref{sec:NashT1} establishes the existence of a one-memory SPE and formulates an algorithm for verifying whether a given profile satisfies this condition. Furthermore, Appendix~\ref{sec:Repeated_Bertrand} demonstrates the application of this theory to the well-known grim trigger strategies.

\section{\texorpdfstring{Collusion Stability under $Q$-Learning}{}} 
\label{sec:RLalgo}
This section establishes key properties of $Q$-learning \citep{watkins1992q}, one of the most widely used reinforcement learning algorithms. 
Section~\ref{sect:basicQlearning} introduces a version of $Q$-learning without experimentation, adapted to the stochastic game framework developed in Section~\ref{sec:Model}. We establish a connection between the fixed points of this algorithm and the $\V^i$-functions defined in Section~\ref{sec:v_function}, showing that these fixed points correspond to the value of the stochastic game at time $t=1$ under a specific class of strategies, which we refer to as induced strategies. We then provide sufficient conditions under which the induced strategies form a Nash equilibrium from time~$t = 1$. Since these strategies are one-memory strategies, the results developed in Section~\ref{sec:NashT1} apply directly. 
Section~\ref{sect:convergenceQlearning} studies a version of $Q$-learning with bounded experimentation. We provide sufficient conditions for its convergence in stochastic games satisfying Assumption~\ref{assumption_deltai}, including the standard dynamic Bertrand competition model as a special case. We also give sufficient conditions under which $Q$-learning leads firms to consistently choose supracompetitive prices. Section~\ref{sec:emergence} bounds the probability that these conditions are reached, and Section~\ref{sec:supporting} then identifies sufficient conditions under which these supracompetitive prices are generated by one of three classes of strategies: naive collusion, grim trigger strategies, or increasing strategies. 

\subsection{\texorpdfstring{A Relationship of a $\Q$-Learning Algorithm with the Stochastic Game}{}}\label{sect:basicQlearning} 
We formulate a version of the $Q$-learning algorithm with no experimentation, while assuming the multi-agent setting of Section \ref{sec:Model}. We then establish the relationship of the $Q$-function of this algorithm with the value functions, $\Vone$ and $\Vonet$, of the stochastic game. The basic idea of this algorithm is to find a policy that maximizes \eqref{conditional_value_function_t1} given the policies of all other agents. The algorithm takes as input $Q_1^i: \S\times \A^{n+1} \to \R$ for $i\in [n]$, as well as several parameters, and outputs $Q_t^i: \S \times \A^{n+1} \to \R$ for $i\in [n]$ and ${t\geq 2}$. We use the notation $\s = (s,\p) \in \S\times\A^n$.

\begin{algorithm}[$Q$-learning with no experimentation]\label{Algo_Q_argmax} Fix an initial price profile $\p_0\in\A^n$, an initial state $s_1\in\S$, and an initialization $Q_1^j:\S\times\A^{n+1}\to\R$ for each $j\in[n]$; the baseline initialization is $Q_1^j\equiv 0$. At each time $t\geq 1$, and given $\s_t=(s_t,\p_{t-1})$:
\begin{enumerate}
\item[(a)] each firm $i$ selects, with uniform probability among the maximizers, a price
\begin{equation}\label{eqn:argmax_prob}
p_t^i  \in \argmax_{a\in\A} Q^i_t(\s_t,a);
\end{equation}
\item[(b)] the profile $\p_t$ and the profits $(\pi^j(\p_t,s_t))_{j=1}^n$ are observed, and the next state $\s_{t+1}=(s_{t+1},\p_t)$ is drawn with probability
$\P(s_{t+1}|\p_t,s_t)$, where $\P$ is the transition law in \eqref{eq:def_P};
\item[(c)] each firm $i$ updates, for every $(\s,p)\in\S\times\A^{n+1}$,
\begin{equation}\label{def:Q_learningUpdate}
    Q_{t+1}^i(\s,p) = (1-\alpha_t)Q_{t}^i(\s,p)+\alpha_t\left\{\pi^i(\p_t,s_t) +\delta_i \mathbb{E}_{\s_{t+1}}\left[\max_{a\in\A}Q_t^i(\s_{t+1},a)\right]
    \right\},
\end{equation}
where $\alpha_t=\alpha_t(\s,p)\in[0,1]$ satisfies $\alpha_t(\s,p)=0$ for every $(\s,p)\neq(\s_t,p_t^i)$. That is, firm $i$ revises only the entry it has just visited.
\end{enumerate}
\end{algorithm}

Since $\S$ and $\A$ are finite, each $Q^i_t$ is a finite array of action values, one for each state $\s$ and price $p$; we call it firm $i$'s $Q$-table at time $t$, and $Q_1=(Q^i_1)_{i\in[n]}$ the initial table.  
When $|\S|>1$, the target in \eqref{def:Q_learningUpdate} averages over the next state instead of using its realization, so that the fixed points studied below are well defined; this step requires knowledge of $\P$. Under Assumption~\ref{assumption_deltai}, which we impose from Section~\ref{sect:convergenceQlearning} onward, the next state is $\s_{t+1}=\p_t$, and \eqref{def:Q_learningUpdate} is the standard update of \cite{watkins1992q}.

Let $\Q_{\!f} = (Q_{\!f}^{i})_{i=1}^n$ be a collection of functions $Q_{\!f}^{i}:\S\times\A^{n+1}\to\R$, and assume that starting from time $t=1$, firms use $Q_{\!f}^{i}$ to play the stochastic game described in Section \ref{sec:Model} as follows: Given $\s\in\S\times\A^n$, each firm $i\in[n]$ chooses
\begin{equation}\label{def:induced_strategy_Qil}
w_{\!f}^{i}(\s)\in\argmax_{p\in\A}Q_{\!f}^{i}(\s,p). 
\end{equation}
We denote $\w_{\!f}(\s)=(w_{\!f}^{i}(\s))_{i=1}^n$.  The latter strategies are often referred to as the strategies induced by $\Q_{\!f}$. Moreover, $\w_{\!f}(\s)$ 
constitutes a one-memory strategy, since $\s$ encodes the previous period's price profile. We say that $\Q_{\!f}$ is a fixed point of Algorithm~\ref{Algo_Q_argmax} under a constant learning rate $\alpha_t=\alpha\in(0,1]$ if the update rule \eqref{def:Q_learningUpdate} leaves it unchanged: for each $i\in[n]$ and $\s\in\S\times\A^n$, when the firms play $\w_{\!f}(\s)$ at state $\s$, the right-hand side of \eqref{def:Q_learningUpdate} evaluated at $\Q_{\!f}$ and at the entry $(\s,w_{\!f}^{i}(\s))$ equals $Q_{\!f}^{i}(\s,w_{\!f}^{i}(\s))$. The following proposition shows that if agents play the stochastic game following the strategies induced by $\Q_{\!f}$, then the \textit{conditional} value function of firm $i$  at time $t=1$ (see \eqref{conditional_value_function_t1}) coincides with $Q_{\!f}^{i}$ at the induced strategies. 
\begin{proposition}[$Q_{\!f}^{i}$ captures the value of the game at time $t=1$]\label{prop:QfixedisV1}
    Assume $\alpha_t = \alpha\in (0,1]$ for each $t\geq 0$ and  $(Q_{\!f}^{i})_{i=1}^n$ is a fixed point of Algorithm \ref{Algo_Q_argmax}. Then, for each $i\in[n]$ and $\s=(s_1,\p_0)\in\S\times\A^n$,
    \begin{equation}\label{eqn:Q*fixedisV1}
        Q_{\!f}^{i}(\s, w_{\!f}^{i}(\s)) = \Vonet(\s, w_{\!f}^{i}|\w_{\!f}^{-i}).
    \end{equation}
\end{proposition}

Proposition~\ref{prop:QfixedisV1} is a policy-evaluation statement: at a fixed point, each firm's $Q$-value at its own greedy action equals the value of the game from time $t=1$ when every firm plays its induced one-memory strategy. It does not by itself show that the induced strategies form a Nash equilibrium from time $t=1$. The following proposition gives a sufficient condition for that.

\begin{proposition}[Sufficient condition for $\Q_{\!f}$ to induce a Nash equilibrium from time $t=1$]\label{prop:Sufficient_Q_inducesNash}
    Assume $\alpha_t = \alpha\in (0,1]$ for each $t\geq 0$, $\Q_{\!f}$ is a fixed point of Algorithm \ref{Algo_Q_argmax}, and for each $i\in[n]$ and $\s=(s_1,\p_0)\in\S\times\A^n$,
    \begin{equation}\label{eqn:sufficient_Q_inducesNash}\begin{split}
        &w_{\!f}^{i}(\s) \in \argmax_{p_1^i\in\A}\Vone(\w_{\!f},p_1^i,\v_{\!f})_{i,\s},\\
    \end{split}
    \end{equation}
where $\w_{\!f} = \{w_{\!f}^{i}(\s)| i\in[n], \s\in\S\times\A^n\}$, $\v_{\!f}\in\R^{nrM}$ is the vector with coordinates $(\v_{\!f})_{i,\s}:=Q_{\!f}^{i}(\s,w_{\!f}^{i}(\s))$, and $\Vone$ is given by \eqref{V1_function}. Then, the strategy induced by $\Q_{\!f}$ is a Nash equilibrium from time $t=1$. 
\end{proposition}

Suppose that given a state $\s\in\S\times\A^n$, firms play a one-stage game with payoffs given by $(\Vone(\cdot,\cdot,\v_{\!f})_{i,\s})_{i\in[n]}$. In this case, Proposition \ref{prop:Sufficient_Q_inducesNash} implies that if the induced strategy by $\Q_{\!f}$ is a Nash equilibrium of the latter one-stage game, then this strategy is a Nash equilibrium from time $t=1$ for the stochastic game of Section \ref{sec:Model}. This observation is interesting since Algorithm \ref{algo:HowtoCheckNash} in Appendix~\ref{sec:NashT1} requires checking two conditions in order to decide whether a given profile is a Nash equilibrium from time $t=1$. However, in the current case only one condition is needed because $\w_{\!f}(\s)$ is induced from a fixed-point of Algorithm \ref{Algo_Q_argmax}.

\subsection{\texorpdfstring{Stability of Supracompetitive Prices}{}}\label{sect:convergenceQlearning}

We demonstrate how $Q$-learning with bounded experimentation can yield stable supracompetitive pricing behavior, which may or may not align with equilibrium incentives. 

First, we specify an assumption that we use in our theory:
\begin{assumption}\label{assumption_deltai}
We require the following two conditions:
\begin{itemize}
    \item[(i)] $|\S| = 1$ and consequently $\pi^i(\p,s) \equiv \pi^i(\p)$.
    \item[(ii)] There exists a Nash equilibrium price $\p^*= (p^*,\dots,p^*)\in\A^n$ of the one-stage game $\{(\A,\pi^i)\}_{i=1}^n$. Furthermore, there exists a price $\p^C = (p^C,\dots,p^C)$ such that $\pi^i(\p^*) < \pi^i(\p^C)$ for each $i\in[n]$. 
We refer to $\p^*$ as the competition price and to $\p^C$ as the collusive-enabling price.
\end{itemize}
\end{assumption}

Condition (i) in Assumption \ref{assumption_deltai} implies that states used in Algorithm \ref{Algo_Q_boundedExperimentation} have the following form: For $t\geq 1$ $\s_t = \p_{t-1}\in \A^n$,  where  $\p_{t-1}$  is the price choice at time  $t-1$.\footnote{We remark that this state choice has been a standard assumption in recent articles on algorithmic pricing (see, e.g. \cite{calvano2020artificial}, \cite{klein2021autonomous} and \cite{chica2024artificial}.)} It reduces our stochastic game to an infinite repeated game. 
Under this condition, we may write $\pi^i(\p)$ instead of $\pi^i(\p,s)$ to refer to the profit function in \eqref{eqn:profit_function}. 

Condition (ii) in Assumption \ref{assumption_deltai} ensures the presence of both a Nash equilibrium price and a price that facilitates collusion. This condition aligns our stochastic game with a key feature of the dynamic Bertrand competition model (see, e.g., \cite{tirole1988theory}), and of recent models of platform competition in two-sided markets (see, e.g., \cite{dewenter2011semi} and \cite{chica2025competition}).

Next, we introduce $Q$-learning with bounded experimentation which combines softmax-based $Q$-learning with the version in Algorithm~\ref{Algo_Q_argmax}. The softmax-based variant of $Q$-learning replaces the deterministic choice of price as a maximum of the $Q$-function, stated in \eqref{eqn:argmax_prob}, with random drawing of the price according to the soft-max probability
\begin{equation}
    \label{eqn:softmax_prob}
    \sigma^i (p_t^i = a|\s_t) = \frac{e^{Q^i_t(\s_t,a)/\beta_t}}{\sum_{\tilde{a}\in \A }e^{Q^i_t(\s_t,\tilde{a})/\beta_t}},
\end{equation}
where $\beta_t >0$.\footnote{The rule in \eqref{eqn:argmax_prob} is recovered from \eqref{eqn:softmax_prob} by letting $\beta_t \to 0$. In this limit, $\mathbb{P}(p_t^i = \tilde{a} | \s) \to 1 / |\argmax_{a \in \A} Q_t^i(\s, a)|$ if $\tilde{a} \in \argmax_{a \in \A} Q_t^i(\s, a)$, and $\mathbb{P}(p_t^i = \tilde{a} | \s) \to 0$ otherwise.
}
This step introduces stochasticity and allows for ``experimentation'' with different prices.

\begin{algorithm}[$Q$-learning with bounded experimentation]
\label{Algo_Q_boundedExperimentation}
Fix a deadline $T\geq 3$.\footnote{$T\geq 3$ guarantees at least two experimentation periods, which is what Theorem~\ref{thm:cleanslate} needs for $y^i_T>0$.} Firms follow Algorithm~\ref{Algo_Q_argmax} at every time, except that during the experimentation phase $t\in\{1,\dots,T-1\}$ step~(a) is replaced by: firm $i$ draws $p_t^i$ from $\A$ according to the softmax probabilities in \eqref{eqn:softmax_prob}, independently across firms. From time $t=T$ onward firms follow Algorithm~\ref{Algo_Q_argmax} unchanged.
\end{algorithm}

The above algorithm captures the constraints firms are often faced with in a competitive market: experimentation is costly and cannot continue indefinitely. In practice, firms face budget and information constraints limiting their ability to continue with exploration. 

We now impose a technical assumption on the learning rate $\alpha_t$, which governs the update rule in Algorithm~\ref{Algo_Q_boundedExperimentation}:

\begin{assumption}\label{assumption_learningrate} The learning rate $\alpha_t$ satisfies the following:  $\alpha_t\in (0,1)$ for each $t\geq 0$ and  $\sum_{t=T}^\infty \alpha_t = \infty$.
\end{assumption}
\noindent The learning rate $\alpha_t$ measures how quickly a firm revises its estimate of future profitability: the smaller it is, the more the firm relies on past experience; the larger it is, the more reactive it becomes to what it has just observed. Assumption \ref{assumption_learningrate} requires only that the firm never stops revising, which is the standard requirement in the $Q$-learning literature \citep{watkins1992q, jaakkola1993convergence}. It is satisfied, in particular, by any constant rate $\alpha_t\equiv\alpha\in(0,1)$. The following lemma, which we use to evaluate the limit of the $Q$-values in Theorem \ref{prop:QConvergence_T}, is the only consequence of the assumption that we need.

\begin{lemma}\label{lem:alphadelta}
Let $\delta\in(0,1)$ and let $(\alpha_t)_{t\geq 0}$ satisfy Assumption \ref{assumption_learningrate}. Then
\begin{equation}\label{eqn:alphadelta}
\lim_{t\to\infty}\ \sum_{k=T+1}^{t}\ \prod_{l=k+1}^{t}
\big(1-\alpha_l(1-\delta)\big)\,\alpha_k\ =\ \frac{1}{1-\delta}.
\end{equation}
\end{lemma}

The above Lemma is proved in Appendix \ref{sect:properties_ofAlpha}.  The main result in this section is formulated as follows.
\begin{theorem}[Absorption at the collusive-enabling price]\label{prop:QConvergence_T} 
Suppose that Assumptions \ref{assumption_deltai} and \ref{assumption_learningrate} hold, firms play with Algorithm \ref{Algo_Q_boundedExperimentation} in the stochastic setting of Section \ref{sec:Model}, and for each $i\in[n]$, $p\in\A\setminus \{p^C\}$ and $\s\in \{\p_{T-1},\p^C\}$: 
\begin{itemize}
    \item[(i)] $Q_{T}^i(\s,p^C)>Q_T^i(\s,p)$;
    \item[(ii)] $\pi^i(\p^C)\geq (1-\delta_i)Q_T^i(\p^C,p)$.
\end{itemize}
Then, for any initial price profile $\p_0 \in \A^n$ and for all $t \geq T$, each firm $i \in [n]$ chooses $p_t^i = p^C$. 
Moreover, 
\begin{align}
\label{eqn_Q_epsilon*}
&Q^{i*}(\s,p):=\lim_{t\to \infty} Q_t^i(\s,p)=\\
&\Scale[0.94]{\nonumber
\begin{cases}
\pi^i(\p^C)/(1-\delta_i) & \textnormal{ if } (\s,p)=(\p^C,p^C),\\
 (1-
 \alpha_T)Q_{T}^i(\p_{T-1},p^C)+\alpha_T\left[\pi^i(\p^C)+\delta_iQ_T^i(\p^C,p^C)\right] & \textnormal{ if } (\s,p)=(\p_{T-1},p^C) 
 \textnormal{ and } \p_{T-1}\neq\p^C,\\
 Q_{T}^i(\s,p) &  \textnormal{ otherwise. }
\end{cases}}
\end{align}
\end{theorem}

The proof of Theorem~\ref{prop:QConvergence_T} is provided in Appendix~\ref{Appendix_RLalgo}. The core idea is as follows. First, Algorithm~\ref{Algo_Q_boundedExperimentation}, together with condition~(i) of the theorem, ensures that the $Q$-learning algorithm selects $p_T^i = p^C$ for each $i \in [n]$ and for all initial price profiles $\p_0 \in \A^n$. Then, condition~(ii) guarantees that firms continue to choose $p_{T+1}^i = p^C$ at time $T+1$. Finally, Assumption~\ref{assumption_learningrate} ensures convergence of the $Q$-values, as formalized in equation~\eqref{eqn_Q_epsilon*}.

Theorem~\ref{prop:QConvergence_T} answers question~1 from the introduction. Condition~(i) is a ranking condition: at the two states that matter for play from time $T$ on, the collusive-enabling price has the highest learned value. Condition~(ii) is a level condition: rearranged as $Q^i_T(\p^C,p)\leq \pi^i(\p^C)/(1-\delta_i)$, it says that the value the firm has learned for any deviation at the collusive state does not exceed the value of charging $p^C$ forever. Neither condition says that deviating is unprofitable in the underlying game; both are statements about what the algorithm has learned by time $T$. Once they hold, every firm charges $p^C$ at every subsequent time, the entries $Q^i(\p^C,p)$ with $p\neq p^C$ are never revised again, and what the algorithm has not learned about deviations is never corrected.

Because both conditions concern $Q_T$, which the learning process itself produces, Theorem~\ref{prop:QConvergence_T} is a stability result: it explains why supracompetitive prices may be persistent, as observed in the simulations of \cite{calvano2020artificial} and \cite{chica2024artificial}. 

\begin{remark}[Decaying exploration]\label{rem:decaying}
The deadline in Algorithm~\ref{Algo_Q_boundedExperimentation} is not essential. If $\varepsilon$-greedy exploration decays fast enough instead of stopping, as in \cite{calvano2020artificial}, it ends almost surely at a finite random time, and from that time on Theorem~\ref{prop:QConvergence_T} holds under the same two conditions (Proposition~\ref{prop:epsgreedy} in Appendix~\ref{app:anyprice}). Long exploration changes how the conditions come about rather than what follows from them: once the values of deviations have been revised many times, condition~(ii) is a statement about learned punishment, the channel of Section~\ref{sec:supporting}, rather than about what was never learned, the channel of Theorem~\ref{thm:cleanslate} below (Remark~\ref{rem:longexploration}).
\end{remark}


\section{\texorpdfstring{Collusion Emergence under $Q$-Learning}{}}
\label{sec:emergence}

Theorem~\ref{prop:QConvergence_T} is a statement about the action values at the deadline: whenever $Q_T$ ranks $p^C$ first at the relevant states and the level condition holds, play is absorbed at $\p^C$. However, it does not answer the central question of when 
such play happens. This section answers it. Proposition~\ref{prop:capture} shows that the conditions of Theorem~\ref{prop:QConvergence_T} can be secured well before the deadline, so that absorption is decided during experimentation. We then bound the probability of absorption through two independent channels. The first is the initialization of the $Q$-table: Proposition~\ref{prop:random} bounds the probability of absorption when the initial table is drawn at random, and Lemma~\ref{cor:basin} in the Appendix exhibits an open set of initial tables from which absorption is certain. The second is exploration alone: Theorem~\ref{thm:cleanslate} starts from the uninformative table $Q_1\equiv 0$ and bounds the probability of absorption from below in closed form.

Throughout this section Assumption~\ref{assumption_deltai} holds, so the state at time $t$ is the previous price profile, $\s_t=\p_{t-1}$, and firms follow Algorithm~\ref{Algo_Q_boundedExperimentation}. We use the following constants:

\begin{equation}\label{def:emergence_constants}
\begin{split}
&\pibar:=\max_{i\in[n]}\max_{\p\in\A^n}\pi^i(\p), \quad
\dbar:=\max_{i\in[n]}\delta_i<1, \quad  K:=\max\Big\{\norminf{Q_1},\ \tfrac{\pibar}{1-\dbar}\Big\}, \\
&
\kappa:=\pibar+(1+\dbar)K, \quad A_T:=\sum_{t=1}^{T-1}\alpha_t, \quad
R_\tau:=\kappa\sum_{k=\tau}^{T-1}\alpha_k\quad(1\leq\tau\leq T),
\end{split}
\end{equation}
where $|\A|=m+1$, $\pi^i\geq 0$ and $\delta_i\in(0,1)$. Here $\norminf{\cdot}$ is the maximum absolute entry over all firms, states and prices. The constant $K$ bounds every action value at every time (Lemma~\ref{lem:bound}, Appendix \ref{sec:lemmas}), and $\kappa$ bounds the change that a single update can make to any entry of the $Q$-table. Since each firm revises exactly one entry per period, $R_\tau$ bounds the total change of any entry between time $\tau$ and the deadline (Lemma~\ref{lem:budget}, Appendix \ref{sec:lemmas}): it is the learning budget that remains after time $\tau$. Finally, $A_T$ is the total learning rate spent during experimentation, so that $R_1=\kappa A_T$.

\begin{proposition}[Absorption before time $T$]\label{prop:capture} Suppose that Assumptions \ref{assumption_deltai} and \ref{assumption_learningrate} hold. Let $1\leq \tau\leq T$ and suppose that at time $\tau$, for every firm $i$ and every $p\in\A\setminus\{p^C\}$
\begin{enumerate}
\item[(i)] $Q^i_\tau(\s,p^C)-Q^i_\tau(\s,p)>2R_\tau$ for every $\s\in\A^n$;
\item[(ii)] $\pi^i(\p^C)\ \geq\ (1-\delta_i)\big(Q^i_\tau(\p^C,p)+R_\tau\big)$.
\end{enumerate}
Then conditions (i) and (ii) of Theorem~\ref{prop:QConvergence_T} hold at time $T$, and therefore $\p_t=\p^C$ for every $t\geq T$.  
\end{proposition}

Proposition~\ref{prop:capture} says that absorption can be decided before the deadline. If, at some time $\tau<T$, the margin of $p^C$ over every alternative exceeds twice the remaining learning budget at every state, and the level condition holds with room $R_\tau$ to spare, then no sequence of softmax draws on $[\tau,T)$ can undo either property, whatever the profile $\p_{T-1}$ turns out to be. The margin is required at every state precisely because $\p_{T-1}$ is not known at time $\tau$. Since $R_T=0$ (the sum defining it is empty), at $\tau=T$ conditions (i) and (ii) of Proposition~\ref{prop:capture} are exactly the hypotheses of Theorem~\ref{prop:QConvergence_T}, strengthened from the two states $\{\p_{T-1},\p^C\}$ to all states; the content of the proposition lies
in the times $\tau<T$.

\begin{corollary}[Capture time]\label{cor:stopping}
Define $$\tau^{*}:=\min\{\tau\in\{1,\dots,T\}\text{ s.t. (i) and (ii) of Proposition \ref{prop:capture} hold at }\tau\}$$ with
$\min\emptyset=\infty$. Then
$\{\tau^{*}\leq T\}\subseteq\{\p_t=\p^C\ \text{for all}\ t\geq T\}$. The event
$\{\tau^{*}=\tau_0\}$ is determined by $(Q_1,\dots,Q_{\tau_0})$.
\end{corollary}

Corollary~\ref{cor:stopping} turns the question of emergence into the probability of a hitting event: $\Prob(\p_t=\p^C \text{ for all } t\geq T)\geq\Prob(\tau^{*}\leq T)$. Because the event $\{\tau^{*}=\tau_0\}$ depends only on the $Q$-tables up to time $\tau_0$, absorption can be certified from the learned table as soon as it is decided, without waiting for the deadline.

The first channel through which absorption emerges is the initial table. If the initial values happen to rank $p^C$ first at every state by a margin larger than $2R_1$, Proposition~\ref{prop:capture} applies at $\tau=1$, and absorption follows whatever the subsequent exploration. The next result quantifies this channel when the initial table is drawn at random.

\begin{proposition}[Supracompetitive Prices under Random Initialization]\label{prop:random}
Suppose that Assumptions \ref{assumption_deltai} and \ref{assumption_learningrate} hold and firms follow Algorithm~\ref{Algo_Q_boundedExperimentation}. Let the $N=n(m+1)^{n+1}$ entries of $Q_1$ be drawn i.i.d.\ uniformly on $[0,K_0]$, and suppose
\begin{equation}
2R_1\ <\ K_0\ \le\ \frac{\pibar}{1-\dbar}
\ \textnormal{ and } \ 
K_0\ <\ \min_{i\in[n]}\frac{\pi^i(\p^C)}{1-\delta_i}-R_1,
\label{condition_start}
\end{equation}
and let $\rho:=\frac{1}{m+1}\Big(1-\frac{2R_1}{K_0}\Big)^{m+1} > 0 $. Then
\begin{equation}
\label{eqn:probability_bound}
\Prob\big(\p_t=\p^C\ \text{for all}\ t\geq T\big)\ \geq\ \rho^{\,n(m+1)^{n}}.
\end{equation}
\end{proposition}

The two conditions inside $\eqref{condition_start}$ have a simple reading: under $K_0\le\pibar/(1-\dbar)$, the constants $K=\pibar/(1-\dbar)$, $\kappa=2\pibar/(1-\dbar)$ and $R_1=2\pibar A_T/(1-\dbar)$ are deterministic. The lower bound $2R_1<K_0$ then requires $A_T<1/4$: the learning rate spent during experimentation must be small relative to the spread of the initial values (otherwise exploration can undo any initial ranking). The upper bound on $K_0$ requires that initial values do not exceed what permanent collusion is worth, and thus the level condition, (ii) from Theorem \ref{prop:QConvergence_T}, holds from the start.\footnote{Lemma~\ref{cor:basin} in the Appendix shows that $A_T<1/4$ is also necessary for any initialization in $[0,\pibar/(1-\dbar)]^N$ to satisfy the conditions of Proposition~\ref{prop:capture} at $\tau=1$. Moreover, $A_T<1/4$ combined with $R_1<\min_i\pi^i(\p^C)/(1-\delta_i)$ is sufficient to have an open set of $Q_1$-initializations inside $[0,\pibar/(1-\dbar)]^N$ that guarantee absorption.}

The bound $\rho^{\,n(m+1)^n}$ itself is crude: it requires the margin condition ((i) in Proposition \ref{prop:capture}) at all $n(m+1)^n$ pairs of firms and states simultaneously, and it deteriorates quickly with the number of firms and prices. Its role is to show that the initialization channel alone produces collusion with explicit positive probability.

One may object that the initialization channel merely moves the assumption from $Q_T$ to $Q_1$. The next result answers this objection. It starts from the uninformative table $Q_1\equiv 0$, under which no price is initially preferred to any other, and lets the softmax draws of the experimentation phase do all of the work.

\begin{theorem}[Emergence from an uninformative initialization]\label{thm:cleanslate}
Let Assumptions~\ref{assumption_deltai} and \ref{assumption_learningrate} hold, let $Q^i_1\equiv0$ for every $i$, and let firms follow Algorithm~\ref{Algo_Q_boundedExperimentation}. Define
\[
y^i_1:=y^i_2:=0,
\qquad
y^i_t:=\frac{\pi^i(\p^C)}{1-\delta_i}
\Big[1-\prod_{l=2}^{t-1}\big(1-\alpha_l(1-\delta_i)\big)\Big]\quad (t\geq3).
\]
Then
\[
\Prob\big(\p_t=\p^C\ \text{for all}\ t\geq T\big)
\ \geq\ \prod_{t=1}^{T-1}\ \prod_{i=1}^{n}\ \frac{1}{1+m\,e^{-y^i_t/\beta_t}}\ >\ 0 .
\]
If moreover $\p_0\neq\p^C$, the right-hand side is exactly the probability of the event that every firm charges $p^C$ at every time $t\in\{1,\dots,T-1\}$.
\end{theorem}

The proof, in Appendix~\ref{sec:proofs_emergence}, follows a single path: every firm draws $p^C$ at every time before the deadline. Along this path the entries $Q^i_t(\s,p)$ with $p\neq p^C$ are never revised and remain at zero, while $Q^i_t(\p^C,p^C)$ grows deterministically to at least $y^i_t>0$. At the deadline both conditions of Theorem~\ref{prop:QConvergence_T} therefore hold, and the theorem applies.

We highlight two features of this mechanism. First, collusion here requires neither communication nor punishment. Condition~(ii) of Theorem~\ref{prop:QConvergence_T} holds because the value of undercutting, $Q^i_T(\p^C,p)$, has never been learned: the algorithm sustains $p^C$ not because a deviation would be punished, but because it has never tried one. After the deadline the only entry at the collusive state that is revised is $(\p^C,p^C)$, so this ignorance is never corrected. Second, the bound is explicit and has transparent comparative statics. Each factor $(1+m\,e^{-y^i_t/\beta_t})^{-1}$ decreases with the temperature $\beta_t$ and increases with $y^i_t$, which in turn increases with the collusive profit $\pi^i(\p^C)$ and with the learning rates: less exploration, faster learning and larger collusive profits all raise the bound. Like the bound in Proposition~\ref{prop:random}, this bound is crude. It counts only one of the paths that can lead to absorption at $\p^C$, and since its factors for $t=1,2$ equal $1/(m+1)$ for every firm, it is at most $(m+1)^{-2n}$. It establishes that lock-in has positive probability; it does not estimate how often lock-in occurs.


\section{Strategies Supporting Supracompetitive Prices}\label{sec:supporting}

We now turn to questions~3 and~4: which policies support the absorbed behavior, and are they equilibria? The three propositions below assume, in addition to conditions on $Q_T$, that the final pre-deadline profile is the collusive one, $\p_{T-1}=\p^C$. This event has positive probability: starting from $Q_1\equiv 0$, the proof of Theorem~\ref{thm:cleanslate} bounds its probability from below by an explicit positive number. 

Under their hypotheses, the realized path from time $T$ onward is $\p^C$ in every period. The prescriptions at states $\s\neq\p^C$, namely reversion to $\p^*$ in Proposition~\ref{prop:grimtriggerQ} and the price ladder in Proposition~\ref{prop:increasing_Q}, are counterfactuals that the learned table carries but that play never triggers: the punishment that sustains the collusive-enabling price is never executed. 
These prescriptions are hypotheses on $Q_T$, and they do not hold in Theorem~\ref{thm:cleanslate}, where $Q^i_T(\s,\cdot)$ is still identically zero at every state $\s\notin\{\p_0,\p^C\}$.
We first formulate the following proposition studying ``naive collusion'', that is, collusion without any punishment and reward behavior.  
It uses the notation $\w^* = (w^{i*})_{i=1}^n$ for the strategy induced by $(Q^{i*})_{i=1}^n$ defined in \eqref{eqn_Q_epsilon*} (see \eqref{def:induced_strategy_Qil} for the definition of induced strategies). Throughout the rest of the section, we assume that firms play with the induced strategies $\w^*$ in the stochastic setting of Section \ref{sec:Model}.

\begin{proposition}[Naive Collusion]\label{prop:naiveCollusionQ} 
Suppose that Assumptions \ref{assumption_deltai} and \ref{assumption_learningrate} hold and $\p_{T-1}=\p^C$. Furthermore, for each $i\in[n]$ and $p\in\A\setminus \{p^C\}$,
\begin{itemize}
    \item[(i)] $Q_{T}^i(\s,p^C)>Q_T^i(\s,p)$ for every $\s\in \A^n$;
    \item[(ii)] $\pi^i(\p^C)>(1-\delta_i)Q_T^i(\p^C,p)$.
\end{itemize}
Then, $\p_t = \p^C$ for every $t\geq T$, and $\w^*(\s) = \p^C$ for every $\s\in\A^n$.  Moreover, $\w^*$ is a Nash equilibrium from time $t = 1$ if and only if $\p^C$ is a Nash equilibrium of the one-stage game $(\pi^i(\cdot))_{i=1}^n$.
\end{proposition}

Proposition \ref{prop:naiveCollusionQ} shows sufficient conditions under which the strategies induced by $(Q^{i*})_{i=1}^n$ never display punishment and reward behavior. Indeed, there is no mechanism to punish a firm that deviates from $p^C$. Instead, firms naively play by always choosing the collusive-enabling price. Therefore, this proposition implies that supracompetitive prices are not always the result of punishment and reward behavior. The final statement of Proposition \ref{prop:naiveCollusionQ} implies that unless $\p^C$ is a Nash equilibrium of the one-stage game $(\pi^i(\cdot))_{i=1}^n$, $\w^*$ cannot be a Nash equilibrium from time $t=1$. However, in general, $\p^C$ is not a Nash equilibrium in most models of interest, such as traditional Bertrand competition or platform competition in two sided markets (see, e.g., \cite{tirole1988theory}, \cite{dewenter2011semi} and \cite{chica2025competition}). Finally, we note that Assumptions~(i) and~(ii) in Proposition~\ref{prop:naiveCollusionQ} imply conditions~(i) and~(ii) in Theorem~\ref{prop:QConvergence_T}. This implication is intuitive: sustaining supracompetitive prices by naively choosing $\p^C$ in all states imposes a stricter requirement than merely achieving such prices in the long run.

The following assumption is only used in Propositions \ref{prop:grimtriggerQ}  and  \ref{prop:grimTriger}. It uses the quantity 
$$\pi^{m,i}:=\max_{p^i\in\A\setminus\{p^C\}} \pi^i(p^i,(\p^C)^{-i}).$$
\begin{assumption}\label{assumption_deltai_simple}
For each $i\in[n]$, \ $\pi^{m,i}- \pi^i(\p^C)\leq \delta_i\big(\pi^{m,i} - \pi^i(\p^*)\big).$\footnote{The quantity $\pi^{m,i}$ is the best payoff firm $i$ can obtain by deviating from $p^C$ when all other firms charge $p^C$. When such a deviation is profitable, $\pi^{m,i}>\pi^i(\p^C)>\pi^i(\p^*)$, and Assumption \ref{assumption_deltai_simple} is the lower bound $\delta_i\geq(\pi^{m,i}- \pi^i(\p^C))/(\pi^{m,i} - \pi^i(\p^*))$, which lies in $(0,1)$: it is the ratio of the distance between $\pi^{m,i}$ and the collusive-enabling payoff, $\pi^i(\p^C)$, to the distance between $\pi^{m,i}$ and the competition payoff $\pi^i(\p^*)$. When no such deviation is profitable, the assumption holds for every $\delta_i\in(0,1)$. }
\end{assumption}

Next, we review the grim trigger strategy. In our setting (under Assumption~\ref{assumption_deltai}), the grim trigger strategy \citep{friedman1971non} is a policy in which a firm cooperates by choosing the price $p^C$ as long as all other firms chose $p^C$ in the previous stage. If, on the other hand, at least one firm deviated in the previous stage by choosing a price $p^i \neq p^C$, all firms permanently revert to $p^*$.\footnote{Since $\p^*$ is a Nash equilibrium, firm $i$ has no incentive to deviate from the punishment path (a fact we verify in Proposition \ref{prop:grimTriger} in Appendix \ref{sec:Repeated_Bertrand}). After deviating, firm $i$ is punished by receiving $\pi^i(\p^*)$ forever, without gaining any competitive advantage, since all firms revert to the same competitive price.} The grim trigger strategy can be expressed as the following one-memory policy:
$$\ssigma^f=(\sigma_0^f,\sigma_1^f), \text{ where } \sigma_0^f(p^C)=1, \sigma_1^{f}(p^C|\p^C)=1 \text{ and } \forall \p_0\in\A^n, \p_0\neq \p^C, \sigma_1^{f}(p^*|\p_0)=1.$$ 

The following proposition shows sufficient conditions under which the strategies induced by $(Q^{i*})_{i=1}^n$ display punishment and reward behavior in a grim trigger fashion.

\begin{proposition}[Grim Trigger Collusion]\label{prop:grimtriggerQ} 
Suppose that Assumptions \ref{assumption_deltai} and \ref{assumption_learningrate} hold and $\p_{T-1}=\p^C$. Furthermore, for each $i\in[n]$,
\begin{itemize}
    \item[(i)] $Q^i_T(\p^C,p^C)>Q^i_T(\p^C,p)$ for every $p\neq p^C$;
    \item[(ii)] $\pi^i(\p^C)>(1-\delta_i)Q^i_T(\p^C,p)$ for every $p\neq p^C$;
    \item[(iii)] $Q^i_T(\s,p^{*})>Q^i_T(\s,p)$ for every $\s\neq\p^C$ and every $p\neq p^{*}$.
\end{itemize}
Then $\p_t=\p^C$ for every $t\geq T$ and, 
\begin{equation}
    \label{eqn:argmaxQ_grim}
    \w^*(\s) = \begin{cases}
        \p^C & \s = \p^C,\\
        \p^* & \s\neq \p^C.
    \end{cases}
\end{equation}
Moreover, under Assumption \ref{assumption_deltai_simple}, $\w^*$ together with $\sigma_0^f$ is an SPE.
\end{proposition}

Proposition~\ref{prop:grimtriggerQ} provides sufficient conditions under which the strategies induced by $(Q^{i*})_{i=1}^n$ coincide with the grim trigger strategies beginning at time~$t = 1$ (see Section~\ref{sec:Repeated_Bertrand}). By definition, these strategies implement punishment-and-reward behavior: firms continue to collude (i.e., choose $p^C$) as long as all firms selected $p^C$ in the previous stage; otherwise, they permanently revert to the competitive price $p^*$. Under Assumption~\ref{assumption_deltai}, we have $\pi^i(\p^C) > \pi^i(\p^*)$, so firms are strictly better off by sustaining collusion indefinitely. 

Finally, we note that Assumption~(iii) in Proposition~\ref{prop:grimtriggerQ} is not in conflict with the assumptions of Theorem~\ref{prop:QConvergence_T}, which, since $\p_{T-1}=\p^C$, only concern the state $\p^C$. Therefore, whenever $Q_T$ satisfies conditions~(i)--(iii), the greedy policy that the firms follow from time $T$ on is the grim trigger strategy~\eqref{eqn:argmaxQ_grim}.

Punishment-and-reward schemes need not be limited to grim trigger strategies. In fact, recent numerical studies \citep{calvano2020artificial, klein2021autonomous, chica2024artificial} show that algorithms can learn more sophisticated forms of collusive behavior. For example, firms may learn to gradually raise prices over time until reaching the collusive-enabling price $\p^C$, while using the competitive price $p^*$ as a threat in response to unilateral deviations. 
Proposition~\ref{prop:increasing_Q} provides sufficient conditions under which the strategies induced by $(Q^{i*})_{i=1}^n$ replicate this type of increasing-price behavior.

\begin{proposition}[Increasing Strategies]\label{prop:increasing_Q} 
Suppose that Assumptions \ref{assumption_deltai} and \ref{assumption_learningrate} hold, $p^*<p^C$ and $\p_{T-1}=\p^C$. Let $p^{(0)}<p^{(1)}<\dots<p^{(k+1)}$ be prices in $\A$ with $p^{(0)}=p^*$ and $p^{(k+1)}=p^C$, and write $\p^{(l)}:=(p^{(l)},\dots,p^{(l)})$. Suppose that for every $i$,
\begin{itemize}
\item[(i)] $Q^i_T(\p^C,p^C)>Q^i_T(\p^C,p)$ and $\pi^i(\p^C)>(1-\delta_i)Q^i_T(\p^C,p)$, for every
$p\neq p^C$;
\item[(ii)] $Q^i_T(\p^{(l)},p^{(l+1)})>Q^i_T(\p^{(l)},p)$ for every $l\in\{0,\dots,k\}$ and every
$p\neq p^{(l+1)}$;
\item[(iii)] $Q^i_T(\s,p^*)>Q^i_T(\s,p)$ for every $\s\notin\{\p^{(l)}\}_{l=0}^{k+1}$ and every
$p\neq p^*$.
\end{itemize}

Then $\p_t=\p^C$ for every $t\geq T$ and
\begin{equation}
\label{eqn:argmaxQ_increasing}
\wstar(\s)=\begin{cases}
\p^C, & \s=\p^C,\\
\p^{(l+1)}, & \s=\p^{(l)},\ l\in\{0,\dots,k\},\\
\p^*, & \text{otherwise.}
\end{cases}
\end{equation}
\end{proposition}

Proposition \ref{prop:increasing_Q} shows sufficient conditions under which the strategies induced by $(Q^{i*})_{i=1}^n$ display an increasing behavior towards the collusive-enabling price $p^C$. Suppose that firms start at the Nash equilibrium price $\p^*$, following \eqref{eqn:argmaxQ_increasing}, firms will choose $\p^{(1)}$ in the next stage, and progressively increase their prices until reaching $\p^{(k+1)} = \p^C$. After any unilateral deviation, firms go back to the Nash equilibrium price and the increasing pattern follows again.

\section{Policy Implications and Extensions}
\label{sec:policy_extensions}

This section connects our main results to the ongoing policy debate in algorithmic price collusion (Section~\ref{sec:policy}) and shows some extensions for some of our previous results (Section~\ref{sec:extensions}).

\subsection{Policy implications}\label{sec:policy}
Under U.S.~legal standards, collusion requires an agreement between multiple parties. Thus, tacit coordination achieved without communication is generally lawful, and so is supracompetitive pricing by autonomous algorithms that never communicate \citep{ezrachi2017artificial,harrington2018developing}. The algorithmic pricing cases brought to US courts so far fit the ``agreement paradigm''. In \emph{United States v.\ Topkins}, sellers of posters on the Amazon marketplace agreed on prices and wrote code instructing their pricing software to implement the agreement \citep{doj2015topkins}. Similarly, in \emph{United States v.\ RealPage}, the Department of Justice alleged that landlords coordinated rents through revenue-management software that pooled competitors' nonpublic data \citep{doj2024realpage}.

Our results describe a situation that falls outside this paradigm. In Theorem~\ref{thm:cleanslate}, each firm runs its own algorithm, observes only its own profits and the public price profile, and starts from a table that carries no information. There is no communication, no shared data and no instruction to collude, yet supracompetitive prices emerge with explicitly bounded positive probability and persist forever. Moreover, the learned behavior does not necessarily involve a reward-punishment scheme. \cite{harrington2018developing} argues that collusion should be identified by the use of strategies that embody such a scheme, rather than by supracompetitive prices. Likewise, the simulations of \cite{calvano2020artificial} document collusion by forcing a deviation and observing the rivals' retaliation. Proposition~\ref{prop:naiveCollusionQ} shows that $\p^C$ can be sustained by naive collusion, a profile with no punishment at all. On the event in Theorem~\ref{thm:cleanslate}, the value of undercutting at the collusive state is never revised from its initial value of zero. Tests that look for punishment can therefore miss this form of supracompetitive pricing.

Empirical work shows that algorithmic pricing can raise prices in practice. In Germany's retail gasoline market, \cite{assad2024algorithmic} find that the adoption of algorithmic pricing software raised margins only in markets with more than one station. This study cannot observe the learning that takes place inside the algorithms; our results supply one mechanism consistent with the pattern documented by \cite{assad2024algorithmic}. Absorption at $\p^C$ requires conditions~(i) and~(ii) of Theorem~\ref{prop:QConvergence_T} to hold for every firm: a single firm whose pricing rule does not rank $p^C$ first at $\p^C$ moves the market away from the collusive profile.

Our results also identify the features of an algorithm's design that govern the risk of collusive lock-in. The bound in Theorem~\ref{thm:cleanslate} rises as exploration falls and as learning rates and collusive profits rise. Lemma~\ref{cor:basin} shows that some initializations make lock-in certain, and Proposition~\ref{prop:random} that random initializations produce it with positive probability, without any instruction to collude. When pricing is outsourced, these design choices are made by a third-party vendor serving competing firms, and \cite{harrington2022effect} shows that the vendor's profit-maximizing design increases industry profit and lowers consumer welfare. Two findings in the literature align closely with our mechanism. In a simple market setting, \cite{asker2022artificial} show that algorithms that learn only from the actions they actually take (asynchronous learning) can sustain prices close to monopoly levels, whereas algorithms that also update the value of actions not taken, using counterfactual profits (synchronous learning), price competitively. In our model, lock-in rests precisely on the deviation entries $Q^i(\p^C,p)$ never being updated. Finally, conditions~(i) and~(ii) of Theorem~\ref{prop:QConvergence_T} and the margin conditions of Proposition~\ref{prop:capture} are inequalities on the learned table, which makes them auditable in principle; \cite{hartline2024regulation} propose a statistical audit of pricing algorithms for plausible non-collusion that a regulator could apply to the data the algorithms collect.

\subsection{Extensions}\label{sec:extensions}

\textbf{Absorption and emergence at any price.} Theorems~\ref{prop:QConvergence_T} and~\ref{thm:cleanslate} are stated for the collusive-enabling price, but neither proof uses any property of $p^C$ other than, in Theorem~\ref{thm:cleanslate}, the positivity of its profits. Appendix~\ref{app:anyprice} states and proves the general versions. For any price $\hat p\in\A$, conditions~(i) and~(ii) of Theorem~\ref{prop:QConvergence_T} stated at $\hat p$ make play absorbing at $\hat\p$ (Proposition~\ref{prop:absorb_any}), and from $Q_1\equiv 0$ absorption at $\hat\p$ occurs with probability at least the bound of Theorem~\ref{thm:cleanslate} with $\pi^i(\p^C)$ replaced by $\pi^i(\hat\p)$, provided $\pi^i(\hat\p)>0$ for every firm (Theorem~\ref{thm:cleanslate_any}). Applied to $\hat p=p^*$, this shows that competitive lock-in also has positive probability whenever Nash profits are positive. Since $\pi^i(\p^C)>\pi^i(\p^*)$, the lower bound is strictly larger at $p^C$ than at $p^*$ when $T\geq 4$ (Corollary~\ref{cor:selection}). This compares the two lower bounds, not the two probabilities, and does not imply that collusive lock-in is the more likely outcome. It does show that both outcomes have positive probability from the same initial table, so which one is realized depends on the exploration path, and an observed supracompetitive outcome may reflect that path rather than a structural feature of the market. When Nash profits are zero, as in Bertrand competition with constant marginal cost and the competitive price equal to that cost, the argument does not apply at $p^*$: along the competitive path every entry of the table remains at zero, and the ranking at the deadline is a tie.

\section{Conclusion}
\label{sec:conclusion}
This paper is motivated by recent simulation studies showing that $Q$-learning agents may learn to charge supracompetitive prices. To provide a theoretical explanation, we study a setting of stochastic games with bounded memory, where firms use $Q$-learning with bounded experimentation. We highlight our key findings:
\begin{enumerate}
    
    \item We show for the case of infinite repeated games that if a one-stage Nash equilibrium price and a collusive-enabling price exist, and the $Q$-function satisfies certain inequalities at the end of experimentation, then firms charge supracompetitive prices in the long run.

    \item We show that these inequalities are reached with positive and explicitly bounded probability, both from a randomly drawn initial table and from the uninformative table $Q_1\equiv 0$. In the latter case, exploration alone can produce collusive lock-in because firms never learn that undercutting pays. The same holds at any price with positive profits, so the long-run outcome is stochastically selected.
    
    \item We provide sufficient conditions under which these supracompetitive prices are generated by: (i) naive collusion, where firms always choose the collusive-enabling price; (ii) grim trigger strategies, which reward continued collusion and punish deviations; or (iii) increasing strategies, where firms gradually converge to the collusive-enabling price while using the Nash equilibrium price as a threat.
    
    \item Among the strategies supporting supracompetitive prices, we find that naive collusion cannot be an SPE unless the collusive-enabling price is a Nash equilibrium of the one-stage game, whereas grim trigger strategies can be.

\end{enumerate}

Future work may extend our emergence results to the case of unbounded experimentation, and we believe that stochastic games with bounded memory remain a promising framework for this direction.  Our results have two limitations. First, our emergence results bound the probability of collusive lock-in from below; they do not characterize the limiting distribution over absorbing profiles, which under persistent exploration and Robbins--Monro learning rates would require an analysis of the invariant sets of the associated mean differential equation. Second, our results describe absorption at a fixed price profile and the one-memory strategies that support it; they are silent on outcomes that are not absorbing, such as the Edgeworth-type price cycles that \cite{klein2021autonomous} documents for $Q$-learning under sequential pricing.

\section*{Declaration of generative AI use}
The authors used Claude (Anthropic) to suggest improvements in
language, clarity, and precision, including minor revisions to the
formulation of mathematical statements. The authors evaluated these
suggestions and incorporated selected revisions into the manuscript.
The mathematical results and proofs are the authors' own, and the
authors take full responsibility for the content of the manuscript.

\newpage
\appendix
\section{Proofs}\label{AppendixA}
\renewcommand{\theequation}{A\arabic{equation}}
\setcounter{equation}{0}
\input{Appendix}

\newpage
\clearpage
\bibliographystyle{apalike}
\bibliography{bibliography}
\addcontentsline{toc}{section}{References}

\end{document}

%% file: Appendix.tex
We formalize the definition of the conditional expected value of an infinite sum (Section~\ref{sect:formalizeEPsigma}), establish a relationship between $\Vonet$ and $\Vone$ (Section~\ref{sec:relation_vi}), prove existence of one-memory SPEs (Section~\ref{sec:NashT1}), and apply Algorithm~\ref{algo:HowtoCheckNash} to grim trigger strategies (Section~\ref{sec:Repeated_Bertrand}). We then prove Propositions~\ref{prop:QfixedisV1} and~\ref{prop:Sufficient_Q_inducesNash}, establish the properties of $\alpha(\delta)$ that yield Lemma~\ref{lem:alphadelta} (Section~\ref{sect:properties_ofAlpha}), and prove Theorem~\ref{prop:QConvergence_T} and Propositions~\ref{prop:naiveCollusionQ}--\ref{prop:increasing_Q}. Section~\ref{sec:lemmas} proves two technical lemmas, and Section~\ref{sec:proofs_emergence} proves the results of Section~\ref{sec:emergence}. Appendix~\ref{app:anyprice} extends Theorems~\ref{prop:QConvergence_T} and~\ref{thm:cleanslate} to an arbitrary price and Theorem~\ref{prop:QConvergence_T} to decaying $\varepsilon$-greedy exploration, and Appendix~\ref{Sec:FinksProof} rewrites the proof of Fink's theorem.

\subsection{Formalizing the Conditional Expected Value of an infinite Sum}\label{sect:formalizeEPsigma}

For $\ssigma_t\in \SSigma_t$, $\p_{t-1}\in\A^n$, $s_t\in\S$, and $g:\A^n\times \S\rightarrow \mathbb{R}$, we define  
\begin{equation}\label{not_Esigma1_g}
    \E_{\ssigma_t}\left[g(\p, s)|\p_{t-1},s_t\right]:= \sum_{\p_t\in\A^n}\sigma_t(\p_t|\p_{t-1},s_t)g(\p_t,s_t).
\end{equation}  
For $\ssigma =(\ssigma_t)_{t\geqq 0} \in \SSigma$, $s_0\in\S$, $\P$ defined in \eqref{eq:def_P}, and $g_t:\A^n\times \S\rightarrow \mathbb{R}$, $t\geqq 0$, we define 
$$\E_{\ssigma,\P}\left[\sum_{t=0}^\infty g_t(\p_t, s_t)|s_0\right]:=\lim_{T\to\infty}\E_{\ssigma,\P}\left[\sum_{t=0}^T g_t(\p_t, s_t)|s_0\right],$$
whenever the limit exists, where for each $T\geqq 1$,
\begin{equation}\label{not_Esigma_P_g}
\begin{split}
    &\E_{\ssigma,\P}\left[\sum_{t=0}^T g_t(\p_t, s_t)|s_0\right]\\
    &= \sum_{\p_0\in\A^n}\sigma_0(\p_0|s_0)\left\{g_0(\p_0,s_0)+\sum_{s_1\in \S}\P(s_1|\p_0,s_0)\E_{(\ssigma_t)_{t\geqq 1}, \P}\left[\sum_{t=1}^T g_t(\p_t, s_t)|\p_0,s_1\right]\right\}
\end{split}
\end{equation}  
and for each $1\leq k \leq T-1$ 
\begin{equation}\label{not_Esigma_P_g2}
\begin{split}
&\E_{(\ssigma_t)_{t\geqq k}, \P}\left[\sum_{t=k}^T g_t(\p_t, s_t)|\p_{k-1},s_k\right]
=    \sum_{\p_k\in \A^n}\sigma_k(\p_k|\p_{k-1},s_k) g_k(\p_k, s_k) \\
   &+ \sum_{\p_k\in \A^n}\sigma_k(\p_k|\p_{k-1},s_k)\sum_{s_{k+1}\in \S}\P(s_{k+1}|\p_k,s_k)\E_{(\ssigma_t)_{t\geqq k+1}, \P}\left[\sum_{t=k+1}^T g_t(\p_t, s_t)|\p_k,s_{k+1}\right] .
\end{split}
\end{equation}  

\subsection{\texorpdfstring{Further Clarification of $\Vonet$ and its Relationship with $\Vone$}{}}
\label{sec:relation_vi}

The following fundamental proposition formulates a Bellman Equation for $\Vonet$. We use it to interpret $\Vonet$ as a weighted sum of conditional expectations and to directly relate $\Vonet$ to $\Vone$.

\begin{proposition}[Lemma 1 of \cite{fink1964equilibrium}]\label{prop:BellmanEqV1} Let $i\in[n]$ and $\ssigma_1 = (\sigma^i_1, \ssigma^{-i}_1)\in \SSigma_1$. For each $(s_1,\p_0)\in \S\times \A^n$, $\Vonet(s_1, \p_0, \sigma^i_1|\ssigma^{-i}_1)$ satisfies the following Bellman Equation,
\begin{multline}
\label{Bellman_equationV1}
  \Vonet(s_1, \p_0, \sigma^i_1|\ssigma^{-i}_1)
=\\
\sum_{\p_1\in\A^n}\sigma_1(\p_1|\p_0,s_1) \left[ \pi^i(\p_1,s_1) +
    \delta_i\sum_{s_2\in\S}\P(s_2|\p_1,s_1)
    \Vonet(s_2, \p_1, \sigma^i_1|\ssigma^{-i}_1)\right].    
\end{multline}
Moreover, the system of $rM$ equations given by \eqref{Bellman_equationV1} has a unique solution in the $rM$ variables $\{\Vonet(s_1, \p_0, \sigma^i_1 | \ssigma^{-i}_1)\}_{(s_1, \p_0) \in \S \times \A^n}$.
\end{proposition}

Proposition \ref{prop:BellmanEqV1} offers a more tractable characterization of the conditional value function at $t=1$, transforming it from an infinite expectation in  \eqref{conditional_value_function_t1} into a finite recursive formula. Furthermore, it leads to an expression of the conditional value function as a weighted average over expected profits at a finite number of state-action pairs. Indeed, following the proof of \eqref{Bellman_equationV1} in Appendix \ref{sec:proof_prop_1}, one can notice that for each $(s_1,\p_0)\in \S\times \A^n$, $\Vonet(s_1, \p_0, \sigma^i_1|\ssigma^{-i}_1)$ is a weighted sum of the entries of $\E_{\ssigma_1}[\pi^i] := (\E_{\ssigma_1}[\pi^i|\p^1,s^1],\cdots,\E_{\ssigma_1}[\pi^i|\p^M,s^r])^T\in\mathbb{R}^{rM}$. Moreover, such weights are uniquely determined by the policies in $\ssigma_1$ and the transition probability $\P$ (see \eqref{Av1_system_proof} in Appendix \ref{sec:proof_prop_1}).

\subsection{Existence of One-Memory SPEs}
\label{sec:NashT1}
We establish the existence of a one-memory subgame perfect equilibrium (SPE) and formulate an algorithm for verifying whether a given profile satisfies this condition. Our analysis consists of three theorems. Theorem~\ref{Thm:Fink}, which corresponds to Theorem 2 of \cite{fink1964equilibrium}, establishes the existence of a fixed point of the $\Vone$ operator with desirable properties. Theorem~\ref{Thm:ExistenceNashT1} shows that such a fixed point corresponds to a Nash equilibrium from time $t=1$. Finally, Theorem~\ref{thm:existence_sigma_0} establishes the existence of a one-memory SPE.

\begin{theorem}[Existence of stationary points with special properties  \citep{fink1964equilibrium}]\label{Thm:Fink}
There exist $\ssigma_1^*\in\SSigma_1$ and $\v^*\in \R^{nrM}$ satisfying 
\begin{equation}\label{NashT1_1}
    \v^* = \Vone(\ssigma_1^{*},\ssigma_1^{*},\v^*) 
\end{equation}
and 
\begin{equation}\label{NashT1_2}
    \v^*_{i,s_1,\p_0} = \max_{\sigma^i_1\in \SSigma_1^i} \Vone(\ssigma_1^{*},\sigma_1^{i},\v^*)_{i,s_1,\p_0} \ \ \forall \ (i,s_1,\p_0)\in [n]\times \S\times \A^n.
\end{equation}
\end{theorem}

\begin{theorem}[Existence of Nash Equilibrium from time $t=1$]\label{Thm:ExistenceNashT1} Suppose that $\ssigma_1^*\in\SSigma_1$ and $\v^*\in \R^{nrM}$ satisfy \eqref{NashT1_1} and \eqref{NashT1_2}. Then, for each $i\in[n]$ and $(s_1,\p_0)\in\S\times\A^n$, 
\begin{equation}
\label{eqn:max_equality_V1_V1conditional}
   \max_{\sigma^i_1\in\SSigma_1^i} \Vone(\ssigma_1^{*},\sigma_1^{i},\v^*)_{i,s_1,\p_0} = \max_{\sigma^i_1\in\SSigma_1^i}\Vonet(s_1,\p_0,\sigma_1^i|\ssigma_1^{-i*}).
\end{equation}
Moreover, $\ssigma_1^*$ is a Nash equilibrium from time $t=1$. 
\end{theorem}

\begin{theorem}[Existence of the one-memory SPE]\label{thm:existence_sigma_0} If $\ssigma_1^*\in \SSigma_1$ is a Nash equilibrium from time $t=1$, then there exists $\ssigma_0^*\in\SSigma_0$ such that $\ssigma^{*}=(\ssigma_0^{*}, \ssigma_1^{*})$ is a one-memory SPE of the stochastic game.
\end{theorem}

This theory suggests the following three-step algorithm for proving that a given profile is a one-memory SPE. If one can only verify the first two steps of the algorithm, then the given profile is a Nash equilibrium from time $t=1$. 
We frequently use this algorithm in our proofs. 

\begin{algorithm}[Proving that a given profile is a one-memory SPE]\label{algo:HowtoCheckNash}
Let $(\ssigma_0^g,\ssigma_1^g)$ be a given one-memory strategy profile. The following algorithm guides the proof that this profile is an SPE. Its first two steps are used for proving a Nash equilibrium from time $t=1$. 
\begin{enumerate}
    \item Plug  $\ssigma_1^g$ into equation \eqref{NashT1_1} and solve it as a linear system with unknowns $v^g_{i,s_1,\p_0}$ for each $(i,s_1,\p_0)$-coordinate.
    \item Plug  $\v^g$ and $\ssigma_1^g$ into \eqref{NashT1_2} and show that  $\v^g$ is a fixed point of the operator $v_{i,s_1,\p_0}\mapsto\max_{\sigma_1^i\in\SSigma_1^i}\Vone(\ssigma_1^g,\sigma_1^i,\v)_{i,s_1,\p_0}$. 
    \item Show that $\ssigma_0^g$ satisfies \eqref{eqn:Nasht_0}.
\end{enumerate}

\end{algorithm}

\textbf{Comments on the Proofs of Theorems \ref{Thm:Fink} and  \ref{Thm:ExistenceNashT1}.} 
The proof of Theorem \ref{Thm:Fink} is due to  \cite{fink1964equilibrium}. For completeness, Appendix \ref{Sec:FinksProof} rewrites Fink's proof using our notation, while including many of the missing details in \cite{fink1964equilibrium}. We find it necessary to refer to the rewritten proof when establishing the theories of Sections \ref{sec:Repeated_Bertrand} and \ref{sec:RLalgo}. 

Although Theorem~\ref{Thm:Fink} establishes the existence of a fixed point of the $\Vone$ operator, it does not, by itself, imply the existence of a Nash equilibrium from time $t=1$. To prove Theorem \ref{Thm:ExistenceNashT1}, one must additionally verify the equality in \eqref{eqn:max_equality_V1_V1conditional} and then invoke Theorem~\ref{Thm:Fink}. 
It is important to note that the identity
\[
\Vone(\ssigma_1^*, \sigma_1^i, \v^*)_{i, s_1, \p_0} = \Vonet(s_1, \p_0, \sigma_1^i | \ssigma_1^{-i*})
\]
does not generally hold for all $\sigma_1^i \in \SSigma_1^i$, and should not be confused with the special case in \eqref{eqn:equality_V1_V1conditional}, where both sides refer to the same strategy profile. The validity of \eqref{eqn:max_equality_V1_V1conditional} must be established through a series of inequalities, as detailed below in Appendix~\ref{sec:proof_thm_2}.

\subsubsection{Proof of Theorem \ref{Thm:ExistenceNashT1}}
\label{sec:proof_thm_2}

Fix $i\in[n]$. We first express $\v^*$ through $\Vonet$. By \eqref{V1_function}, equation \eqref{NashT1_1} yields for each $(s_1,\p_0)\in\S\times\A^n$,
\begin{equation}\label{NashT1_1_coordinate}
v^*_{i,s_1,\p_0}= \sum_{\p_1\in\A^n}\sigma_1^{*}(\p_1|\p_0,s_1)\left[\pi^i(\p_1,s_1)+\delta_i\sum_{s_2\in\S}\P(s_2|\p_1,s_1)v^*_{i,s_2,\p_1}\right].
\end{equation}
By Proposition \ref{prop:BellmanEqV1}, the sequence $\{\Vonet(s_1,\p_0,\sigma_1^{i*}|\ssigma_1^{-i*})\}_{(s_1,\p_0)\in\S\times\A^n}$ is the unique solution to the system described by \eqref{NashT1_1_coordinate}. Therefore, for each $(s_1,\p_0)\in\S\times\A^n$ 
\begin{equation}\label{NashT1_1_coordinate2}
    v^*_{i,s_1,\p_0} = \Vonet(s_1,\p_0,\sigma_1^{i*}|\ssigma_1^{-i*}).
\end{equation} 

We now prove that for each $(s_1,\p_0)$-coordinate 
\begin{equation*}
   \underbrace{\max_{\sigma^i_1\in\SSigma_1^i} \Vone(\ssigma_1^{*},\sigma_1^{i},\v^*)_{i,s_1,\p_0}}_{LHS} = \underbrace{\max_{\sigma^i_1\in\SSigma_1^i}\Vonet(s_1,\p_0,\sigma_1^i|\ssigma_1^{-i*})}_{RHS}.
\end{equation*}
We first prove that LHS $\leq$ RHS and then that LHS $\geqq$ RHS.

\textbf{Proof of LHS $\leq$ RHS:}
Since $\v^*$ satisfies \eqref{NashT1_1} and \eqref{NashT1_2}, for each $(i,s_1,\p_0)$-coordinate 
\begin{equation}\label{proof:max_lemma12}
    \max_{\sigma^i_1} \Vone(\ssigma_1^{*},\sigma_1^{i},\v^*)_{i,s_1,\p_0} =  \Vone(\ssigma_1^{*},\ssigma_1^{*},\v^*)_{i,s_1,\p_0}.
\end{equation}
From \eqref{V1_function}, \eqref{NashT1_1_coordinate2} and Proposition \ref{prop:BellmanEqV1},
\begin{equation}\label{proof:max_lemma13}
\begin{split}
&\Vone(\ssigma_1^{*},\ssigma_1^{*},\v^*)_{i,s_1,\p_0}\\
&= \sum_{\p_1\in\A^n}\sigma_1^{*}(\p_1|\p_0,s_1)\left[\pi^i(\p_1,s_1)+\delta_i\sum_{s_2\in\S}\P(s_2|\p_1,s_1)\Vonet(s_2,\p_1,\sigma_1^{i*}|\ssigma_1^{-i*})\right]\\
&=\Vonet(s_1,\p_0,\sigma_1^{i*}|\ssigma_1^{-i*}).
\end{split}
\end{equation}
Clearly, \eqref{proof:max_lemma12} and \eqref{proof:max_lemma13} imply that  LHS$\leq$RHS. 

\textbf{Proof of LHS $\geqq$ RHS:}  
For each coordinate $(i,s_1,\p_0)$ and $\sigma_1^i\in\SSigma_1^i$, we estimate the following quantity, 
\begin{equation}\label{proof:max_lemma14}
    \begin{split}
        &\Vonet(s_1,\p_0,\sigma_1^{i}|\ssigma_1^{-i*})-\Vone(\ssigma_1^{*},\sigma_1^{i},\v^*)_{i,s_1,\p_0}=\sum_{\p_1\in\A^n}\sigma_1^{i}(p_1^i|\p_0,s_1)\sigma_1^{-i*}(\p_1^{-i}|\p_0,s_1) \\
        &\cdot  \delta_i\sum_{s_2\in\S}\P(s_2|\p_1,s_1)(\Vonet(s_2,\p_1,\sigma_1^{i}|\ssigma_1^{-i*})-\Vonet(s_2,\p_1,\sigma_1^{i*}|\ssigma_1^{-i*})).
    \end{split}
\end{equation}
We have used equation \eqref{NashT1_1}, which claims that 
$\v^* = \Vone(\ssigma_1^{*},\ssigma_1^{*},\v^*)$ and we have used equation \eqref{proof:max_lemma13}. We denote $\Delta\Vone^i(s_2,\p_1,\sigma_1^{i},\ssigma_1^{*}):=\Vonet(s_2,\p_1,\sigma_1^{i}|\ssigma_1^{-i*})-\Vonet(s_2,\p_1,\sigma_1^{i*}|\ssigma_1^{-i*})$. Applying first the fact that $-\Vonet(s_1,\p_0,\sigma_1^{i*}|\ssigma_1^{-i*})\leq-\Vone(\ssigma_1^{*},\sigma_1^{i},\v^*)_{i,s_1,\p_0}$ (which follows from \eqref{proof:max_lemma12} and \eqref{proof:max_lemma13}) and then  \eqref{proof:max_lemma14} result in
\begin{align}\label{proof:max_lemma15}
        &\Delta\Vone^i(s_1,\p_0,\sigma_1^{i},\ssigma_1^{*}) \notag \\&
        \leq \sum_{\p_1\in\A^n}\sigma_1^{i}(p_1^i|\p_0,s_1)\sigma_1^{-i*}(\p_1^{-i}|\p_0,s_1)  \delta_i\sum_{s_2\in\S}\P(s_2|\p_1,s_1)\max_{(s_2,\p_1)\in\S\times\A^n}\Delta\Vone^i(s_2,\p_1,\sigma_1^{i},\ssigma_1^{*}) \notag \\
        &= \delta_i\max_{(s_2,\p_1)\in\S\times\A^n}\Delta\Vone^i(s_2,\p_1,\sigma_1^{i},\ssigma_1^{*}).
\end{align}
Since \eqref{proof:max_lemma15} holds for all $(s_1,\p_0)\in\S\times \A^n$ and $\delta_i<1$
\begin{equation*}
\max_{(s_1,\p_0)\in\S\times\A^n}\Delta\Vone^i(s_1,\p_0,\sigma_1^{i},\ssigma_1^{*})\leq 0.
\end{equation*}
That is, $\Vonet(s_1,\p_0,\sigma_1^{i}|\ssigma_1^{-i*})\leq \Vonet(s_1,\p_0,\sigma_1^{i*}|\ssigma_1^{-i*})$ for each $(s_1,\p_0)\in\S\times \A^n$ and $\sigma_1^i\in\SSigma_1^i$. We thus conclude that LHS $\geqq$ RHS. 

Lastly, we show that $\ssigma_1^*$ is a Nash equilibrium from time $t=1$. By \eqref{NashT1_2} and \eqref{NashT1_1_coordinate2}, 
\begin{equation}\label{NashT1_1_coordinate3}
\begin{split}
    &\Vonet(s_1,\p_0,\sigma_1^{i*}|\ssigma_1^{-i*}) = \max_{\sigma^i_1\in\SSigma_1^i} \Vone(\ssigma_1^{*},\sigma_1^{i},\v^*)_{i,s_1,\p_0}.
\end{split} 
\end{equation}
By \eqref{eqn:max_equality_V1_V1conditional}, which we proved above, 
\begin{equation}\label{NashT1_1_coordinate4}
\begin{split}
    &\Vonet(s_1,\p_0,\sigma_1^{i*}|\ssigma_1^{-i*}) = \max_{\sigma^i_1\in \SSigma_1^i}\Vonet(s_1,\p_0,\sigma_1^i|\ssigma_1^{-i*}).
\end{split} 
\end{equation}
It follows that $\ssigma^*_1$ is a Nash equilibrium from $t=1$.   \qed


\subsubsection{Proof of Theorem \ref{thm:existence_sigma_0}}\label{sec:proof_theorem3}
Let $\ssigma_1^*\in \SSigma_1$ and $\v^*\in \R^{nrM}$ be the quantities given by Theorem \ref{Thm:Fink}. By Theorem \ref{Thm:ExistenceNashT1}, $\ssigma_1^*$ is a Nash equilibrium from time $t=1$. To prove the theorem, we need to show that there exists $\ssigma_0^*\in\SSigma_0$ satisfying for each $i\in[n]$
\begin{equation}\label{proof_eqn:Nasht_0}
    \sigma_0^{i*}\in \argmax_{\sigma_0^i\in\SSigma_0^i} \Vzerot(s_0, (\sigma_0^{i},\sigma_1^{i*})|(\ssigma_0^{-i*},\ssigma_1^{-i*})).
\end{equation}
We can rewrite the above equation by defining 
for each $(\p_0,s_0)\in \A^n\times \S$
\begin{equation}
    \label{proof_def:NashQ_function}
    \hat{v}^i(\p_0,s_0) \coloneq \pi^i(\p_0,s_0) + \delta_i \sum_{s_1\in\S} \P(s_1|\p_0,s_0)v_{i,s_1,\p_0}^*
\end{equation}
and noting that 
\begin{equation}\label{proof_eqn:Nasht_0_Q0}
    \Vzerot(s_0, (\sigma_0^{i},\sigma_1^{i*})|(\ssigma_0^{-i*},\ssigma_1^{-i*})) = \E_{(\sigma_0^{i},\ssigma_0^{-i*})}[\hat{v}^i(\p,s)|s_0].
\end{equation} 
Note that \eqref{proof_def:NashQ_function} is similar to a quantity that appears in \cite{hu2003nash}, where it is referred to as the Nash Q-function of agent $i$ at $(\p_0,s_0)$. By Theorem \ref{Thm:Fink} and equation \eqref{eqn:equality_V1_V1conditional}, $\Vonet(s_1,\p_0,\sigma_1^{i*}|\ssigma_1^{-i*})=v^*_{i,s_1,\p_0}$. Using the latter fact, and \eqref{not_Esigma_P_g}, \eqref{conditional_value_function} and \eqref{conditional_value_function_t1} we prove \eqref{proof_eqn:Nasht_0_Q0} by obtaining for each $s_0\in\S$ and $\ssigma = (\ssigma_0, \ssigma_1^*)$
\begin{equation}
\begin{split}\label{proof_conditional_value_function2}
    \Vzerot(s_0, \ssigma^i|\ssigma^{-i})
    &=\sum_{\p_0\in\A^n}\sigma_0(\p_0|s_0)\left\{\pi^i(\p_0,s_0)+\delta_i\sum_{s_1\in \S}\P(s_1|\p_0,s_0)\Vonet(s_1,\p_0,\sigma_1^{i*}|\ssigma_1^{-i*}) \right\}\\
    &=\sum_{\p_0\in\A^n}\sigma_0(\p_0|s_0)\left\{\pi^i(\p_0,s_0)+\delta_i\sum_{s_1\in \S}\P(s_1|\p_0,s_0)v^*_{i,s_1,\p_0} \right\}\\
    & = \E_{\ssigma_0}[\hat{v}^i(\p,s)|s_0]. 
\end{split}
\end{equation}

The use of \eqref{proof_eqn:Nasht_0_Q0} in \eqref{proof_eqn:Nasht_0} easily concludes the proof. Indeed, the existence of $\ssigma_0^*\in\SSigma_0$ satisfying for each $i\in[n]$
\begin{equation*}
    \sigma_0^{i*}\in \argmax_{\sigma_0^i\in\SSigma_0^i} \E_{(\sigma_0^{i},\ssigma_0^{-i*})}[\hat{v}^i(\p,s)|s_0].
\end{equation*} 
is guaranteed by the existence of Nash equilibrium in mixed strategies in \cite{nash1950equilibrium}. The profile $(\ssigma_0^*,\ssigma_1^*)$, where $\ssigma_1^*$ is given by Theorem \ref{Thm:ExistenceNashT1} and $\ssigma_0^*$ is given by \eqref{eqn:Nasht_0}, is a one-memory SPE of the stochastic game. 
\qed

\subsection{Application of Algorithm~\ref{algo:HowtoCheckNash}: Grim Trigger Strategies as an SPE}
\label{sec:Repeated_Bertrand}

The results in Section~\ref{sec:NashT1} apply to a broad class of stochastic games. Leveraging this generality, we derive non-trivial implications for how collusion can be sustained under one-memory strategies. In particular, we provide sufficient conditions under which a grim trigger strategy that supports a collusive-enabling price constitutes a one-memory SPE. These conditions also apply to other theoretical statements. 

\begin{proposition}[The grim trigger strategy is a one-memory SPE]\label{prop:grimTriger}
Under the assumptions of Section \ref{sec:Model} and Assumptions \ref{assumption_deltai} and \ref{assumption_deltai_simple}, the grim trigger strategy is an SPE of the stochastic game. Moreover, 
\begin{equation}\label{eqn:V0_bertrand}
    \Vzerot(\ssigma^f)
    =\frac{1}{1-\delta_i}\pi^i(\p^C).
\end{equation}
\end{proposition}


\subsubsection{Proof of Proposition \ref{prop:grimTriger}}\label{sec:proofgrimTriger}

Recall that each firm uses $\ssigma^f=(\sigma_0^f,\sigma_1^f)$, where $\sigma_0^f(p^C)=1$, $\sigma_1^{f}(p^C|\p^C)=1$, and $\sigma_1^{f}(p^*|\p_0)=1$ for each $\p_0\in\A^n \setminus \{\p^C\}$. We use Algorithm \ref{algo:HowtoCheckNash} to show that $\ssigma^f$ is an SPE of the stochastic game. 

\textbf{Step 1 of Algorithm \ref{algo:HowtoCheckNash}:} We plug $\ssigma_1^f$ into equation \eqref{NashT1_1} and solve it as a linear system with unknowns listed in the vector $\v^f = (v^f_{i,\p_0})_{i\in[n],\p_0\in\A^n}$, and obtain
\begin{equation}\label{proof:titfortat_vf2}
        v^f_{i,\p_0} = \Vone(\ssigma_1^f,\ssigma_1^f,\v^f)_{i,\p_0}.
\end{equation}
By \eqref{V1_function}, \eqref{proof:titfortat_vf2} is equivalent to 
\begin{equation*}
        v^f_{i,\p_0}= \sum_{\p_1\in\A^n}\sigma_1^{f}(\p_1|\p_0)\left[\pi^i(\p_1)+\delta_iv_{i,\p_1}^{f}\right].
\end{equation*}
It follows that for each $i\in[n]$,
\begin{equation}\label{proof:titfortat_vf3}
v^f_{i,\p_0}=\frac{1}{1-\delta_i}\cdot\begin{cases}
   \pi^i(\p^C) & \textnormal{ if }  \p_0 = \p^C, \\
    \pi^i(\p^*) & \textnormal{ if } \p_0\neq \p^C.
\end{cases}
\end{equation}
    
\textbf{Step 2 of Algorithm \ref{algo:HowtoCheckNash}:}  We plug  $\v^f$ and $\ssigma_1^f$ into \eqref{NashT1_2} and show that $\v^f$ is a fixed point of the operator $v_{i,\p_0}\mapsto\max_{\sigma_1^i\in\SSigma_1^i}\Vone(\ssigma_1^f,\sigma_1^i,\v)_{i,\p_0}$. By Assumption \ref{assumption_deltai}, $\p^*$ is a Nash equilibrium of the game $(\pi^i(\cdot))_{i=1}^n$, and thus 
\begin{equation}\label{proof:titfortat_vf4}
        \frac{\pi^i(\p^*)}{1-\delta_i} \geqq \max_{p^i\in\A\setminus\{p^*\}} \pi^i(p^i,(\p^*)^{-i})  + \delta_i\frac{\pi^i(\p^*)}{1-\delta_i}.
\end{equation}
Similarly, by rewriting  Assumption \ref{assumption_deltai_simple}, we obtain
\begin{equation}\label{proof:titfortat_vf5}
        \frac{\pi^i(\p^C)}{1-\delta_i} \geqq \max_{p^i\in\A\setminus\{p^C\}} \pi^i(p^i,(\p^C)^{-i})  + \delta_i\frac{\pi^i(\p^*)}{1-\delta_i}.
\end{equation}
By \eqref{proof:titfortat_vf3}, \eqref{proof:titfortat_vf4} and \eqref{proof:titfortat_vf5}, it follows that
\begin{equation*}
    \begin{split}
        &\max_{\tau_1^i\in\SSigma_1^i}\Vone(\ssigma_1^f,\tau_1^i,\v^f)_{i,\p_0} \\
        &= \max_{\tau_1^i\in\SSigma_1^i}\sum_{p^i\in\A}\tau_1^i(p^i|\p_0)\cdot \begin{cases}
  \pi^i(p^i,(\p^C)^{-i})+\delta_i v^f_{i,(p^i,(\p^C)^{-i})}   & \textnormal{ if }  \p_0 = \p^C, \\
   \pi^i(p^i,(\p^*)^{-i})+\delta_i v^f_{i,(p^i,(\p^*)^{-i})}   & \textnormal{ if } \p_0\neq \p^C,
   \end{cases}\\
   &=\frac{1}{1-\delta_i}\cdot\begin{cases}
   \pi^i(\p^C) & \textnormal{ if }  \p_0 = \p^C, \\
    \pi^i(\p^*) & \textnormal{ if } \p_0\neq \p^C
\end{cases} \\
& = v_{i,\p_0}^f. 
\end{split}
\end{equation*}
We thus conclude that $\v^f$ is a fixed point of the operator $v_{i,\p_0}\mapsto\max_{\sigma_1^i\in\SSigma_1^i}\Vone(\ssigma_1^f,\sigma_1^i,\v)_{i,\p_0}$.
    
\textbf{Step 3 of Algorithm \ref{algo:HowtoCheckNash}:} 
Applying \eqref{proof:titfortat_vf3}, \eqref{proof:titfortat_vf4} and \eqref{proof:titfortat_vf5} in a similar way as in step 2 above, we obtain that $$\sigma_0^{f}\in \argmax_{\tau_0^i\in\SSigma_0^i} \Vzerot( (\tau_0^{i},\sigma_1^{f})|(\ssigma_0^{f},\ssigma_1^{f})^{-i}),$$ where
\begin{equation*}
\begin{split}
    \Vzerot( (\tau_0^{i},\sigma_1^{f})|(\ssigma_0^{f},\ssigma_1^{f})^{-i})= \sum_{p_0^i\in\A}\tau^i(p_0^i)\left\{\pi^i(p_0^i,(\p^C)^{-i})+\delta_iv^f_{i,(p_0^i,(\p^C)^{-i})} \right\}.
\end{split}
\end{equation*}
We thus conclude that $\ssigma_0^f$ satisfies \eqref{eqn:Nasht_0}. 
Lastly, the combination of the above equation with \eqref{proof:titfortat_vf3} yields for each $i\in[n]$,
\begin{equation}\label{proof:V0_bertrand}
    \Vzerot(\ssigma^f)
    =\frac{1}{1-\delta_i}\pi^i(\p^C).
\end{equation}

\qed


\subsection{Proof of Proposition~\ref{prop:QfixedisV1}}\label{proofpropQfixedisV1}
Recall that $\alpha_t = \alpha\in (0,1]$ for each $t\geqq 0$ and $(Q_{\!f}^{i})_{i=1}^n$ is a fixed point of Algorithm \ref{Algo_Q_argmax}. Furthermore, for $\s=(s_1,\p_0)\in\S\times\A^n$, each firm $i\in[n]$ chooses an action according to \eqref{def:induced_strategy_Qil} and consequently
\begin{equation*}
\max_{p\in\A}Q_{\!f}^{i}(\hat{\s},p) = Q_{\!f}^{i}(\hat{\s},w_{\!f}^{i}(\hat{\s})). 
\end{equation*}
Because $(Q_{\!f}^{i})_{i=1}^n$ is a fixed point of Algorithm \ref{Algo_Q_argmax}, then the next update of $Q_{\!f}^{i}$ satisfies
\begin{equation}\label{proof_prop:QfixedisV1_2}
    Q_{\!f}^{i}(\s,w_{\!f}^{i}(\s)) = (1-\alpha)Q_{\!f}^{i}(\s,w_{\!f}^{i}(\s))+\alpha\left\{\pi^i(\w_{\!f}(\s),\s)+\delta_i \mathbb{E}_{\hat{\s}}\left[\max_{p\in\A}Q_{\!f}^{i}(\hat{\s},p)\right]\right\},
\end{equation}
where $\hat{s}= (s_2,\w_{\!f}(\s))$ represents the new state after the firms play with $\w_{\!f}(\s)$. 

Subtracting $(1-\alpha)Q_{\!f}^{i}(\s,w_{\!f}^{i}(\s))$ from both sides of \eqref{proof_prop:QfixedisV1_2}, dividing by $\alpha\neq 0$, and using $\s=(s_1,\p_0)$,
\begin{equation}\label{proof_prop:QfixedisV1_3}
    Q_{\!f}^{i}(\s,w_{\!f}^{i}(\s)) = \pi^i(\w_{\!f}(\s),s_1)+\delta_i \sum_{s_2\in\S}\P(s_2|\w_{\!f}(\s),s_1)Q_{\!f}^{i}(\hat{\s},w_{\!f}^{i}(\hat{\s})).
\end{equation}
It follows from Proposition \ref{prop:BellmanEqV1} that 
for each $\s=(s_1,\p_0)\in\S\times\A^n$ 
\begin{equation*}
        Q_{\!f}^{i}(\s, w_{\!f}^{i}(\s)) = \Vonet(\s, w_{\!f}^{i}|\w_{\!f}^{-i}).
\end{equation*}
\qed

\subsection{Proof of Proposition \ref{prop:Sufficient_Q_inducesNash}}
Recall that $\alpha_t = \alpha\in (0,1]$ for each $t\geqq 0$, $\Q_{\!f} = (Q_{\!f}^{i})_{i=1}^n$ is a fixed point of Algorithm \ref{Algo_Q_argmax}, and 
\eqref{eqn:sufficient_Q_inducesNash} holds for each $i\in[n]$ and $\s=(s_1,\p_0)\in\S\times\A^n$. We use steps 1 and 2 of Algorithm \ref{algo:HowtoCheckNash} to show that $\w_{\!f} = \{w_{\!f}^{i}(\s)| i\in[n], \s\in\S\times\A^n\}$ is a Nash equilibrium from time $t=1$. 

\textbf{Step 1 of Algorithm \ref{algo:HowtoCheckNash}:} We plug $\w_{\!f}$ into equation \eqref{NashT1_1} and solve it as a linear system with unknowns $v_{i,\s}$ for each $(i,\s)\in[n]\times\S\times\A^n$ and obtain
\begin{equation}\label{proof_Sufficient_Q_inducesNash_1}
    v_{i,\s} = \pi^i(\w_{\!f}(\s),s_1)+\delta_i\sum_{s_2\in\S}\P(s_2|\w_{\!f}(\s),s_1)v_{i,\hat{s}},
\end{equation}
where $\hat{s} = (s_2,\w_{\!f}(\s))$. By Proposition \ref{prop:BellmanEqV1}, $v_{i,\s} = \Vonet(\s,w_{\!f}^{i}|\w_{\!f}^{-i})$ for each $\s\in\S\times\A^n$, $i\in[n]$. Moreover, by Proposition \ref{prop:QfixedisV1}, 
\begin{equation}\label{proof_Sufficient_Q_inducesNash_2}
    v_{i,\s} =
        Q_{\!f}^{i}(\s, w_{\!f}^{i}(\s)).
    \end{equation}

\textbf{Step 2 of Algorithm \ref{algo:HowtoCheckNash}:} We plug $\v = (v_{i,\s})_{i\in[n],\s\in\S\times\A^n}$ and $\w_{\!f}$ into \eqref{NashT1_2} to show that $\v$ is a fixed point of the operator $v_{i,\s}\mapsto\max_{\sigma_1^i\in\SSigma_1^i}\Vone(\w_{\!f},\sigma_1^i,\v)_{i,\s}$. By \eqref{eqn:sufficient_Q_inducesNash} and \eqref{proof_Sufficient_Q_inducesNash_2},
\begin{equation*}
    \begin{split}
&\max_{\sigma_1^i\in\SSigma_1^i}\Vone(\w_{\!f},\sigma_1^i,\v)_{i,\s} = \max_{p_1^i\in\A}\Vone(\w_{\!f},p_1^i,\v)_{i,\s}=\Vone(\w_{\!f},w_{\!f}^{i},\v)_{i,\s}\\
&=\pi^i(\w_{\!f}(\s),s_1)+\delta_i\sum_{s_2\in\S}\P(s_2|\w_{\!f}(\s),s_1)v_{i,\hat{s}} = v_{i,\s}.
    \end{split}
\end{equation*}
The above verification of the first two steps of Algorithm \ref{algo:HowtoCheckNash} implies that $\w_{\!f} = \{w_{\!f}^{i}(\s)| i\in[n], \s\in\S\times\A^n\}$ is a Nash equilibrium from time $t=1$.

\qed


\subsection{\texorpdfstring{Properties of $\alpha(\delta)$}{Properties of alpha(delta)}}\label{sect:properties_ofAlpha}

Let $\alpha_t$ be a sequence such that $\alpha_t\in(0,1)$ for each $t$. Let $\delta\in (0,1)$, we define 
$$\alpha(\delta) := \lim_{t\to \infty} \sum_{k=T+1}^t \prod_{l=k+1}^t (1-\alpha_l(1-\delta)) \alpha_k.$$

\begin{proposition}[Properties of $\alpha(\delta)$]\label{prop:identity} Fix $\delta\in(0,1)$, $T\geqq 0$, and any sequence $\alpha_t\in(0,1)$. For $t>T$ set
\[ P_t:=\prod_{l=T+1}^{t}\big(1-\alpha_l(1-\delta)\big),
\qquad S_t:=\sum_{k=T+1}^{t}\Big[\prod_{l=k+1}^{t}\big(1-\alpha_l(1-\delta)\big)\Big]\alpha_k . \]
Then $(1-\delta)\,S_t=1-P_t$ for every $t>T$. Consequently:
\begin{enumerate}[label=(\alph*),leftmargin=*]
\item $\lim_t S_t$ exists if and only if $\lim_t P_t$ exists, and $(1-\delta)\lim_t S_t=1-\lim_t P_t\leq 1$. In particular $\alpha(\delta)(1-\delta)\leq 1$.
\item If $\sum_{t>T}\alpha_t=\infty$, then $P_t\to0$, the limit defining $\alpha(\delta)$ exists and $\alpha(\delta)(1-\delta)=1$.
\end{enumerate}
\end{proposition}

\begin{proof}
Since $\alpha_l\in(0,1)$ and $1-\delta\in(0,1)$, each factor $1-\alpha_l(1-\delta)$ lies in $(0,1)$, so $P_t\in(0,1)$ for every $t>T$. We prove that $(1-\delta)\,S_t=1-P_t$ for every $t>T$, by induction on $t$. 

\begin{itemize}
\item For $t=T+1$ the inner product in $S_{T+1}$ is empty, so $S_{T+1}=\alpha_{T+1}$ and $(1-\delta)S_{T+1}=1-\big(1-\alpha_{T+1}(1-\delta)\big)=1-P_{T+1}$. 

\item For the step, separate the $k=t+1$ term of $S_{t+1}$ from the rest, 
\[
S_{t+1}=\big(1-\alpha_{t+1}(1-\delta)\big)S_t+\alpha_{t+1}.
\]
Using $P_{t+1}=\big(1-\alpha_{t+1}(1-\delta)\big)P_t$ and the inductive hypothesis $(1-\delta)S_t=1-P_t$,
\begin{align}
(1-\delta)S_{t+1}
&=\big(1-\alpha_{t+1}(1-\delta)\big)(1-P_t)+(1-\delta)\alpha_{t+1}\nonumber\\
&=1-\alpha_{t+1}(1-\delta)+(1-\delta)\alpha_{t+1}-P_{t+1} =1-P_{t+1}. \nonumber
\end{align}
\end{itemize}

\emph{(a).} Rearranging, $S_t=(1-P_t)/(1-\delta)$, so one limit exists iff the other does and
$(1-\delta)\lim_t S_t=1-\lim_t P_t$. As $P_t\in(0,1)$ we have $\lim_t P_t\in[0,1]$, whence
$(1-\delta)\lim_t S_t\leq 1$.

\emph{(b).} Using the inequality $1-x\leq e^{-x}$ on each factor $P_t$,
\[
0<P_t\leq\exp\Big(-(1-\delta)\sum_{l=T+1}^{t}\alpha_l\Big)\xrightarrow[t\to\infty]{}0
\]
when $\sum_{l>T}\alpha_l=\infty$. Then $\alpha(\delta)=\lim_t S_t=1/(1-\delta)$.
\end{proof}

\begin{proof}[Proof of Lemma \ref{lem:alphadelta}] The proof of this lemma follows directly from part \textit{(b)} of Proposition \ref{prop:identity}.
    
\end{proof}

\subsection{Proof of Theorem \ref{prop:QConvergence_T}}\label{Appendix_RLalgo}
We break down the proof of Theorem \ref{prop:QConvergence_T} into two main steps: (I) We prove Lemma \ref{claim_thm_convergence_T} below which concludes the first claim of Theorem \ref{prop:QConvergence_T} and also characterizes the values of the $Q$-function given by \eqref{def:Q_learningUpdate} for each $t\geqq T$; (II) We use the latter claim to compute the limit in equation \eqref{eqn_Q_epsilon*}. 

\vspace{5pt}
\textbf{Step (I):} We formulate and establish Lemma \ref{claim_thm_convergence_T}. It uses the definition $\tilde{\alpha}_k := (1-\alpha_k(1-\delta_i))$, for each $k\in \mathbb{N}$, and the convention that $\prod_{k=l}^{l-1}\tilde{\alpha}_k = 1$ for each $l\in\mathbb{N}$.

\begin{lemma}\label{claim_thm_convergence_T} If the assumptions of Theorem \ref{prop:QConvergence_T} hold, then for each $i\in[n]$, $t \geqq T$, $p_t^i = p^C$. Moreover, for each $i\in[n]$, $t \geqq T$ and $p\in\A\setminus\{p^C\}$, $Q_t^i(\p_{t-1},p^C)>Q_t^i(\p_{t-1},p)$ and the following equations hold true, 
\begin{equation}
\label{proof:Q*Convergence_au1}
 Q_{T+1}^i(\s,p)=\begin{cases}
        (1-\alpha_T)Q_{T}^i(\p_{T-1},p^C) & \\
        \qquad+\alpha_T[\pi^i(\p^C)+\delta_iQ_T^i(\p^C,p^C)] & \textnormal{if } (\s,p)=(\p_{T-1}, p^C), \\
        Q_T^i(\s,p) & \textnormal{otherwise, }
    \end{cases}
\end{equation}
and for each $t\geqq T+1$
\begin{equation}
\label{proof:Q*Convergence_au1_1}
 Q_t^i(\s,p)=\begin{cases}
        \prod_{k=T+1}^{t-1}\tilde{\alpha}_kQ_{T+1}^i(\p^C,p^C)+\sum_{k=T+1}^{t-1}\prod_{l=k+1}^{t-1}\tilde{\alpha}_l\alpha_k\pi^i(\p^C) & \textnormal{if } (\s,p)=(\p^C, p^C), \\
        Q_{T+1}^i(\s,p) & \textnormal{otherwise. }
    \end{cases}
\end{equation}
\end{lemma}

\vspace{5pt}
\textbf{Proof of Lemma \ref{claim_thm_convergence_T}.} We fix $i\in[n]$ and $t=T$. We note that Assumption (i) in Theorem \ref{prop:QConvergence_T} implies that for each $p\in\A\setminus\{p^C\}$, $Q_T^i(\p_{T-1},p^C)>Q_T^i(\p_{T-1},p)$ and consequently  $$\argmax_{a\in\A}Q_{T}^i(\p_{T-1},a) = \{p^C\}.$$ 
This observation and Algorithm \ref{Algo_Q_boundedExperimentation} imply that $\s_{T+1} = \p_T = \p^C$. We thus conclude that 
for each $i\in[n]$ and $p\in\A\setminus\{p^C\}$, $p_T^i = p^C$ and $Q_T^i(\p_{T-1},p^C)>Q_T^i(\p_{T-1},p)$. 

To prove the statements in Lemma \ref{claim_thm_convergence_T} for $t\geqq T+1$ we use strong induction.  

\vspace{5pt}
$\bullet$ \textit{Base Case.} Let $t=T+1$. We first show that \eqref{proof:Q*Convergence_au1} and \eqref{proof:Q*Convergence_au1_1} hold true. Then, we use \eqref{proof:Q*Convergence_au1} to show that for each $i\in[n]$ and $p\in\A\setminus\{p^C\}$, $p_{T+1}^i = p^C$ and $Q_{T+1}^i(\p_{T},p^C)>Q_{T+1}^i(\p_{T},p)$.

In view of what we proved and Assumption \ref{assumption_deltai}-(i), $(\s_{T},p_{T}^i) = (\p_{T-1},p^C)$. Using the update rule \eqref{def:Q_learningUpdate} from Algorithm \ref{Algo_Q_argmax}, for each $(\s,p)\neq (\p_{T-1},p^C)$, $Q_{T+1}^i(\s,p)=Q_{T}^i(\s,p)$ and
\begin{equation}\label{proof:Q*Convergence_au1_T1}
 Q_{T+1}^i(\p_{T-1}, p^C)=\
        (1-\alpha_T)Q_{T}^i(\p_{T-1},p^C)+\alpha_T[\pi^i(\p^C)+\delta_i\max_{p\in\A}Q_T^i(\p^C,p)]. 
\end{equation}
In particular, \eqref{proof:Q*Convergence_au1} holds when $(\s,p)\neq (\p_{T-1}, p^C)$. On the other hand, Assumption (i) in Theorem \ref{prop:QConvergence_T} implies that 
\begin{equation}\label{proof:Q*Convergence_au1_T1_1_1}
    \max_{p\in\A}Q_T^i(\p^C,p) = Q_T^i(\p^C,p^C).
\end{equation}
Equation \eqref{proof:Q*Convergence_au1_T1_1_1} into \eqref{proof:Q*Convergence_au1_T1} yields \eqref{proof:Q*Convergence_au1} when $(\s,p)=(\p_{T-1}, p^C)$. Finally, note that for $t=T+1$, \eqref{proof:Q*Convergence_au1_1} trivially holds since $\prod_{k=T+1}^T\tilde{\alpha_k} = 1$.

Now, we use \eqref{proof:Q*Convergence_au1} to show that for each $p\in\A\setminus\{p^C\}$, $Q_{T+1}^i(\p_{T},p^C)>Q_{T+1}^i(\p_{T},p)$. We do so in two cases:
\begin{itemize}
    \item[$\diamond$] $\p_{T-1}\neq \p^C$. By \eqref{proof:Q*Convergence_au1} and \eqref{proof:Q*Convergence_au1_T1_1_1}, for each $p\in\A$,  $Q_{T+1}^i(\p^C,p^C)=Q_{T}^i(\p^C,p^C)>Q_{T}^i(\p^C,p)=Q_{T+1}^i(\p^C,p)$. 
    \item[$\diamond$] $\p_{T-1}= \p^C$. Using \eqref{proof:Q*Convergence_au1} and Assumption (ii) in Theorem \ref{prop:QConvergence_T}, we obtain for each $p\in\A\setminus\{p^C\}$
\begin{equation}\label{proof:Q*Convergence_au1_T1_2}
\begin{split}
 Q_{T+1}^i(\p^C, p^C)&=\
        (1-\alpha_T)Q_{T}^i(\p^C,p^C)+\alpha_T[\pi^i(\p^C)+\delta_iQ_T^i(\p^C,p^C)]\\
        &=(1-\alpha_T+\alpha_T\delta_i)Q_{T}^i(\p^C,p^C)+\alpha_T\pi^i(\p^C)\\
        &\geqq (1-\alpha_T+\alpha_T\delta_i)Q_{T}^i(\p^C,p^C)+\alpha_T(1-\delta_i)Q_T^i(\p^C,p)\\
        &= (1-\alpha_T(1-\delta_i))\underbrace{[Q_{T}^i(\p^C,p^C)-Q_T^i(\p^C,p)]}_{>0, \textnormal{ by (i) in Theorem \ref{prop:QConvergence_T}} } + Q_T^i(\p^C,p).
\end{split}
\end{equation}
Given that $\alpha_T(1-\delta_i)<1$, \eqref{proof:Q*Convergence_au1} and \eqref{proof:Q*Convergence_au1_T1_2} imply that for each $p\in\A\setminus\{p^C\}$, $Q_{T+1}^i(\p^C, p^C) > Q_{T}^i(\p^C, p) = Q_{T+1}^i(\p^C, p)$.
\end{itemize}

The inequality we have just established, namely $Q_{T+1}^i(\p_T, p^C) > Q_{T+1}^i(\p_T, p)$ for each $p \in \A \setminus \{p^C\}$, together with Algorithm~\ref{Algo_Q_boundedExperimentation}, implies that $\p_{T+1} = \p^C$.

\vspace{5pt}
$\bullet$ \textit{Inductive case}. Let \( t \geqq T+1 \), and assume that Lemma~\ref{claim_thm_convergence_T} holds for each \( k \in \{T+1, \dots, t\} \). We now prove that it also holds for \( t+1 \).
 By the inductive hypothesis, $\s_{k+1}= \p_k = \p^C$ and $Q_t^i(\p^C,p^C)>Q_t^i(\p^C,p)$ for each $T+1\leq k\leq t$ and  $p\in\A\setminus\{p^C\}$. By rule \eqref{def:Q_learningUpdate} with $(\s,p) = (\s_t,p_t^i) = (\p^C,p^C)$,
\begin{equation}\label{proof:Q*Convergence_au2}
\begin{split}
    Q_{t+1}^i(\p^C,p^C) &= (1-\alpha_t)Q_{t}^i(\p^C,p^C)+\alpha_t\left[\pi^i(\p^C)+\delta_i\max_{p\in\A}Q_t^i(\p^C,p)\right]\\
    &=(1-\alpha_t)Q_{t}^i(\p^C,p^C)+\alpha_t\left[\pi^i(\p^C)+\delta_iQ_t^i(\p^C,p^C)\right]\\
    &= (1-\alpha_t(1-\delta_i))Q_t^i(\p^C,p^C)+\alpha_t\pi^i(\p^C).
\end{split}
\end{equation}
Moreover, because $\p_k = \p^C$ for each $T+1\leq k\leq t$, by \eqref{proof:Q*Convergence_au1} and rule \eqref{def:Q_learningUpdate} for each $p\in\A\setminus\{p^C\}$, 
\begin{equation}\label{proof:Q*Convergence_au2_2}
\begin{split}
    Q_{t+1}^i(\p^C,p) = Q_{t}^i(\p^C,p) = \dots = Q_{T}^i(\p^C,p). 
\end{split}
\end{equation}
Combining \eqref{proof:Q*Convergence_au2}, \eqref{proof:Q*Convergence_au2_2} and Assumption (ii) in Theorem \ref{prop:QConvergence_T}, we obtain for each $p\in\A\setminus\{p^C\}$
\begin{equation}
\begin{split}
    Q_{t+1}^i(\p^C,p^C) &> (1-\alpha_t(1-\delta_i))Q_{T}^i(\p^C,p)+\alpha_t\pi^i(\p^C)\\
    &\geqq (1-\alpha_t(1-\delta_i))Q_{T}^i(\p^C,p)+\alpha_t(1-\delta_i)Q_T^i(\p^C,p)\\
    & = Q_T^i(\p^C,p) = Q_{t+1}^i(\p^C,p).
\end{split}
\end{equation}
It follows that $Q_{t+1}^i(\p^C,p^C)>Q_{t+1}^i(\p^C,p)$ for each $p\in\A\setminus\{p^C\}$. The latter along with Algorithm \ref{Algo_Q_boundedExperimentation} imply that $\p_{t+1}=\p^C$. Finally, since by the inductive hypothesis \eqref{proof:Q*Convergence_au1_1} holds for $T+1\leq k\leq t$, we plug it into \eqref{proof:Q*Convergence_au2} and obtain
\begin{equation*}
    \begin{split}
    Q_{t+1}^i(\p^C,p^C) &= (1-\alpha_t(1-\delta_i))Q_t^i(\p^C,p^C)+\alpha_t\pi^i(\p^C)\\
    &= \tilde{\alpha}_t\prod_{k=T+1}^{t-1}\tilde{\alpha}_kQ_{T+1}^i(\p^C,p^C)+\tilde{\alpha}_t\sum_{k=T+1}^{t-1}\prod_{l=k+1}^{t-1}\tilde{\alpha}_l\alpha_k\pi^i(\p^C)+\alpha_t\pi^i(\p^C)\\
    &=\prod_{k=T+1}^{t}\tilde{\alpha}_kQ_{T+1}^i(\p^C,p^C)+\sum_{k=T+1}^{t}\prod_{l=k+1}^{t}\tilde{\alpha}_l\alpha_k\pi^i(\p^C)
\end{split}
\end{equation*}
and thus conclude the proof of  \eqref{proof:Q*Convergence_au1_1} for $t+1$.  
\qed

\vspace{5pt}
\textbf{Step (II):} We use Lemma \ref{claim_thm_convergence_T} to compute  $Q^{i*}(\s,p):=\lim_{t\to \infty} Q_t^i(\s,p)$.

\vspace{5pt}
\textbf{Case 1:} $(\s,p)=(\p^C,p^C)$. By \eqref{proof:Q*Convergence_au1_1}, for each $t\geqq T+1$, $i\in [n]$
\begin{equation}\label{proof:Q*Convergence_au3}
\begin{split}
    Q_{t+1}^i(\p^C,p^C)=\prod_{k=T+1}^{t}\tilde{\alpha}_kQ_{T+1}^i(\p^C,p^C)+\sum_{k=T+1}^{t}\prod_{l=k+1}^{t}\tilde{\alpha}_l\alpha_k\pi^i(\p^C).
\end{split}
\end{equation}
By definition of $\tilde{\alpha}_k = 1-\alpha_k(1-\delta_i)$, $\tilde{\alpha}_k \in (0,1)$ for each $k\geqq 1$. Using Assumption \ref{assumption_learningrate}, we obtain the following
\begin{equation}\label{proof:Q*Convergence_au4}
    \begin{split}
        \prod_{k=T+1}^{t}\tilde{\alpha}_k = e^{\sum_{k=T+1}^t\log(\tilde{\alpha}_k)}\leq e^{\sum_{k=T+1}^t\tilde{\alpha}_k-1} = e^{-(1-\delta_i)\sum_{k=T+1}^t \alpha_k} \to 0 \textnormal{ as } t\to\infty.
    \end{split}
\end{equation}
Thus, $\lim_{t\to\infty}\prod_{k=T+1}^{t}\tilde{\alpha}_k = 0$. Combining the latter fact with \eqref{proof:Q*Convergence_au3} yields 
\begin{equation*}
\begin{split}
    Q^{i*}(\p^C,p^C)= \lim_{t\to \infty} Q_{t}^i(\p^C,p^C)
    &=\lim_{t\to \infty}\sum_{k=T+1}^{t}\prod_{l=k+1}^{t}\tilde{\alpha}_l\alpha_k\pi^i(\p^C) = \alpha(\delta_i)\pi^i(\p^C)=\frac{\pi^i(\p^C)}{1-\delta_i},
\end{split}
\end{equation*}
where the last equality is Lemma~\ref{lem:alphadelta}.

\vspace{5pt}
\textbf{Case 2:} $(\s,p)=(\p_{T-1}, p^C)$ and $\p_{T-1}\neq \p^C$. Using  \eqref{def:Q_learningUpdate} and \eqref{proof:Q*Convergence_au1_T1_1_1},
\begin{equation*}
    Q_{T+1}^i(\p_{T-1},p^C) =
        (1-\alpha_T)Q_{T}^i(\p_{T-1},p^C)+\alpha_T\left[ \pi^i(\p^C) +\delta_iQ_T^i(\p^C,p^C) \right].
\end{equation*}

\vspace{5pt}
\textbf{Case 3:} $(\s,p)$ not covered by cases 1 and 2 above. From Lemma \ref{claim_thm_convergence_T}, $Q_{t+1}^i(\s,p) = Q_T^i(\s,p)$ for each $t\geqq T$. Thus, $Q^{i*}(\s,p) = Q_T^i(\s,p)$. 

\qed

\subsection{Two Technical Lemmas}\label{sec:lemmas}

Recall that under Assumption~\ref{assumption_deltai}, the update rule of Algorithm~\ref{Algo_Q_argmax} reads 
\begin{equation}\label{eq:U}
Q^i_{t+1}(\s_t,p^i_t)=(1-\alpha_t)\,Q^i_t(\s_t,p^i_t)
+\alpha_t\Big[\pi^i(\p_t)+\delta_i\max_{a\in\A}Q^i_t(\p_t,a)\Big],
\end{equation}
with every other entry of the $Q$-matrix unchanged.

\begin{lemma}[Uniform bound on $Q$]\label{lem:bound}
For every $t\geqq1$, every $i$, every $\s$ and every $p$, $|Q^i_t(\s,p)|\leq K$, where $K$ is defined by \eqref{def:emergence_constants}.
\end{lemma}

\begin{proof}[Proof of Lemma \ref{lem:bound}] We prove the Lemma by induction. At $t=1$ the claim is $\norminf{Q_1}\leq K$, which holds true by definition of $K$. Now, assume the inequality holds true for each $1 \leq s \leq t$. At $t+1$, entries other than $(\s_t,p^i_t)$ are unchanged, so they satisfy the inequality we are trying to prove by the inductive hypothesis. Next, we check the visited entry $(\s_t,p^i_t)$.

Let $G:=\pi^i(\p_t)+\delta_i\max_a Q^i_t(\p_t,a)$. Since $\pi^i\geqq0$, $\max_a Q^i_t\geqq -K$, and by definition of $\dbar$,

\[G \geqq 0-\delta_i K \geqq -\dbar K \geqq -K .\]
Since $\pi^i\leq\pibar$ and $\max_a Q^i_t\leq K$, and since $K\geqq\pibar/(1-\dbar)$ is equivalent to $K\geqq\pibar+\dbar K$,

\[G \leq \pibar+\delta_i K \leq \pibar+\dbar K \leq K .\]
So $G\in[-K,K]$. By \eqref{eq:U}, $Q_{t+1}^i(\s_t,p_t^i)$ is a convex combination of $Q^i_t(\s_t,p^i_t)\in[-K,K]$ and $G\in[-K,K]$, hence it lies in $[-K,K]$.
\end{proof}

\begin{lemma}[Residual uniform bound]\label{lem:budget} For all $1\leq\tau\leq t<\infty$, every $i$, every $\s$ and every $p$,

\[
\big|Q^i_t(\s,p)-Q^i_\tau(\s,p)\big| \leq \kappa\sum_{k=\tau}^{t-1}\alpha_k ,
\]
and in particular $|Q^i_T-Q^i_\tau|\leq R_\tau$ pointwise, where $\kappa$ and $R_{\tau}$ are defined by \eqref{def:emergence_constants}. 
\end{lemma}

\begin{proof}
At time $k$ firm $i$ revises exactly one entry, and by \eqref{eq:U} the change at this entry equals
\begin{align*}
\big|Q^i_{k+1}(\s_k,p^i_k)-Q^i_k(\s_k,p^i_k)\big|
&=\alpha_k\Big|\pi^i(\p_k)+\delta_i\max_a Q^i_k(\p_k,a)-Q^i_k(\s_k,p^i_k)\Big|\\
&\leq \alpha_k\big(\pibar+\dbar K+K\big)=\alpha_k\kappa.
\end{align*}
The rest follows from the triangle inequality 
\[
\big|Q^i_t(\s,p)-Q^i_\tau(\s,p)\big| \leq \sum_{k=\tau}^{t-1} \big|Q^i_{k+1}(\s,p)-Q^i_k(\s,p)\big|,
\]
and the fact that in the above summation the only term that may be nonzero is the $k$-term when $(\s,p) = (\s_k,p^i_k)$.
\end{proof}

\subsection{Proofs of Section \ref{sec:emergence}}\label{sec:proofs_emergence}

\begin{proof}[Proof of Proposition \ref{prop:capture}] By Lemma~\ref{lem:budget}, $Q^i_\tau(\s,p)-R_\tau\leq Q^i_T(\s,p)\leq Q^i_\tau(\s,p)+R_\tau$ for every entry $(\s,p)$. Hence for every $\s$ and every $p\neq p^C$,
\[
Q^i_T(\s,p^C)-Q^i_T(\s,p)\ \geqq\ \big(Q^i_\tau(\s,p^C)-Q^i_\tau(\s,p)\big)-2R_\tau\ >\ 0 ,
\]
which is condition~(i) of Theorem \ref{prop:QConvergence_T} at every state, in particular at the two states
$\s\in\{\p_{T-1},\p^C\}$. For condition~(ii),
\[
(1-\delta_i)Q^i_T(\p^C,p)\ \leq\ (1-\delta_i)\big(Q^i_\tau(\p^C,p)+R_\tau\big)\ \leq\ \pi^i(\p^C).
\]
Theorem~\ref{prop:QConvergence_T} now applies.
\end{proof}

\begin{proof}[Proof of Corollary \ref{cor:stopping}] Immediate from Proposition~\ref{prop:capture}. Measurability holds because assumptions (i) and (ii) of Proposition~\ref{prop:capture} at time $\tau$ are finitely many inequalities in the entries of
$Q_\tau$.
\end{proof}

Before proving Proposition \ref{prop:random}, we need a technical Lemma. 

\begin{lemma}\label{cor:basin} Suppose that Assumptions \ref{assumption_deltai} and \ref{assumption_learningrate} hold. Let $N=n(m+1)^{n+1}$ be the number of entries of $Q_1$, $B_0:=\big[0,\ \pibar/(1-\dbar)\big]^{N}$ and
\begin{equation}
\begin{split}\label{eqn:range_B0}
B:=\Big\{ & Q_1\in B_0:\ Q^i_1(\s,p^C)-Q^i_1(\s,p)>2R_1\ \ \text{and}\ \ \pi^i(\p^C)>(1-\delta_i)\big(Q^i_1(\p^C,p)+R_1\big)\\
& \forall i\in [n],\ \s\in\A^n,\ p\neq p^C \Big\}.
\end{split}
\end{equation}
Then $B$ is open in $B_0$, and every $Q_1\in B$ leads to $\p_t=\p^C$ for all $t\geq T$. Moreover, 
\begin{equation}
\label{eqn:nonnull_AT}
B\neq\emptyset \quad \Longleftrightarrow \quad A_T<\frac{1}{4}
\textnormal{ and }
R_1<\min_{i\in[n]}\frac{\pi^i(\p^C)}{1-\delta_i}.
\end{equation}
\end{lemma}

\begin{proof}[Proof of Lemma \ref{cor:basin}] By definition,  $\norminf{Q_1}\leq\pibar/(1-\dbar)$ on $B_0$, so $K=\pibar/(1-\dbar)$ and
$\kappa=\pibar+(1+\dbar)\pibar/(1-\dbar)=\pibar\frac{(1-\dbar)+1+\dbar}{1-\dbar} =2\pibar/(1-\dbar)$ (see, e.g., \eqref{def:emergence_constants}). $B$ is defined by finitely many strict inequalities in the entries of $Q_1$, hence is open in $B_0$. By definition, (i) and (ii) of Proposition~\ref{prop:capture} hold at $\tau=1$ inside $B$, hence $\p_t = \p^C$ for all $t\geq T$. 

\vspace{5pt}
We now prove the equivalence in \eqref{eqn:nonnull_AT}.
\vspace{5pt}

($\Leftarrow$) Suppose $A_T<1/4$ and $R_1<\min_i\pi^i(\p^C)/(1-\delta_i)$. From \eqref{def:emergence_constants},
$2R_1=2\kappa A_T=4\pibar A_T/(1-\dbar)<\pibar/(1-\dbar)$. For all $i\in[n],\s\in\A^n$, pick  $Q^i_1(\s,p^C)\in (2R_1, \pibar/(1-\dbar)]$ and set $Q^i_1(\s,p):=0$ for $p\neq p^C$. 

Then $Q_1\in B_0$, its margins satisfy $Q^i_1(\s,p^C)-0>2R_1$, and $(1-\delta_i)(0+R_1)<\pi^i(\p^C)$ by hypothesis; hence $Q_1\in B$.

\vspace{5pt}

($\Rightarrow$) Let $Q_1\in B$ and pick any $p\neq p^C$. By definition of $B$ and $B_0$, $2R_1<Q^i_1(\s,p^C)-Q^i_1(\s,p)\leq \pibar/(1-\dbar)-0$, then  $4\pibar A_T/(1-\dbar)<\pibar/(1-\dbar)$, i.e., $A_T<1/4$. Similarly, since
$Q^i_1(\p^C,p)\geq0$, then $\pi^i(\p^C)>(1-\delta_i)R_1$ for every $i$.
\end{proof}

\begin{remark}[On Lemma \ref{cor:basin}]\label{rem:basinreading}
Condition $A_T<1/4$ says that the total learning rate spent before the deadline must be small relative to the spread of the initial values; otherwise the updates made during experimentation can undo any initial ranking. It constrains only the experimentation phase and is therefore compatible with Assumption~\ref{assumption_learningrate}, which constrains only $\sum_{t\geq T}\alpha_t$. The second condition says the head start must not be so large that the alternatives look better than permanent collusion is worth.
\end{remark}

We can now prove Proposition \ref{prop:random}. 

\begin{proof}[Proof of Proposition \ref{prop:random}] We want to apply Proposition \ref{prop:capture}, by showing that conditions (i) and (ii) in such proposition hold at $\tau=1$ with probability $\rho^{\,n(m+1)^{n}}$, where $\rho:=\frac{1}{m+1}\Big(1-\frac{2R_1}{K_0}\Big)^{m+1} > 0 $. Take $\tau = 1$. 

\vspace{5pt}
$(a)$ First we verify condition (ii) in Proposition \ref{prop:capture} holds almost surely. By hypothesis, $Q^i_1(\p^C,p)\leq K_0$ surely and the hypotheses in \eqref{condition_start} yield
$(1-\delta_i)(Q^i_1(\p^C,p)+R_1)\leq(1-\delta_i)(K_0+R_1)<\pi^i(\p^C)$.

\vspace{5pt}
$(b)$ To verify condition (i) in Proposition \ref{prop:capture}, fix a pair $(i,\s)$ then $Q^i_1(\s,p^C)-Q^i_1(\s,p)>2R_1$ for every $p\in\A\setminus\{p^C\}$ becomes a random event, where the family of $Q_1^i(\s, \cdot)$ are i.i.d.~$U[0,K_0]$. Recall $m+1 = |\A|$ and by independence, the CDF of $\max_{p\neq p^C}Q^i_1(\s,p)$ is $(z/K_0)^m$ whenever $z\in[0,K_0]$. Then, 
\begin{equation}
    \label{eqn:almost_cond1propcapture}
    \begin{split}
    \Prob\Big(\max_{p\neq p^C}Q^i_1(\s,p)+2R_1<Q^i_1(\s,p^C)\Big) & = \frac{1}{K_0}\int_0^{K_0} \Prob\Big(\max_{p\neq p^C}Q^i_1(\s,p)+2R_1<x\Big) dx \\
   \text{\scriptsize (using $y=x-2R_1$)} &= \frac{1}{K_0}\int_{-2R_1}^{K_0-2R_1} \Prob\Big(\max_{p\neq p^C}Q^i_1(\s,p)<y\Big) dy \\
    &= \frac{1}{K_0}\int_{0}^{K_0-2R_1} \Big(\frac{y}{K_0}\Big)^m dy = \frac{(K_0-2R_1)^{m+1}}{(m+1)K_0^{m+1}}=\rho.
    \end{split}
\end{equation}
Distinct pairs $(i,\s)$ involve disjoint sets of entries of $Q_1$, so these $n(m+1)^{n}$ events are
independent and (i) holds with probability $\rho^{\,n(m+1)^{n}}$. Apply
Proposition~\ref{prop:capture}.

\end{proof}

\begin{proof}[Proof of Theorem \ref{thm:cleanslate}]
Let $E$ be the event that $p^i_t= p^C$ for every firm $i$ and every $t\in\{1,\dots,T-1\}$.

\emph{Claim 1: On $E$, $Q^i_t(\s,p)=0$ for every $i$, $t\leq T$, $\s\in\A^n$ and $p\neq p^C$.} By \eqref{eq:U} firm $i$ revises only the entry $(\s_t,p^i_t)$. On $E$ we have $p^i_t= p^C$ at every time $t\leq T-1$, so no entry $(\s,p)$ with $p\neq p^C$ is ever revised before time $T$. Since $Q^i_1\equiv0$, $Q^i_t(\s,p)=0$ for every $t\leq T$, $\s$ and $p\neq p^C$.

\emph{Claim 2: On $E$, $u_t^i:= Q^i_t(\p^C, p^C) \geq y^i_t\ $ for $2\leq t\leq T$.} We start by rewriting $y^i_t=\frac{\pi^i(\p^C)}{1-\delta_i}\,(1-P_t)$ for every $t\geq3$, where $P_t:=\prod_{l=2}^{t-1}c_l$ and $c_l:=1-\alpha_l(1-\delta_i)$, with $P_2=1$ (empty product). Recall $c_l>0$, since $\alpha_l, \delta_i\in (0,1)$ for each $l$, and $i$.

\emph{(a)} At $t=1$ on $E$, $(\s_1,p_1^i)=(\p_0,p^C)$. If $\p_0\neq\p^C$, $(\p^C,p^C)$ is not revised and $u_2^i=u_1^i=0$ (since $Q_1^i\equiv 0)$. If $\p_0=\p^C$, the revised entry is $(\p^C,p^C)$, and by \eqref{eq:U}, $u_2^i = \alpha_1\pi^i(\p^C) \geq 0 =y^i_2$.

\emph{(b)} For $2\leq t\leq T-1$, using $c_tP_t=P_{t+1}$ and $c_t+\alpha_t(1-\delta_i)=1$,
\[y^i_{t+1} = \frac{\pi^i(\p^C)}{1-\delta_i}\big(1-P_{t+1}\big) = \frac{\pi^i(\p^C)}{1-\delta_i}\Big(c_t-c_tP_t+\alpha_t(1-\delta_i)\Big) = 
c_t\,y^i_t+\alpha_t\,\pi^i(\p^C)
\]
Similarly, by Claim~1 on $E$, $\max_{a}Q^i_t(\p^C,a)=u_t$ for $2\leq t\leq T-1$. Then,  \eqref{eq:U} becomes the scalar recursion
\[
u_{t+1}^i=\big(1-\alpha_t(1-\delta_i)\big)u_t^i+\alpha_t\,\pi^i(\p^C), \qquad 2\leq t\leq T-1.
\]
It follows that $u_{t+1}^i-y^i_{t+1}=c_t\,(u_t^i-y^i_t)$ for $2\leq t\leq T-1$. It follows that $(u_t^i-y^i_t) = P_t(u_2^i-y_2^i)=P_tu_2^i$ for each $2\leq t\leq T$, and therefore
\[
u_t^i=y^i_t+u_2^i\,P_t,\qquad 2\leq t\leq T.
\]
Since $u_2^i\geq0$ and $P_t>0$, $u_t^i\geq y^i_t$. Also notice that when $u_2^i = 0$, $u_t^i = y_t^i$ for each $2\leq t\leq T$.

\emph{Claim 3: $\Prob(E)\geq\prod_{t=1}^{T-1}\prod_i\big(1+m e^{-y^i_t/\beta_t}\big)^{-1}$.} At time $t\leq T-1$, given $\s_t$, firm $i$ draws $p^C$ with softmax probability
\begin{equation}
    \label{eqn:softmax_onE}
    \frac{e^{Q^i_t(\s_t, p^C)/\beta_t}}{\sum_{a}e^{Q^i_t(\s_t,a)/\beta_t}}.
\end{equation}
On $E$ the denominator has $m$ terms equal to $e^{0}=1$ by Claim~1. Thus, \eqref{eqn:softmax_onE} equals $\big(1+m\,e^{-Q^i_t(\s_t, p^C)/\beta_t}\big)^{-1}$, which is increasing in $Q^i_t(\s_t, p^C)$. At $t=1$ all entries are zero and the probability is exactly $1/(m+1)=\big(1+m e^{-y^i_1/\beta_1}\big)^{-1}$. 

At $2\leq t\leq T-1$, $\s_t=\p^C$ and Claim~2 yields $Q^i_t(\p^C, p^C)=u_t\geq y^i_t$. Draws are independent across firms and, given the history, across times, so multiplying the conditional probabilities along the block path yields $\Prob(E)\geq\prod_{t=1}^{T-1}\prod_i\big(1+m e^{-y^i_t/\beta_t}\big)^{-1}$, with equality when $\p_0\neq\p^C$ by the equality in Claim~2. The bound is strictly positive because it is a finite product of strictly positive numbers.

\emph{Claim 4: On $E$, the absorption theorem applies.} On $E$ we have $\p_{T-1}=\p^C$,
hence $\s_T=\p^C$ and the two states $\p_{T-1}$ and $\p^C$ of the theorem's hypotheses
\emph{coincide}. Since $T\geq3$, $\alpha_2>0$ and $\pi^i(\p^C)>\pi^i(\p^*)\geq0$ by Assumption~\ref{assumption_deltai}, Claim~2 gives
$Q^i_T(\p^C, p^C)\geq y^i_T>0$, while Claim~1 gives $Q^i_T(\p^C,p)=0$ for $p\neq p^C$.
Therefore, for every $i$ and $p\neq p^C$,
\[
\text{(i)}\quad Q^i_T(\p^C, p^C)>0=Q^i_T(\p^C,p),
\qquad
\text{(ii)}\quad \pi^i(\p^C)>0=(1-\delta_i)\,Q^i_T(\p^C,p).
\]
By Theorem~\ref{prop:QConvergence_T},
$\p_t=\p^C$ for every $t\geq T$. Hence $E\subseteq\{\p_t=\p^C\ \forall t\geq T\}$ and
the bound of Claim~3 applies.
\end{proof}


\subsection{Proof of Proposition \ref{prop:naiveCollusionQ}}

Since $\p_{T-1}=\p^C$, the two states in the hypotheses of Theorem~\ref{prop:QConvergence_T} both equal $\p^C$. Assumption~(i) at $\s=\p^C$ is condition~(i) of that theorem, and Assumption~(ii) implies its condition~(ii). Hence $\p_t=\p^C$ for every $t\geq T$, and \eqref{eqn_Q_epsilon*} holds with its middle case vacuous. We now prove that $\w^*(\s) = \p^C$ for each $\s\in \A^n$, distinguishing the cases $\s=\p^C$ and $\s\neq \p^C$. We fix $i\in[n]$ for the entire proof.

\vspace{5pt}
$\bullet$ \textbf{Case 1:} $\s = \p_{T-1}= \p^C$. By \eqref{eqn_Q_epsilon*}, 
\begin{equation}
    \label{proof:prop_naivecollusion_1}
    Q^{i*}(\p^C,p) = \begin{cases}
         (1-\delta_i)^{-1} \pi^i(\p^C) & \textnormal{ if } p=p^C,\\
         Q_{T}^i(\p^C,p) & \textnormal{ if } p\neq p^C.\\
    \end{cases}
\end{equation}
By Assumption (ii) in Proposition \ref{prop:naiveCollusionQ}, $\pi^i(\p^C)>(1-\delta_i)Q_{T}^i(\p^C,p)$ for each $p\in\A\setminus\{p^C\}$. Multiplying both sides of the latter inequality by $(1-\delta_i)^{-1}$, \eqref{proof:prop_naivecollusion_1} yields  $Q^{i*}(\p^C,p^C) > Q^{i*}(\p^C,p)$ for each $p\in\A\setminus\{p^C\}$. Thus, $\argmax_{p\in\A}Q^{i*}(\p^C,p) = \{p^C\}$, which implies that $w^{i*}(\p^C) = p^C$.

\vspace{5pt}
$\bullet$ \textbf{Case 2:} $\s\neq \p_{T-1}$. By \eqref{eqn_Q_epsilon*}, $Q^{i*}(\s,p) = Q_{T}^i(\s,p)$ for each $p\in \A$. By Assumption (i) in Proposition \ref{prop:naiveCollusionQ}, $Q_{T}^i(\s,p^C)>Q_{T}^i(\s,p)$ for each  $p\in\A\setminus\{p^C\}$. It follows that $w^{i*}(\s) = p^C$.

\vspace{5pt}
Finally, we prove that $\w^*$ is a Nash equilibrium from time $t = 1$ if and only if $\p^C$ is a Nash equilibrium of the one-stage game $(\pi^i(\cdot))_{i=1}^n$. 

\vspace{5pt}
\textbf{Proof of the ``if'' direction:}  Suppose that $\p^C$ is a Nash equilibrium of the one-stage game $(\pi^i(\cdot))_{i=1}^n$. We use Algorithm \ref{algo:HowtoCheckNash} to show that $\w^*$ is a Nash equilibrium from time $t = 1$. By step (i) in Algorithm \ref{algo:HowtoCheckNash}, we first plug $\w^* = \p^C$ into equation \eqref{NashT1_1} and solve it as a linear system with unknowns $\v = (v_{i,\p_0})_{\p_0\in\A^n}$, as follows: 
\begin{equation}\label{proof:prop_naivecollusion_4}
\begin{split}
        v_{i,\p_0}& = \Vone(\w^*,\w^*,\v)_{i,\p_0}\\
        & \underbrace{=}_{\textnormal{By \eqref{V1_function}}} \pi^i(\p^C)+\delta_iv_{i,\p^C}.
    \end{split}
\end{equation}
Solving \eqref{proof:prop_naivecollusion_4} for $\v$, yields for each $\p_0\in\A^n$
\begin{equation}\label{proof:prop_naivecollusion_5}
v_{i,\p_0}=\frac{1}{1-\delta_i}
   \pi^i(\p^C).
\end{equation}
Following step (ii) of Algorithm \ref{algo:HowtoCheckNash} , we plug $\v$ and $\w^* = \p^C$ into \eqref{NashT1_2} to check if $\v$ is a fixed point of the operator $v_{i,\p_0}\mapsto\max_{\sigma_1^i\in\SSigma_1^i}\Vone(\w^*,\sigma_1^i,\v)_{i,\p_0}$. Indeed, by \eqref{V1_function} and \eqref{proof:prop_naivecollusion_5}, 
\begin{equation}\label{proof:prop_naivecollusion_6}
    \begin{split}
        \max_{\sigma_1^i\in\SSigma_1^i}\Vone(\w^*,\sigma_1^i,\v)_{i,\p_0} &= \max_{\sigma_1^i\in\SSigma_1^i} \sum_{p_1^i\in\A}\sigma_1^i(p_1^i|\p_0)[\pi^i(p_1^i,(\p^C)^{-i})+\delta_i v_{i, (p_1^i,(\p^C)^{-i}) }]\\
        &= \max_{\sigma_1^i\in\SSigma_1^i} \sum_{p_1^i\in\A}\sigma_1^i(p_1^i|\p_0)\left[\pi^i(p_1^i,(\p^C)^{-i})+ \frac{\delta_i}{1-\delta_i}
   \pi^i(\p^C) \right].
    \end{split}
\end{equation}
Since $\p^C$ is a Nash equilibrium of the one-stage game $(\pi^i(\cdot))_{i=1}^n$, the maximum in \eqref{proof:prop_naivecollusion_6} is achieved at $\sigma_1^i(p^C|\p_0)=1$ for each $p_1^i\in\A$. Thus, $$\max_{\sigma_1^i\in\SSigma_1^i}\Vone(\w^*,\sigma_1^i,\v)_{i,\p_0} = \frac{1}{1-\delta_i}
   \pi^i(\p^C) = v_{i,\p_0}.$$
By Algorithm \ref{algo:HowtoCheckNash}, $\w^*$ is a Nash equilibrium from time $t = 1$.

\vspace{5pt}
\textbf{Proof of the ``only if'' direction:} Suppose that $\w^*=\p^C$ is a Nash equilibrium from time $t = 1$. By definition \eqref{eqn:def_nash_eq}, for each $\p_0\in\A^n$
$$\w^{i*}(\p_0) = p^C\in\argmax_{\sigma_1^i\in\SSigma_1^i}\Vonet(\p_0, \sigma^i_1|\w^{-i*}).$$
By the above and equation \eqref{conditional_value_function_t1}, for each $\p_0\in\A^n$ and $\sigma_1^i\in\SSigma_1^i$
\begin{equation}\label{proof:prop_naivecollusion_7}
    \sum_{t=1}^\infty\delta_i^{t-1}\pi^i(\p^C)\geqq \E_{(\sigma_1^i,\w^{-i*})}\left[ \sum_{t=1}^\infty \delta_i^{t-1} \pi^i(\p_t)\Big|\p_0 \right].
\end{equation}
For each $\hat{p}\in\A\setminus\{p^C\}$, define $\hat{\sigma}_1^i$ as follows: $\hat{\sigma}_1^i(p|\p^{*}) = 1$ if $p = \hat{p}$, and $\hat{\sigma}_1^i(p|\p^{*}) = 0$ if $p \neq \hat{p}$. Moreover, let $\hat{\sigma}_1^i(p^C|\p_0) = 1$  for any $\p_0\neq \p^{*}$. Taking $\p_0 = \p^{*}$ and $\sigma_1^i = \hat{\sigma}_1^i$ in \eqref{proof:prop_naivecollusion_7} yields, 
\begin{equation*}
\begin{split}
    \frac{1}{1-\delta_i}
   \pi^i(\p^C) &\geqq \E_{\hat{\sigma}_1^i}\left[  \pi^i(p_1^i,(\p^C)^{-i})+\sum_{t=2}^\infty \delta_i^{t-1} \pi^i(p_t^i,(\p^C)^{-i})\Big|\p^{*} \right]\\
   & = \pi^i(\hat{p},(\p^C)^{-i}) + \E_{\hat{\sigma}_1^i}\left[  \sum_{t=2}^\infty \delta_i^{t-1} \pi^i(p_t^i,(\p^C)^{-i})\Big|(\hat{p},(\p^C)^{-i})\right]\\
   & = \pi^i(\hat{p},(\p^C)^{-i}) +  \frac{\delta_i}{1-\delta_i}
   \pi^i(\p^C).
\end{split}
\end{equation*}
The above inequality holds for each $\hat{p}\in\A\setminus\{p^C\}$ and $i\in[n]$, implying that $\p^C$ is a Nash equilibrium of the one-stage game $(\pi^i(\cdot))_{i=1}^n$.

\qed

\subsection{Proof of Proposition \ref{prop:grimtriggerQ}}

Since $\p_{T-1}=\p^C$, all the assumptions of Theorem \ref{prop:QConvergence_T} are satisfied, hence $\p_t = \p^C$ for all $t \geqq T$. The same theorem and $\p_{T-1} = \p^C$ yields,
\begin{equation}\label{eq:Qstarsimple}
Q^{i*}(\s,p)=
\begin{cases}
\dfrac{\pi^i(\p^C)}{1-\delta_i}, & (\s,p)=(\p^C, p^C),\\[6pt]
Q^i_T(\s,p), & \text{otherwise,}
\end{cases}
\end{equation}

We fix $i\in[n]$ for the rest of the proof. 

\vspace{5pt}
$\bullet$ \textbf{Case 1:} $\s = \p^C$. This case is identical to the case $\s = \p^C$ in the Proof of Proposition \ref{prop:naiveCollusionQ}, so we omit it. However, we recall that this case uses the Assumption (ii) in Proposition \ref{prop:grimtriggerQ}. Thus, $w^{i*}(\p^C) = p^C$.

\vspace{5pt}
$\bullet$ \textbf{Case 2:} $\s\neq \p^C$. By \eqref{eqn_Q_epsilon*}, $Q^{i*}(\s,p) = Q_{T}^i(\s,p)$ for each $p\in \A$. By Assumption (iii) in Proposition \ref{prop:grimtriggerQ}, $Q_{T}^i(\s,p^*)>Q_{T}^i(\s,p)$ for each $p\in\A\setminus\{p^*\}$. It follows that $w^{i*}(\s) = p^*$.

\vspace{5pt}
Finally, $\w^*=\ssigma_1^f$, so $(\sigma_0^f,\w^*)$ is the grim trigger strategy $\ssigma^f$, which is an SPE under Assumption~\ref{assumption_deltai_simple} by Proposition~\ref{prop:grimTriger}.

\qed

\subsection{Proof of Proposition \ref{prop:increasing_Q}}\label{proofpropincreasing_Q}
Since $\p_{T-1}=\p^C$, assumption~(i) gives conditions~(i) and~(ii) of Theorem~\ref{prop:QConvergence_T}. Hence $\p_t=\p^C$ for every $t\geq T$, and \eqref{eqn_Q_epsilon*} holds with its middle case vacuous. We now prove that
\begin{equation*}
\w^*(\s) = \begin{cases}
\p^C & \s = \p^C,\\
\p^{(l+1)} & \s = \p^{(l)},\ l\in\{0,\dots,k\},\\
\p^* & \s\notin \{\p^{(l)}\}_{l=0}^{k+1}.
\end{cases}
\end{equation*}
We split the proof into three cases: $\s = \p^C$; $\s=\p^{(j)}$ for some $j\in\{0,\dots,k\}$; and $\s\in\A^n\setminus\{\p^{(l)}\}_{l=0}^{k+1}$. We fix $i\in[n]$ for the entire proof.

\vspace{5pt}
$\bullet$ \textbf{Case 1:} $\s = \p^C$. using \eqref{proof:prop_naivecollusion_1} from the proof of Proposition \ref{prop:naiveCollusionQ}, and by assumption (i) in this proposition we obtain $w^{i*}(\p^C) = p^C$.

\vspace{5pt}
$\bullet$ \textbf{Case 2:} $\s=\p^{(j)}$ for some $j\in\{0,\dots, k\}$. By \eqref{eqn_Q_epsilon*}, $Q^{i*}(\p^{(j)},p) = Q_{T}^i(\p^{(j)},p)$ for each $p\in \A$. By Assumption (ii) in this proposition, $Q_{T}^i(\p^{(j)},p^{(j+1)})>Q_{T}^i(\p^{(j)},p)$ for each  $p\in\A\setminus\{p^{(j+1)}\}$. It follows that $w^{i*}(\p^{(j)}) = p^{(j+1)}$.

\vspace{5pt}
$\bullet$ \textbf{Case 3:} $\s\in\A^n\setminus\{\p^{(l)}\}_{l=0}^{k+1}$. By assumption (iii) in this proposition, $w^{i*}(\s) = p^{*}$.

\qed


\renewcommand{\theequation}{B\arabic{equation}}
\setcounter{equation}{0}
   \section{Absorption and Emergence at an Arbitrary Price}\label{app:anyprice}

Theorems~\ref{prop:QConvergence_T} and~\ref{thm:cleanslate} are stated for the collusive-enabling price. This appendix shows that neither uses any property of $p^C$ beyond the positivity of its profits (see the proof of Theorem~\ref{thm:cleanslate}). Fix $\hat p\in\A$ and write $\hat\p:=(\hat p,\dots,\hat p)\in\A^n$. Throughout, part~(i) of Assumption~\ref{assumption_deltai} holds, so that the state at time $t$ is $\s_t=\p_{t-1}$, and firms follow Algorithm~\ref{Algo_Q_boundedExperimentation}. Part~(ii) of Assumption~\ref{assumption_deltai}, which defines $p^*$ and $p^C$, is used only in Corollaries~\ref{cor:selection} and~\ref{cor:complockin}.

\begin{proposition}[Absorption at an arbitrary price]\label{prop:absorb_any}
Suppose that Assumption~\ref{assumption_learningrate} holds and that, for each $i\in[n]$, each $p\in\A\setminus\{\hat p\}$ and each $\s\in\{\p_{T-1},\hat\p\}$,
\begin{itemize}
\item[(i)] $Q^i_T(\s,\hat p)>Q^i_T(\s,p)$;
\item[(ii)] $\pi^i(\hat\p)\geq(1-\delta_i)\,Q^i_T(\hat\p,p)$.
\end{itemize}
Then $\p_t=\hat\p$ for every $t\geq T$, and $\lim_{t\to\infty}Q^i_t(\s,p)$ is given by \eqref{eqn_Q_epsilon*} with $p^C$ and $\p^C$ replaced by $\hat p$ and $\hat\p$.
\end{proposition}

\begin{proof}
For each firm $i$ write $u^i_t:=Q^i_t(\hat\p,\hat p)$ and $q^i:=\max_{p\neq\hat p}Q^i_T(\hat\p,p)$. Condition~(i) at $\s=\hat\p$ gives $u^i_T>q^i$, and condition~(ii) gives $\pi^i(\hat\p)\geq(1-\delta_i)q^i$.

\emph{Step 1 (time $T$).} By condition~(i) at $\s=\p_{T-1}$, $\argmax_{a\in\A}Q^i_T(\p_{T-1},a)=\{\hat p\}$ for every $i$, so $\p_T=\hat\p$ and $\s_{T+1}=\hat\p$. The entry that firm $i$ revises at time $T$ is $(\p_{T-1},\hat p)$, and its target is $\pi^i(\hat\p)+\delta_i\max_{a}Q^i_T(\hat\p,a)=\pi^i(\hat\p)+\delta_iu^i_T$, by condition~(i) at $\s=\hat\p$. Hence, if $\p_{T-1}\neq\hat\p$, then $u^i_{T+1}=u^i_T$ and the only changed entry is the one in the middle case of \eqref{eqn_Q_epsilon*}; if $\p_{T-1}=\hat\p$, then $u^i_{T+1}=(1-\alpha_T(1-\delta_i))u^i_T+\alpha_T\pi^i(\hat\p)$. In both cases the entries $Q^i(\hat\p,p)$ with $p\neq\hat p$ are unchanged.

\emph{Step 2 (one inequality).} If $u>q$, $\pi^i(\hat\p)\geq(1-\delta_i)q$ and $\alpha\in(0,1)$, then
\[
\big(1-\alpha(1-\delta_i)\big)u+\alpha\,\pi^i(\hat\p)\ >\ \big(1-\alpha(1-\delta_i)\big)q+\alpha(1-\delta_i)q\ =\ q,
\]
because $1-\alpha(1-\delta_i)>0$.

\emph{Step 3 (induction).} We claim that for every $t\geq T+1$ and every firm $i$: $\s_t=\hat\p$, $u^i_t>q^i$, and $Q^i_t(\hat\p,p)=Q^i_T(\hat\p,p)$ for every $p\neq\hat p$. At $t=T+1$ this follows from Steps~1 and~2. Suppose it holds at $t$. Then $\hat p$ is the unique maximizer of $Q^i_t(\hat\p,\cdot)$ for every $i$, so $\p_t=\hat\p$ and $\s_{t+1}=\hat\p$; the only entry that firm $i$ revises is $(\hat\p,\hat p)$, with target $\pi^i(\hat\p)+\delta_iu^i_t$. Hence $u^i_{t+1}=(1-\alpha_t(1-\delta_i))u^i_t+\alpha_t\pi^i(\hat\p)>q^i$ by Step~2, and the entries $Q^i(\hat\p,p)$ with $p\neq\hat p$ are unchanged. Together with Step~1, this shows that $\p_t=\hat\p$ for every $t\geq T$.

\emph{Step 4 (limit).} For $t\geq T+1$, iterating $u^i_{t+1}=(1-\alpha_t(1-\delta_i))u^i_t+\alpha_t\pi^i(\hat\p)$ gives $u^i_{t+1}=P_t\,u^i_{T+1}+S_t\,\pi^i(\hat\p)$, with $P_t$ and $S_t$ as in Proposition~\ref{prop:identity} for $\delta=\delta_i$. By part~(b) of that proposition and Assumption~\ref{assumption_learningrate}, $P_t\to 0$ and $S_t\to 1/(1-\delta_i)$, so $u^i_t\to\pi^i(\hat\p)/(1-\delta_i)$. Every other entry is constant from time $T+1$ on. This is \eqref{eqn_Q_epsilon*} with $p^C$ replaced by $\hat p$.
\end{proof}

Taking $\hat p=p^C$ in Proposition~\ref{prop:absorb_any} recovers Theorem~\ref{prop:QConvergence_T}.

\begin{theorem}[Emergence at an arbitrary price]\label{thm:cleanslate_any}
Suppose that Assumption~\ref{assumption_learningrate} holds, $Q^i_1\equiv 0$ for every $i$, the temperatures satisfy $\beta_t>0$, the deadline satisfies $T\geq 3$, and $\pi^i(\hat\p)>0$ for every $i$. Define $y^i_1(\hat p):=y^i_2(\hat p):=0$ and
\[
y^i_t(\hat p):=\frac{\pi^i(\hat\p)}{1-\delta_i}\Big[1-\prod_{l=2}^{t-1}\big(1-\alpha_l(1-\delta_i)\big)\Big]\qquad(t\geq 3).
\]
Then
\[
\Prob\big(\p_t=\hat\p\ \text{for all}\ t\geq T\big)\ \geq\ L(\hat p):=\prod_{t=1}^{T-1}\prod_{i=1}^{n}\frac{1}{1+m\,e^{-y^i_t(\hat p)/\beta_t}}\ >\ 0 .
\]
If moreover $\p_0\neq\hat\p$, then $L(\hat p)$ is exactly the probability that every firm charges $\hat p$ at every time $t\in\{1,\dots,T-1\}$.
\end{theorem}

\begin{proof}
Let $E(\hat p)$ be the event that $p^i_t=\hat p$ for every $i$ and every $t\in\{1,\dots,T-1\}$. Claims~1-3 of the proof of Theorem~\ref{thm:cleanslate} (Appendix~\ref{sec:proofs_emergence}) use only four facts: firm $i$ revises only the entry it visits, by \eqref{eq:U}; $Q_1\equiv 0$ and $\pi^i\geq 0$, so every $Q$-value is non-negative; the scalar recursion that \eqref{eq:U} becomes at the revised entry; and the softmax formula \eqref{eqn:softmax_prob}. None of them refers to $p^C$. Repeating those claims with $p^C$ replaced by $\hat p$ shows that, on $E(\hat p)$, $Q^i_t(\s,p)=0$ for every $t\leq T$, every $\s$ and every $p\neq\hat p$; that $Q^i_t(\hat\p,\hat p)\geq y^i_t(\hat p)$ for $2\leq t\leq T$, with equality if $\p_0\neq\hat\p$; and that $\Prob(E(\hat p))\geq L(\hat p)$, with equality if $\p_0\neq\hat\p$.

The only place where the proof of Theorem~\ref{thm:cleanslate} uses a property of $p^C$ is Claim~4, through the positivity of $\pi^i(\p^C)$; here positivity is a hypothesis. On $E(\hat p)$ we have $\p_{T-1}=\hat\p$, so the two states in Proposition~\ref{prop:absorb_any} both equal $\hat\p$. Since $T\geq 3$, $\alpha_2>0$ and $\pi^i(\hat\p)>0$, we have $y^i_T(\hat p)>0$, hence $Q^i_T(\hat\p,\hat p)>0=Q^i_T(\hat\p,p)$ for every $p\neq\hat p$, which is condition~(i), and $\pi^i(\hat\p)>0=(1-\delta_i)Q^i_T(\hat\p,p)$, which is condition~(ii). Proposition~\ref{prop:absorb_any} gives $E(\hat p)\subseteq\{\p_t=\hat\p\ \text{for all}\ t\geq T\}$, and the bound follows.
\end{proof} 

\begin{remark}[Why positive profits are needed]\label{rem:positivity}
If $\pi^i(\hat\p)=0$ for some firm $i$, then on $E(\hat p)$ every target of firm $i$ is zero, all of its entries remain zero, condition~(i) fails because of a tie, and the uniform tie-breaking of Algorithm~\ref{Algo_Q_argmax} at time $T$ can move play away from $\hat\p$. This is the case of Bertrand competition with constant marginal cost when the competitive price equals that cost, so that Nash profits are zero.
\end{remark}

\begin{corollary}[Stochastic selection]\label{cor:selection}
Suppose that Assumptions~\ref{assumption_deltai} and~\ref{assumption_learningrate} hold, $Q^i_1\equiv0$ for every $i$, $\beta_t>0$, $T\geq 4$, and $\pi^i(\p^*)>0$ for every $i$. Then the disjoint events $\{\p_t=\p^C\ \text{for all}\ t\geq T\}$ and $\{\p_t=\p^*\ \text{for all}\ t\geq T\}$ have probabilities at least $L(p^C)>0$ and $L(p^*)>0$, respectively, and $L(p^C)>L(p^*)$.
\end{corollary}

\begin{proof}
By Assumption~\ref{assumption_deltai}(ii), $\pi^i(\p^C)>\pi^i(\p^*)>0$, so Theorem~\ref{thm:cleanslate_any} applies at both prices; the events are disjoint because $\p_T$ cannot equal both profiles. For $t\geq 3$ the bracket in the definition of $y^i_t$ is positive, because $\alpha_2>0$, so $y^i_t(\hat p)$ is strictly increasing in $\pi^i(\hat\p)$, and each factor of $L(\hat p)$ is strictly increasing in $y^i_t(\hat p)$. For $t\in\{1,2\}$ the factors equal $1/(m+1)$ at both prices. Since $T\geq 4$, the factor with $t=3$ appears in the product, and $L(p^C)>L(p^*)$ follows.
\end{proof}

\begin{corollary}[The initialization channel at an arbitrary price]\label{cor:complockin}
Proposition~\ref{prop:capture}, Lemma~\ref{cor:basin} and Proposition~\ref{prop:random} hold with $p^C$ and $\p^C$ replaced by any $\hat p\in\A$ and $\hat\p$, with $\pi^i(\p^C)$ replaced by $\pi^i(\hat\p)$ wherever it appears: in condition~(ii) of Proposition~\ref{prop:capture}, in the definition of $B$ and in \eqref{eqn:nonnull_AT}, and in \eqref{condition_start}. In particular, under Assumption~\ref{assumption_deltai}(ii), the set of $K_0$ satisfying \eqref{condition_start} at $p^*$ is contained in the set satisfying it at $p^C$, and the containment is strict whenever the former is nonempty.
\end{corollary}

\begin{proof}
The proofs of Proposition~\ref{prop:capture}, Lemma~\ref{cor:basin} and Proposition~\ref{prop:random} use $p^C$ only through the conditions of Theorem~\ref{prop:QConvergence_T}, which they verify at time $T$ by means of Lemma~\ref{lem:budget}, and Lemma~\ref{lem:budget} does not depend on any price. Replacing Theorem~\ref{prop:QConvergence_T} by Proposition~\ref{prop:absorb_any} gives the claim. The comparison follows from $\pi^i(\p^C)>\pi^i(\p^*)$ for every $i$, which raises the upper bound $\min_i\pi^i(\hat\p)/(1-\delta_i)-R_1$ in \eqref{condition_start}; since this bound at $p^*$ lies strictly below $\pibar/(1-\dbar)$, every $K_0$ admissible at $p^*$ is admissible at $p^C$, and values slightly above the bound at $p^*$ are admissible only at $p^C$.
\end{proof}

\begin{remark}[Capturing a pattern, not a point]\label{rem:pattern} Let $w:\A^n\to\A$ be any pattern. If at some $\tau\leq T$, $Q^i_\tau(\s,w(\s))-Q^i_\tau(\s,p)>2R_\tau$ for every $i$, every $\s$ and every $p\neq w(\s)$, then by the argument of Proposition~\ref{prop:capture} the greedy policy of every firm at time $T$ is $w$ at every state. This is the kind of hypothesis that Propositions~\ref{prop:grimtriggerQ} and~\ref{prop:increasing_Q} need at states other than $\p^C$, and it explains where such hypotheses come from. It does \emph{not} by itself produce absorption at $\p^C$: for that one needs, in addition, $w(\p^C)=p^C$, the level condition at $\p^C$, and $w(\p_{T-1})=p^C$. For the patterns of Propositions~\ref{prop:grimtriggerQ} and~\ref{prop:increasing_Q}, which prescribe prices other than $p^C$ off $\p^C$, the last requirement amounts to the deadline draw $\p_{T-1}=\p^C$.
\end{remark}

\subsection*{Decaying exploration}
We make Remark~\ref{rem:decaying} precise. Let $(U^i_t)_{i\in[n],\,t\geq1}$ be independent Bernoulli random variables with $\mathbb{P}(U^i_t=1)=\varepsilon_t$, where $(\varepsilon_t)_{t\geq1}\subset[0,1]$, independent of all other randomness. Under \emph{$\varepsilon$-greedy exploration}, firms follow Algorithm~\ref{Algo_Q_argmax} at every time $t\geq1$, except that in step~(a) firm $i$ draws $p^i_t$ uniformly from $\A$ when $U^i_t=1$ and according to \eqref{eqn:argmax_prob} when $U^i_t=0$. Let
\begin{equation*}
\tau:=1+\sup\{t\geq1:\ U^i_t=1 \text{ for some } i\in[n]\},\qquad \sup\emptyset:=0,
\end{equation*}
so that no firm experiments at any time $t\geq\tau$.

\begin{proposition}[Absorption under decaying $\varepsilon$-greedy exploration]\label{prop:epsgreedy}
Suppose that Assumptions~\ref{assumption_deltai} and~\ref{assumption_learningrate} hold, firms follow Algorithm~\ref{Algo_Q_argmax} with $\varepsilon$-greedy exploration, and $\sum_{t\geq1}\varepsilon_t<\infty$. Then:
\begin{itemize}
\item[(a)] $\tau<\infty$ almost surely.
\item[(b)] On the event $\{\tau<\infty\}$, if for each $i\in[n]$, $p\in\A\setminus\{p^C\}$ and $\s\in\{\p_{\tau-1},\p^C\}$,
\begin{itemize}
\item[(i)] $Q^i_\tau(\s,p^C)>Q^i_\tau(\s,p)$;
\item[(ii)] $\pi^i(\p^C)\geq(1-\delta_i)\,Q^i_\tau(\p^C,p)$,
\end{itemize}
then $p^i_t=p^C$ for every $i\in[n]$ and every $t\geq\tau$, and \eqref{eqn_Q_epsilon*} holds with $T$ replaced by $\tau$.
\end{itemize}
\end{proposition}

\begin{proof}
(a) Since $\sum_{t\geq1}\sum_{i\in[n]}\mathbb{P}(U^i_t=1)=n\sum_{t\geq1}\varepsilon_t<\infty$, the Borel--Cantelli lemma implies that, with probability one, $U^i_t=1$ for only finitely many pairs $(i,t)$. Hence $\tau<\infty$ almost surely.

(b) The time $\tau$ is not a stopping time, so we argue path by path. Fix a sample path on which $\tau<\infty$ and conditions (i) and (ii) hold at time $\tau$. From time $\tau$ on, every firm selects its price according to \eqref{eqn:argmax_prob}, exactly as in Algorithm~\ref{Algo_Q_boundedExperimentation} from time $T$ on. The proof of Theorem~\ref{prop:QConvergence_T} in Appendix~\ref{Appendix_RLalgo} uses only four facts: that every firm selects its price by \eqref{eqn:argmax_prob} from time $T$ on, the update rule \eqref{def:Q_learningUpdate}, conditions (i) and (ii) at time $T$, and Lemma~\ref{lem:alphadelta}, which requires $\sum_{t>T}\alpha_t=\infty$. With $T$ replaced by $\tau$, the first three hold on the fixed path by construction and by assumption. The fourth holds because $\alpha_t\in(0,1)$ for every $t$, so Assumption~\ref{assumption_learningrate} implies $\sum_{t\geq t_0}\alpha_t=\infty$ for every finite $t_0$. The argument of Appendix~\ref{Appendix_RLalgo} therefore applies on the path with $T$ replaced by $\tau$, and it yields $p_t=p^C$ for every $t\geq\tau$ together with \eqref{eqn_Q_epsilon*} with $T$ replaced by $\tau$.
\end{proof}

\begin{remark}[Long exploration]\label{rem:longexploration}
Proposition~\ref{prop:epsgreedy} takes conditions~(i) and~(ii) at time $\tau$ as given, and long exploration matters for how they arise. Theorem~\ref{thm:cleanslate} obtains them on a path along which the deviation entries $Q^i_t(\p^C,p)$, $p\neq p^C$, are never revised. If instead such an entry is revised many times while the rivals charge $p^C$, each revision moves it toward $\pi^i(\p')+\delta_i\max_{a\in\A}Q^i_t(\p',a)$ with $\p':=(p,(\p^C)^{-i})$, and condition~(ii) then requires, approximately, that
\begin{equation*}
\frac{\pi^i(\p^C)}{1-\delta_i}\ \geq\ \pi^i(\p')+\delta_i\max_{a\in\A}Q^i_\tau(\p',a),
\end{equation*}
that is, the learned value of the state reached after a deviation must be low enough to offset the one-period gain from undercutting. This is the incentive condition of a punishment scheme, as in Assumption~\ref{assumption_deltai_simple} for the grim trigger strategy, and the emergence results of Section~\ref{sec:emergence}, which rely on the deviation entries never being revised (Theorem~\ref{thm:cleanslate}) or changing by at most $R_1$ during experimentation (Proposition~\ref{prop:random}), do not apply to it.
\end{remark}


\renewcommand{\theequation}{C\arabic{equation}}
\setcounter{equation}{0}
\section{Rewriting the Proof of Fink's Theorem}\label{Sec:FinksProof}

We rewrite the proof of Theorem 2 of \cite{fink1964equilibrium}, that is, Theorem \ref{Thm:Fink} in this work. The rewritten proof uses our notation and adds many missing details. We find it necessary to refer to the rewritten proof when establishing the theories of Sections \ref{sec:Repeated_Bertrand} and \ref{sec:RLalgo}.  Section \ref{sec:proof_prop_1} first proves Proposition \ref{prop:BellmanEqV1} and Section \ref{sec:proof_thm_1} establishes several other propositions and then concludes the proof of Theorem \ref{Thm:Fink}.

\subsection{Proof of Proposition \ref{prop:BellmanEqV1}}\label{sec:proof_prop_1}
Let $\ssigma_1 = (\sigma^i_1, \ssigma^{-i}_1)\in  \SSigma_1$, $s_1\in \S$ and $\p_0\in \A^n$ be given. From \eqref{conditional_value_function_t1}, 
\begin{equation}\label{cvf_t1_proof1}
\begin{split}
    &\Vonet(s_1, \p_0, \sigma^i_1|\ssigma^{-i}_1) =\\
    &\sum_{\p_1\in\A^n}\sigma_1(\p_1|\p_0,s_1) \left\{ \pi^i(\p_1,s_1) +
    \delta_i\sum_{s_2\in\S}\P(s_2|\p_1,s_1)
    \E_{\ssigma_1, \P}\left[ \sum_{t=2}^\infty \delta_i^{t-2} \pi^i(\p_t,s_t)\Big|\p_1, s_2 \right]\right\}.
\end{split}
\end{equation}
To obtain \eqref{Bellman_equationV1} from \eqref{cvf_t1_proof1}, note that the profit function $\pi^i$ is time independent, which implies that $$\E_{\ssigma_1, \P}\left[ \sum_{t=2}^\infty \delta_i^{t-2} \pi^i(\p_t,s_t)\Big|\p_1, s_2 \right] = \Vonet(s_2, \p_1, \sigma^i_1|\ssigma^{-i}_1).$$ 

We now show that there exists a unique solution to \eqref{Bellman_equationV1}. Expanding \eqref{Bellman_equationV1}, for each $s_1\in\S$ and $\p_0\in\A^n$, we obtain the following
\begin{equation*}
    \begin{split}
    &\Vonet(s_1, \p_0, \sigma^i_1|\ssigma^{-i}_1)\\
    &=\E_{\ssigma_1}\left[ \pi^i|\p_0,s_1 \right]+
    \delta_i\sum_{\p_1\in\A^n}\sigma_1(\p_1|\p_0,s_1) \sum_{s_2\in\S}\P(s_2|\p_1,s_1)
    \Vonet(s_2, \p_1, \sigma^i_1|\ssigma^{-i}_1), 
    \end{split}
\end{equation*}
which can be rewritten as 
\begin{equation}\label{cvf_t1_proof2}
    \begin{split}
    &\left[1-\delta_i\sigma_1(\p_0|\p_0,s_1)\P(s_1|\p_0,s_1)\right]\Vonet(s_1, \p_0, \sigma^i_1|\ssigma^{-i}_1)\\
    &-\delta_i\sigma_1(\p_0|\p_0,s_1) \sum_{s_2\neq s_1}\P(s_2|\p_0,s_1)
    \Vonet(s_2, \p_0, \sigma^i_1|\ssigma^{-i}_1)\\
    &-\delta_i\sum_{\p_1\neq \p_0}\sigma_1(\p_1|\p_0,s_1) \sum_{s_2\in\S}\P(s_2|\p_1,s_1)
    \Vonet(s_2, \p_1, \sigma^i_1|\ssigma^{-i}_1)=\E_{\ssigma_1}\left[ \pi^i|\p_0,s_1 \right]. 
    \end{split}
\end{equation}
Let $\E_{\ssigma_1}[\pi^i] := (\E_{\ssigma_1}[\pi^i|\p^1,s^1],\cdots,\E_{\ssigma_1}[\pi^i|\p^M,s^r])^T\in\mathbb{R}^{rM}$. By \eqref{cvf_t1_proof2}, the vector $\Vonet(\sigma^i_1|\ssigma^{-i}_1)$ given by
\begin{equation}\label{conditional_value_function_t1_vector}
    \Vonet(\sigma^i_1|\ssigma^{-i}_1) := (\Vonet(s^1, \p^1, \sigma^i_1|\ssigma^{-i}_1),\cdots, \Vonet(s^r, \p^M, \sigma^i_1|\ssigma^{-i}_1))^T\in \mathbb{R}^{rM}
\end{equation}
satisfies the following linear system
\begin{equation}
    \label{Av1_system_proof}
    \begin{split}
        \bA \Vonet(\sigma^i_1|\ssigma^{-i}_1) = \E_{\ssigma_1}[\pi^i],
    \end{split}
\end{equation}
where $\bA$ is a matrix whose rows and columns are indexed by the set $\S\times \A^n$: the entry in row $(s^j,\p^k)$ and column $(s^l,\p^o)$ is given by 
\begin{equation}
    \label{Av1_system_proof2}
    \begin{split}
        \bA \left((s^j,\p^k),(s^l,\p^o)\right) = \begin{cases}
            1-\delta_i\sigma_1(\p^k|\p^k,s^j)\P(s^j|\p^k,s^j) & \textnormal{ if } (s^j,\p^k) = (s^l,\p^o)\\
            -\delta_i\sigma_1(\p^o|\p^k,s^j)\P(s^l|\p^o,s^j) & \textnormal{ if } (s^j,\p^k) \neq (s^l,\p^o)
        \end{cases}
    \end{split}.
\end{equation}
For each $(s^j,\p^k)\in \S\times \A^n$, the following holds true 
\begin{equation}
    \label{Av1_system_proof3}
    \begin{split}
        &\bA \left((s^j,\p^k),(s^j,\p^k)\right) - \sum_{(s^l,\p^o)\neq (s^j,\p^k)} \left| \bA \left((s^j,\p^k),(s^l,\p^o)\right) \right| \\
       &=1-\delta_i\sigma_1(\p^k|\p^k,s^j)\P(s^j|\p^k,s^j) -\sum_{(s^l,\p^o)\neq (s^j,\p^k)}\delta_i\sigma_1(\p^o|\p^k,s^j)\P(s^l|\p^o,s^j) \\
       &= 1-\delta_i\sum_{(s^l,\p^o)}\sigma_1(\p^o|\p^k,s^j)\P(s^l|\p^o,s^j)\\
       &=1-\delta_i\sum_{\p^o}\sigma_1(\p^o|\p^k,s^j)\sum_{s^l}\P(s^l|\p^o,s^j) = 1-\delta_i
    \end{split}.
\end{equation}
From Gershgorin Circle Theorem (See page 244 in \cite{bhatia2013matrix}), for any eigenvalue of $\bA$, say $\lambda$, there exists $(s^j,\p^k)\in \S\times \A^n$ such that
$$|\lambda-\bA \left((s^j,\p^k),(s^j,\p^k)\right)|\leq \sum_{(s^l,\p^o)\neq (s^j,\p^k)} \left| \bA \left((s^j,\p^k),(s^l,\p^o)\right) \right|. $$
The above inequality combined with the reverse triangle inequality and equation \eqref{Av1_system_proof3}, imply that $|\lambda|\geqq 1-\delta_i>0$. Thus, $0$ is not an eigenvalue of $\bA$ and $\bA^{-1}$ exists. Therefore,  \eqref{Av1_system_proof} has a unique solution. \qed


\subsection{Proof of Theorem \ref{Thm:Fink}}\label{sec:proof_thm_1}
Before getting into the details of the proof, we summarize some of the crucial steps in the proof of \cite{fink1964equilibrium}:
\begin{enumerate}
    \item 
 $\Vone$ is continuous in its domain of definition (see Proposition \ref{prop:propertiesV1}). 
 \item 
 For each $\v\in\R^{nrM}$ and $\ssigma_1\in\SSigma_1$, there is a well-defined mapping $(\v,\ssigma_1)\mapsto T(\v,\ssigma_1)$ whose $(i,s_1,\p_0)$-coordinate is given by 
$$T(\v,\ssigma_1)_{i,s_1,\p_0} = \max_{\tau_1^i\in\SSigma_1^i}\Vone(\ssigma_1,\tau_1^i,\v)_{i,s_1,\p_0}.$$
The mapping $\v\mapsto T(\v,\ssigma_1)$ is a contraction from $\R^{nrM}$ to itself (see Proposition \ref{prop:T_contraction}). Thus, there is a well-defined mapping $\ssigma_1\mapsto b(\ssigma_1)\in\R^{nrM}$, where $b(\ssigma_1)$ is the unique fixed point of $T(\cdot,\ssigma_1)$. 
\item The set-valued mapping $\Gamma:\SSigma_1\to 2^{\SSigma_1}$ given by $\ssigma_1\mapsto \Gamma(\ssigma_1)$, where
$$ \Gamma(\ssigma_1) := \{\ttau_1\in\SSigma_1| b(\ssigma_1) = \Vone(\ssigma_1,\ttau_1,b(\ssigma_1)) \},$$
satisfies the hypotheses of Kakutani's theorem (see Theorem \ref{Thm:Kakutani} and Proposition \ref{prop_b_properties}). Therefore, $\Gamma$ has a fixed point $\ssigma_1^{*}\in\SSigma_1$, i.e., there is a policy in $\SSigma_1$ such that $\ssigma_1^{*}\in\Gamma(\ssigma_1^{*})$. Such policy is the stationary point of Theorem \ref{Thm:Fink}. Moreover, the vector $\v^*$ from Theorem \ref{Thm:Fink} is given by $\v^* = b(\ssigma_1^*)$. 
\end{enumerate}

\subsubsection*{Preliminary Results and Definitions for the Proof of Theorem \ref{Thm:Fink}.} Given two nonempty sets $X$ and $Y$, a correspondence from $X$ to $Y$ is a map $\Gamma:X\longrightarrow 2^{Y}$ such that for each $x\in X$, $\Gamma(x)\neq \emptyset$. We say that $\Gamma$ is a self-correspondence on $X$, if $\Gamma$ is a correspondence from $X$ to $X$. If $Y\subset \R^d$ and $\Gamma(x)$ is convex for each $x\in X$, then we say that $\Gamma$ is convex-valued. Let $X$ and $Y$ be two metric spaces, $\Gamma$ is said to be closed-valued if $\Gamma(x)$ is a closed subset of $Y$. Now, $\Gamma$ is said to be closed at $x\in X$, if for any two sequences $(x_k)_k\subset X$ and $(y_k)_k\subset Y$ with $x_k\to x$ and $y_k\to y\in Y$, if $y_k\in \Gamma(x_k)$ for each $k$, then $y\in \Gamma(x)$. Moreover, $\Gamma$ has a closed graph if it is closed at every $x\in X$. 

\begin{theorem}[Kakutani's Fixed Point Theorem]\label{Thm:Kakutani}
Let $X\subset \R^d$ be a nonempty, compact and convex set. If $\Gamma$ is a convex-valued self-correspondence on $X$ that has a closed graph, then $\Gamma$ has a fixed point, i.e., there exists $x\in X$ with $x\in \Gamma(x)$.  
\end{theorem}
For a proof of Kakutani's fixed point theorem see Page 331 in \cite{ok2007real}. Proposition \ref{prop:propertiesV1}, Proposition \ref{prop:T_contraction} and Proposition \ref{prop_b_properties} below ensure that we can use Kakutani's fixed point theorem to prove Theorem \ref{Thm:Fink}.

\begin{proposition}[Properties of $\Vone$]\label{prop:propertiesV1} The function $\Vone$ as given by \eqref{V1_function} satisfies all of the following:
\begin{itemize}
    \item[(a)] $\Vone$ is continuous on $\SSigma_1\times \SSigma_1\times\mathbb{R}^{nrM}$;
    \item[(b)] Let $\ssigma_1, \ttau_1\in \SSigma_1$ and $\delta := \max_{i\in[n]} \delta_i$. For each $\v,\u\in \mathbb{R}^{nrM}$, and each $(i,s_1,\p_0)$-coordinate 
    $$ \Vone(\ssigma_1,\tau_1^i,\v)_{i,s_1,\p_0} - \Vone(\ssigma_1,\tau_1^i,\u)_{i,s_1,\p_0}\leq \delta |\v-\u|_{\infty}, $$
    where $|\cdot|_{\infty}$ denotes the infinity norm in $\R^{nrM}$; 
    \item[(c)] $\Vone(\ssigma_1,\ttau_1,\v)$ is linear in $\ttau_1$. 
\end{itemize}
    
\end{proposition}

\begin{proof}[{\bf Proof of Proposition \ref{prop:propertiesV1}}] 
Let $\ssigma_1, \ttau_1\in\SSigma_1$ and $\v\in \mathbb{R}^{nrM}$. From \eqref{V1_function}, for each $(i,s_1,\p_0)$-coordinate
\begin{equation}\label{proof:PropertiesV11}
\begin{split}
&\Vone(\ssigma_1,\tau_1^i,\v)_{i,s_1,\p_0}\\
&= \sum_{\p_1\in\A^n}\tau_1^i(p_1^i|\p_0,s_1)\sigma_1^{-i}(\p_1^{-i}|\p_0,s_1)\left[\pi^i(\p_1,s_1)+\delta_i\sum_{s_2\in\S}\P(s_2|\p_1,s_1)v_{i,s_2,\p_1}\right].
\end{split}
\end{equation}
From \eqref{proof:PropertiesV11}, it is straightforward to see that $\Vone(\ssigma_1,\tau_1^i,\v)_{i,s_1,\p_0}$ is continuous w.r.t~$(\tau_1^i,\ssigma_1^{-i})$, and continuous w.r.t.~$v_{i,s_2,\p_1}$ for all $(i,s_2,\p_1)$. Similarly, from \eqref{proof:PropertiesV11} it is not difficult to see that $\Vone(\ssigma_1,\tau_1^i,\v)_{i,s_1,\p_0}$ is linear w.r.t. $\tau_1^i$. Thus, proving (a) and (c). For (b), we estimate 
\begin{equation*}
\begin{split}
    &\Vone(\ssigma_1,\tau_1^i,\v)_{i,s_1,\p_0} - \Vone(\ssigma_1,\tau_1^i,\u)_{i,s_1,\p_0}\\ &=\delta_i\sum_{\p_1\in\A^n}\tau_1^i(p_1^i|\p_0,s_1)\sigma_1^{-i}(\p_1|\p_0,s_1)\sum_{s_2\in\S}\P(s_2|\p_1,s_1)[v_{i,s_2,\p_1}-u_{i,s_2,\p_1}]\\
    &\leq \delta \max_{j,s_2,\p_1}|v_{j,s_2,\p_1}-u_{j,s_2,\p_1}|.
\end{split}
\end{equation*}

\end{proof}

\textbf{The $T$ mapping:} From \eqref{def:Sigma_ti}, we know that $\SSigma_1^i$ is a compact subset of $\mathbb{R}^{(m+1)rM}$. By Proposition \ref{prop:propertiesV1}, $\Vone$ is a continuous function. Based on these two observations, it makes sense to define the following mapping:   $$T:\mathbb{R}^{nrM}\times \SSigma_1 \longrightarrow \mathbb{R}^{nrM} \textnormal{ s.t. }(\v,\ssigma_1) \mapsto T(\v,\ssigma_1)$$ where the $(i,s_1,\p_0)$-coordinate of $T(\v,\ssigma_1)$ is given by 
\begin{equation}\label{def:T_vsigma}
\begin{split}
&T(\v,\ssigma_1)_{i,s_1,\p_0}:= \max_{\tau_1^i\in\SSigma_1^i} \Vone(\ssigma_1,\tau_1^i,\v)_{i,s_1,\p_0}.
\end{split}
\end{equation}

\begin{proposition}[Properties of $T$]\label{prop:T_contraction}

\begin{itemize}
    \item[(i)]  For each $\ssigma_1\in \SSigma_1$, the mapping from $\mathbb{R}^{nrM}$ to $\mathbb{R}^{nrM}$ given by $\v\mapsto T(\v,\ssigma_1)$ is a contraction mapping. In particular, for every $\ssigma_1\in \SSigma_1$, $T(\cdot,\ssigma_1)$ has a unique fixed point. 
    \item[(ii)] For each $\v\in \R^{nrM}$, the mapping from $\SSigma_1$ to $\R^{nrM}$ given by $\ssigma_1\mapsto T(\v,\ssigma_1)$ is continuous. Moreover, for each bounded subset $B\subset \R^{nrM}$, the family of functions $\{T(\v;\cdot)\}_{\v\in B}$ is equicontinuous. 
\end{itemize}

\end{proposition}

\begin{proof}[{\bf Proof of Proposition \ref{prop:T_contraction}}] 
(i) Let $\ssigma_1\in \SSigma_1$ and $\u,\v\in \mathbb{R}^{nrM}$. For each $(i,s_1,\p_0)$-coordinate, let $\tau_1^i$, $\iota_1^i\in \SSigma_1^i$ be such that 
\begin{equation*}
\begin{split}
    &T(\u,\ssigma_1)_{i,s_1,\p_0}=\Vone(\ssigma_1,\tau_1^i,\u)_{i,s_1,\p_0} \textnormal{ and }\\
    &T(\v,\ssigma_1)_{i,s_1,\p_0}=\Vone(\ssigma_1,\iota_1^i,\v)_{i,s_1,\p_0}.
\end{split}
\end{equation*}
From \eqref{def:T_vsigma} and the above equations, it follows that $-T(\u,\ssigma_1)_{i,s_1,\p_0}\leq -\Vone(\ssigma_1,\iota_1^i,\u)_{i,s_1,\p_0}$ and $-T(\v,\ssigma_1)_{i,s_1,\p_0}\leq -\Vone(\ssigma_1,\tau_1^i,\v)_{i,s_1,\p_0}$. Thus, 
\begin{equation}\label{proof:Tcontraction1}
    \begin{split}
        &[T(\u,\ssigma_1)-T(\v,\ssigma_1)]_{i,s_1,\p_0}\leq [\Vone(\ssigma_1,\tau_1^i,\u)-\Vone(\ssigma_1,\tau_1^i,\v)]_{i,s_1,\p_0} \textnormal{ and} \\
        &[T(\v,\ssigma_1)-T(\u,\ssigma_1)]_{i,s_1,\p_0}\leq [\Vone(\ssigma_1,\iota_1^i,\v)-\Vone(\ssigma_1,\iota_1^i,\u)]_{i,s_1,\p_0}
    \end{split}
\end{equation}
The combination of  \eqref{proof:Tcontraction1} and (b) in Proposition \ref{prop:propertiesV1} yields
\begin{equation}\label{proof:Tcontraction2}
    \max_{i,s_1,\p_0}|T(\u,\ssigma_1)-T(\v,\ssigma_1)|_{i,s_1,\p_0}\leq  \delta \max_{j,s_2,\p_1}|v_{j,s_2,\p_1}-u_{j,s_2,\p_1}|,
\end{equation}
where $\delta = \max_{i\in[n]}\delta_i<1$. Thus, $T(\cdot,\ssigma_1)$ is a contraction mapping. The fact that $T(\cdot,\ssigma_1)$ has a unique fixed point follows from Banach Fixed point Theorem.\\

(ii) Let $\ssigma_1,\ttau_1\in\SSigma_1$ and $\v\in \R^{nrM}$. For each $(i,s_1,\p_0)$-coordinate, let $\gamma_1^i$, $\iota_1^i\in \SSigma_1^i$ be such that 
\begin{equation*}
\begin{split}
&T(\v,\ssigma_1)_{i,s_1,\p_0}=\Vone(\ssigma_1,\gamma_1^i,\v)_{i,s_1,\p_0} \textnormal{ and }\\    &T(\v,\ttau_1)_{i,s_1,\p_0}=\Vone(\ttau_1,\iota_1^i,\v)_{i,s_1,\p_0}.
\end{split}
\end{equation*}
From \eqref{def:T_vsigma} and the above equations, it follows that $-T(\v,\ssigma_1)_{i,s_1,\p_0}\leq -\Vone(\ssigma_1,\iota_1^i,\v)_{i,s_1,\p_0}$ and $-T(\v,\ttau_1)_{i,s_1,\p_0}\leq -\Vone(\ttau_1,\gamma_1^i,\v)_{i,s_1,\p_0}$. Thus, 
\begin{equation}\label{proof:Lemma3_1}
    \begin{split}
        &[T(\v,\ssigma_1)-T(\v,\ttau_1)]_{i,s_1,\p_0}\leq [\Vone(\ssigma_1,\gamma_1^i,\v)-\Vone(\ttau_1,\gamma_1^i,\v)]_{i,s_1,\p_0} \textnormal{ and} \\
        &[T(\v,\ttau_1)-T(\v,\ssigma_1)]_{i,s_1,\p_0}\leq [\Vone(\ttau_1,\iota_1^i,\v)-\Vone(\ssigma_1,\iota_1^i,\v)]_{i,s_1,\p_0}.
    \end{split}
\end{equation}
Let $\epsilon>0$, by part (a) in Proposition \ref{prop:propertiesV1}, there exists $\theta>0$ such that for each $\kappa \in \{ \sigma_1^i, \tau_1^i \}$ if 
\begin{equation}\label{proof:Lemma3_2}
    | \ssigma_1-\ttau_1  |_{\infty}<\theta \implies |\Vone(\ssigma_1,\kappa,\v)-\Vone(\ttau_1,\kappa,\v)|_{\infty} <\epsilon,
\end{equation}
where  $| \ssigma_1 |_{\infty}$ denotes the supremum norm of $\ssigma_1\in \SSigma_1\subset \R^{n\hat{M}}$ (see \eqref{def:Sigma_ti}). From \eqref{proof:Lemma3_1} and \eqref{proof:Lemma3_2}, it follows that the mapping $\ssigma_1\mapsto T(\v,\ssigma_1)$ is continuous. 

Let $B$ be a bounded subset of $\R^{nrM}$. By \eqref{def:Sigma_ti}, the set $\SSigma_1\times \SSigma_1 \times \bar{B}$ is compact. By Proposition \ref{prop:propertiesV1}, $\Vone$ is uniformly continuous on $\SSigma_1\times \SSigma_1 \times \bar{B}$. It follows that for each $\epsilon>0$, there exists $\theta>0$ such that for each $\ssigma_1,\ttau_1$ and $\kkappa_1$ in $\SSigma_1$ and $\v\in B$, if
\begin{equation}\label{proof:Lemma3_3}
    | \ssigma_1-\ttau_1  |_{\infty}<\theta \implies |\Vone(\ssigma_1,\kkappa_1,\v)-\Vone(\ttau_1,\kkappa_1,\v)|_{\infty}<\epsilon.
\end{equation}
Replacing \eqref{proof:Lemma3_2} with \eqref{proof:Lemma3_3} shows that the family of functions $\{T(\v,\cdot)\}_{\v\in B}$ is equicontinuous.

\end{proof}

\textbf{The mapping $b$ and the correspondence $\Gamma$:} Let $\ssigma_1\in \SSigma_1$. From Part (i) in Proposition \ref{prop:T_contraction}, there exists a unique vector $b(\ssigma_1)\in\R^{nrM}$ such that $b(\ssigma_1) = T(b(\ssigma_1),\ssigma_1)$. Thus, there is a well-defined mapping $b:\SSigma_1\longrightarrow \R^{nrM}$ such that $\ssigma_1\mapsto b(\ssigma_1)$. In particular, by \eqref{def:T_vsigma}, for each $(i,s_1,\p_0)$-coordinate
\begin{equation}
    \label{def:b_sigma}
    b(\ssigma_1)_{i,s_1,\p_0} = \max_{\tau_1^i\in \SSigma_1^i}\Vone(\ssigma_1,\tau_1^i,b(\ssigma_1))_{i,s_1,\p_0}.
\end{equation}
From \eqref{def:b_sigma} and the compactness of $\SSigma_1^i$, there exists $\tilde{\ttau}_1\in \SSigma_1$ such that for each $(i,s_1,\p_0)$-coordinate,  $b(\ssigma_1)_{i,s_1,\p_0} = \Vone(\ssigma_1,\tilde{\tau}_1^i,b(\ssigma_1))_{i,s_1,\p_0}.$ The previous argument shows that for each $\ssigma_1\in \SSigma_1$, the following set is nonempty, 
\begin{equation}
    \label{def:Gamma_sigma}
    \Gamma(\ssigma_1): = \{ \ttau_1\in \SSigma_1 | b(\ssigma_1) = \Vone(\ssigma_1,\ttau_1,b(\ssigma_1)) \}. 
\end{equation}
The mapping $\Gamma$ from $\SSigma_1$ to $2^{\SSigma_1}$ is a self-correspondence on $\SSigma_1$.

\begin{proposition}[Properties of $b$ and $\Gamma$]\label{prop_b_properties}
\begin{itemize}
    \item[(i)]  $b$ is continuous; 
    \item[(ii)] $\Gamma$ is a convex- and closed-valued self-correspondence on $\SSigma_1$. Moreover, it has a closed graph.  
\end{itemize}
    
\end{proposition}
\begin{proof}[{\bf Proof of Proposition \ref{prop_b_properties}}] (i) We first show that $b(\SSigma_1)\subset \R^{nrM}$ is bounded. Let $\ssigma_1\in\SSigma_1$. By Proposition \ref{prop:T_contraction}, the definition of $b$ and the Banach Fixed Point Theorem: the sequence given by $\v_0 = 0\in  \R^{nrM}$ and $\v_n=T(\v_{n-1},\ssigma_1)$ for $n\geqq 1$ converges to $b(\ssigma_1)$. Moreover, 
\begin{equation}\label{proof:propb_Gamma1}
    \max_{i,s_1,\p_0}|b(\ssigma_1)_{i,s_1,\p_0} - 0|\leq \frac{1}{1-\delta}\max_{i,s_1,\p_0}|(\Vone)_{i,s_1,\p_0} - 0|.
\end{equation}
Note that $
    (\Vone)_{i,s_1,\p_0} = T(0;\ssigma_1)_{i,s_1,\p_0} = \max_{\tau_1^i\in\SSigma_1^i} \Vone(\ssigma_1,\tau_1^i,0)_{i,s_1,\p_0}
$ and by \eqref{V1_function},
\begin{equation}\label{proof:propb_Gamma2}
    \Vone(\ssigma_1,\tau_1^i,0)_{i,s_1,\p_0}= \sum_{\p_1\in\A^n}\tau_1^i(p_1^i|\p_0,s_1)\sigma_1^{-i}(\p_1^{-i}|\p_0,s_1)\pi^i(\p_1,s_1)
\end{equation}
Combining \eqref{proof:propb_Gamma1} with \eqref{proof:propb_Gamma2} yields, 
\begin{equation*}
    \max_{i,s_1,\p_0}|b(\ssigma_1)_{i,s_1,\p_0}|\leq \frac{1}{1-\delta} \max_{i,s_1,\p_1}|\pi^i(\p_1,s_1)|.
\end{equation*}
Proving that the set $b(\SSigma_1)$ is bounded in $\R^{nrM}$. We now show that $b$ is continuous. Let $\ssigma_1$ and $\iiota_1$ in $\SSigma_1$. We estimate the following supremum norm
\begin{equation}\label{proof:propb_Gamma3}
    \begin{split}
        &|b(\ssigma_1)-b(\iiota_1)|_{\infty} = |T(b(\ssigma_1),\ssigma_1)-T(b(\iiota_1),\iiota_1)|_{\infty}\\
        &\leq |T(b(\ssigma_1),\ssigma_1)-T(b(\iiota_1),\ssigma_1)|_{\infty} + |T(b(\iiota_1),\ssigma_1)-T(b(\iiota_1),\iiota_1)|_{\infty}.
    \end{split}
\end{equation}
From \eqref{proof:Tcontraction2} in the Proof of Proposition \ref{prop:T_contraction}, $|T(b(\ssigma_1),\ssigma_1)-T(b(\iiota_1),\ssigma_1)|_{\infty}\leq \delta |b(\ssigma_1)-b(\iiota_1)|_{\infty}$, where $\delta = \max_{i\in [n]}\delta_i$. Thus,
\begin{equation}\label{proof:propb_Gamma4}
    \begin{split}
        &|b(\ssigma_1)-b(\iiota_1)|_{\infty} \leq \frac{1}{1-\delta} |T(b(\iiota_1),\ssigma_1)-T(b(\iiota_1),\iiota_1)|_{\infty}.
    \end{split}
\end{equation}
Since $b(\SSigma_1)$ is bounded, by part (ii) in Proposition \ref{prop:T_contraction}, the family of functions $\{T(\v;\cdot)\}_{\v\in b(\SSigma_1)}$ is equicontinuous. It follows that for each $\epsilon>0$ there exists $\theta > 0$ such that for any $\ssigma_1,\iiota_1\in\SSigma_1$ and $\v\in b(\SSigma_1)$, if $|\ssigma_1-\iiota_1|_{\infty}<\theta$, then 
$$|T(\v,\ssigma_1)-T(\v,\iiota_1)|_{\infty}<\epsilon(1-\delta).$$
It follows from \eqref{proof:propb_Gamma4} that $b$ is continuous. \\

(ii) Let $\ssigma_1\in\SSigma_1$. That $\Gamma(\ssigma_1)$ is convex follows from (c) in Proposition \ref{prop:propertiesV1}, as for any $\ttau_1, \iiota_1\in \Gamma(\ssigma_1)$ and $\alpha\in[0,1]$,
\begin{equation*}
\begin{split}
     b(\ssigma_1) &= \alpha b(\ssigma_1) + (1-\alpha) b(\ssigma_1)=
     \alpha \Vone(\ssigma_1,\ttau_1,b(\ssigma_1)) + (1-\alpha) \Vone(\ssigma_1,\iiota_1,b(\ssigma_1))\\
     &=\Vone(\ssigma_1,\alpha\ttau_1+(1-\alpha)\iiota_1,b(\ssigma_1)).
\end{split}
\end{equation*}
We now show that $\Gamma(\ssigma_1)$ is closed in $\SSigma_1$: Let $(\ttau_{1,k})_{k\geqq 1}\subset \Gamma(\ssigma_1)$ be a sequence such that $\ttau_{1,k}\to \ttau_1\in \SSigma_1$ as $k\to \infty$. By definition of $\Gamma(\ssigma_1)$ and continuity of $\Vone$, 
\begin{equation*}
    b(\ssigma_1) = \Vone(\ssigma_1,\ttau_{1,k},b(\ssigma_1))\to \Vone(\ssigma_1,\ttau_{1},b(\ssigma_1)), \textnormal{ as } k\to \infty.
\end{equation*}
It follows that $\ttau_{1} \in \Gamma(\ssigma_1)$. 

We show that $\Gamma$ has a closed graph. Let $(\ssigma_{1,k})_k$ and $(\iiota_{1,k})_k$ be two sequences in $\SSigma_1$ such that $\ssigma_{1,k}\to \ssigma_1\in \SSigma_1$ and $\iiota_{1,k}\to \iiota_1\in \SSigma_1$ as $k\to \infty$. Suppose that $\iiota_{1,k}\in \Gamma(\ssigma_{1,k})$ for each $k\geqq 1$. By definition, for each $k\geqq 1$,
\begin{equation*}
    b(\ssigma_{1,k}) = \Vone(\ssigma_{1,k},\iiota_{1,k},b(\ssigma_{1,k})).
\end{equation*}
By part (i) in this proposition, $b$ is continuous, therefore $b(\ssigma_{1,k})\to b(\ssigma_{1})$ as $k\to \infty$. By Proposition \ref{prop:propertiesV1}, $\Vone$ is continuous, thus, $\Vone(\ssigma_{1,k},\iiota_{1,k},b(\ssigma_{1,k}))\to \Vone(\ssigma_{1},\iiota_{1},b(\ssigma_{1}))$ as $k\to \infty$. It follows that 
\begin{equation*}
    b(\ssigma_{1}) = \Vone(\ssigma_{1},\iiota_{1},b(\ssigma_{1})), 
\end{equation*}
and $\iiota_{1}\in \Gamma(\ssigma_1)$. 

\end{proof}


\subsubsection*{Conclusion of Theorem \ref{Thm:Fink}}

By Proposition \ref{prop_b_properties}, $\Gamma$ as given by \eqref{def:Gamma_sigma} is a convex-valued self-correspondence on $\SSigma_1$ that has a closed graph. Moreover, $\SSigma_1$ is compact and convex. By Theorem \ref{Thm:Kakutani}, there exists $\ssigma_1^*\in \SSigma_1$ such that 
$\ssigma_1^*\in \Gamma(\ssigma_1^*)$, i.e., 
\begin{equation*}
     b(\ssigma_1^*) = \Vone(\ssigma_1^*,\ssigma_1^*,b(\ssigma_1^*)).
\end{equation*}